\documentclass[a4paper,10pt]{article}
\usepackage{graphicx}
\usepackage{amsmath}
\usepackage{citesort}
\usepackage[sort&compress,numbers]{natbib}
\usepackage{amsfonts}
\usepackage{amssymb}
\usepackage{amsmath}
\usepackage{latexsym}
\usepackage{color}
\usepackage{ntheorem}
\usepackage{braket}
\usepackage[multiple]{footmisc}

\newcommand{\mnote}[1]
{\protect{\stepcounter{mnotecount}}$^{\mbox{\footnotesize
$
\bullet$\themnotecount}}$ \marginpar{
\raggedright\tiny\em
$\!\!\!\!\!\!\,\bullet$\themnotecount: #1} }

\newcounter{mnotecount}[section]

\renewcommand{\themnotecount}{\thesection.\arabic{mnotecount}}

\newcommand{\tim}[1]{\mnote{{\bf tim:}#1}}

\newtheorem{theorem}{\sc  Theorem\rm}[section]

\newtheorem{definition}[theorem]{\sc  Definition\rm}

\newtheorem{lemma}[theorem]{\sc Lemma\rm}

\newtheorem{proposition}[theorem]{\sc Proposition\rm}

\newtheorem{remark}[theorem]{\sc Remark\rm}

\newcommand{\ol}[1]{\overline{#1}{}}

\newcommand{\jlcax}[1]{}
\newcommand{\eean}{\nonumber\end{eqnarray}}

\newcommand{\kk}[1]{}

\newcommand{\beq}{\begin{equation}}

\newcommand{\FS}       
                  {F}

\newcommand{\HS} 
       {H_{\mbox{\scriptsize volume}}}

{\ptc{this should be removed in the oberwolfach version}}%

\newcommand{\eeal}[1]{\label{#1}\end{eqnarray}}
\newcommand{\bed}{\begin{deqarr}}
\newcommand{\eed}{\end{deqarr}}
\newcommand{\bedl}[1]{\begin{deqarr}\label{#1}}
\newcommand{\eedl}[2]{\arrlabel{#1}\label{#2}\end{deqarr}}

\newcommand{\bel}[1]{\begin{equation}\label{#1}}
\newcommand{\bea}{\begin{eqnarray}}
\newcommand{\bean}{\begin{eqnarray}\nonumber}
\newcommand{\beal}[1]{\begin{eqnarray}\label{#1}}
\newcommand{\eea}{\end{eqnarray}}

\newcommand{\qed}{\hfill $\Box$ \medskip}

\newcommand{\be}{\begin{equation}}
\newcommand{\eeq}{\end{equation}}
\newcommand{\ee}{\end{equation}}
\newcommand{\beqa}{\begin{eqnarray}}
\newcommand{\eeqa}{\end{eqnarray}}
\newcommand{\beqan}{\begin{eqnarray*}}
\newcommand{\eeqan}{\end{eqnarray*}}
\newcommand{\ba}{\begin{array}}
\newcommand{\ea}{\end{array}}

\newcommand{\mcM}{{\mycal M}}

\newcommand{\scri}{{\mycal I}}%
\newcommand{\Scri}{\scri}

\newcommand{\warn}[1]
{\protect{\stepcounter{mnotecount}}$^{\mbox{\footnotesize
$
\bullet$\themnotecount}}$ \marginpar{
\raggedright\tiny\em
$\!\!\!\!\!\!\,\bullet$\themnotecount: {\bf Warning:} #1} }

\newcommand{\eq}[1]{(\ref{#1})}

\newcommand{\ptc}[1]{\mnote{{\bf ptc:}#1}}

\newcommand{\beqar}{\begin{deqarr}}
\newcommand{\eeqar}{\end{deqarr}}

\newcommand{\beaa}{\begin{eqnarray*}}
\newcommand{\eeaa}{\end{eqnarray*}}

\DeclareFontFamily{OT1}{rsfs}{}
\DeclareFontShape{OT1}{rsfs}{m}{n}{ <-7> rsfs5 <7-10> rsfs7 <10-> rsfs10}{}
\DeclareMathAlphabet{\mycal}{OT1}{rsfs}{m}{n}

{\catcode `\@=11 \global\let\AddToReset=\@addtoreset}
\AddToReset{equation}{section}

{\catcode `\@=11 \global\let\AddToReset=\@addtoreset}
\AddToReset{figure}{section}

{\catcode `\@=11 \global\let\AddToReset=\@addtoreset}
\AddToReset{table}{section}

\newcommand{\lambdahere}{\lambda}%

\begin{document}

\title{Conformally covariant systems of wave equations and their equivalence to Einstein's field equations%
\thanks{Preprint UWThPh-2013-15.}
\vspace{0.5em}}
\author{
 Tim-Torben Paetz%
\thanks{E-mail:  Tim-Torben.Paetz@univie.ac.at}  \vspace{0.5em}\\  \textit{Gravitational Physics, University of Vienna}  \\ \textit{Boltzmanngasse 5, 1090 Vienna, Austria }}

\maketitle

\begin{abstract}
We derive, in $3+1$ spacetime dimensions, two alternative systems of quasi-linear wave equations, based on  Friedrich's conformal field equations. We analyse their equivalence to  Einstein's vacuum field equations when appropriate  constraint equations are satisfied by the initial data. As an application, the characteristic initial value problem for the Einstein
equations with data on past null infinity is reduced to a characteristic initial value problem for wave equations with data on an ordinary light-cone.
\end{abstract}


\noindent
\hspace{2.1em} PACs numbers: 04.20.Ex

\tableofcontents

\section{Introduction}

\subsection{Asymptotic flatness}

In general relativity there is the endeavour to characterize those spacetimes which one would regard as being ``asymptotically flat'', possibly merely in certain (null) directions.
Spacetimes which possess this property would be well-suited to describe e.g.\ purely radiative spacetimes or  isolated gravitational systems.
However, due to the absence of a non-dynamical background field this is an intricate issue in general relativity.
In~\cite{p1,p2} (see e.g.\ \cite{geroch} for an overview) R.~Penrose proposed a geometric approach to resolve this problem: The starting point is a $3+1$-dimensional
spacetime $(\tilde{\mcM\enspace}\hspace{-0.5em}, \tilde g)$, the \textit{physical spacetime}.
It then proves fruitful to introduce a so-called \textit{unphysical spacetime} $(\mcM,g)$ into which (a part of) $(\tilde{\mcM\enspace}\hspace{-0.5em}, \tilde g)$
is conformally embedded,
\begin{equation*}
 \tilde g \overset{\phi}{\mapsto} g:=\Theta^2 \tilde g\;, \quad  \tilde{\mcM\enspace}\hspace{-0.5em} \overset{\phi}{\hookrightarrow} \mcM\;,
  \quad \Theta|_{\phi( \tilde{\mcM\enspace}\hspace{-0.5em})}>0
 \;.
\end{equation*}
%
The part of $\partial \phi( \tilde{\mcM\enspace}\hspace{-0.5em})$ where the conformal factor $\Theta$ vanishes can be interpreted as representing
infinity of the original, physical spacetime, for the physical affine parameter diverges along null geodesics when approaching this part of the boundary.
The subset $\{\Theta=0\;, \, \mathrm{d}\Theta\ne 0\} \subset \partial \phi( \tilde{\mcM\enspace}\hspace{-0.5em})$ is called \textit{Scri}, denoted by\ $\scri$.
Large classes of solutions of the Einstein equations (with vanishing cosmological constant) possess a $\scri$ which forms a smooth  null hypersurface in $(\mcM,g)$, on which null geodesics in $(\mcM,g)$ acquire end-points.
 The hypersurface $\scri$ is therefore regarded as providing a representation of null infinity.

Penrose's proposal
to distinguish those spacetimes which have an ``asymptotically flat'' structure [in certain null directions] is to require that the unphysical metric tensor $g$ extends smoothly across [a part of] $\scri$.%
\footnote{One may also think of weaker requirements here.}
The idea is that such a smooth conformal extension is possible whenever the gravitational field has an appropriate ``asymptotically flat'' fall-off behaviour in these directions.

Null infinity can be split into two components,
past and future null infinity $\scri^-$ and $\scri^+$, which are generated by the past and future endpoints of null geodesics in $\mcM$, respectively.
If the spacetime is further supposed to be asymptotically flat in all spacelike directions, one may require the existence
of a  point $i^0$, representing \textit{spacelike infinity}, where all the spacelike geodesics meet. However, $i^0$ cannot be assumed to be smooth
(it cannot even assumed to be $C^1$~\cite{ashtekar}).

In this work we are particularly interested in spacetimes (and the construction thereof) which, at sufficiently early times, possess a conformal infinity which
is similar to that of Minkowski spacetime. By that we mean that (a part of)  $(\tilde{\mcM\enspace}\hspace{-0.5em}, \tilde g)$
can be conformally mapped into an unphysical spacetime, where all timelike geodesics originate from one regular
 point, which represents \textit{past timelike infinity}, denoted by $i^-$;
moreover, we assume that, at least sufficiently close to $i^-$, a regular $\scri^-$ exists and is generated by the null geodesics emanating from
$i^-$, i.e.\ forms the future null cone of $i^-$, denoted by $C_{i^-}:=\scri^- \cup \{i^-\}$.
By the term ``regular'' we mean that the conformally rescaled metric $g$, and also the rescaled Weyl tensor,
admit smooth extensions.
In fact, in $3+1$ dimensions the extendability assumption across $\scri$ on the rescaled Weyl tensor is automatically satisfied in the current setting.
At $i^-$ this assumption will be dropped in Section~\ref{alternative_system}.
Purely radiative spacetimes are expected to possess such a conformal structure \cite{F7}.

It is an important issue to understand the interplay between the geometric concept of asymptotic flatness and the
Einstein equations, and whether all relevant physical systems
are compatible with the notion of a regular conformal infinity.
There are various results indicating that this is a reasonable concept, cf.\  \cite{ac,andersson, p2,kannar,F8,ChDelay} and references given therein.
An open issue is to characterize the set of asymptotically Euclidean initial data on a spacelike hypersurface  which lead  to solutions of Einstein's field equations which are  ``null asymptotically
flat''.

Since we have a characteristic initial value problem at $C_{i^-}$ in mind, we want to avoid too many technical assumptions which might lead to a more reasonable (and rigid) notion of asymptotic flatness, asymptotic simplicity, etc.\
(cf.\ e.g.\ \cite{geroch2}).
In a nutshell, we are concerned with solutions of the vacuum Einstein equations (with vanishing cosmological constant) which admit a 
regular null cone at past timelike infinity, at least near $i^-$.

\subsection{Conformal field equations}

Due to the geometric construction outlined above, the asymptotic behaviour of the gravitational field can be analysed in terms of a local problem in a neighbourhood of $\scri$ (as well as $i^{\pm}$ and $i^0$).
However, the vacuum Einstein equations, regarded as equations for the unphysical metric $g$, are (formally) singular at conformal infinity (set $\Box_g:=\nabla^{\mu}\nabla_{\mu}$),
\begin{eqnarray}
 &&\hspace{-3em}\tilde R_{\mu\nu}[\tilde g] = \lambda\tilde g_{\mu\nu} \quad  \Longleftrightarrow
 \nonumber
\\
&&\hspace{-3em} R_{\mu\nu}[g] + 2\Theta^{-1} \nabla_{\mu}\nabla_{\nu} \Theta  + g_{\mu\nu} \big( \Theta^{-1}\Box_g\Theta
  - 3\Theta^{-2} \nabla^{\sigma}\Theta\nabla_{\sigma}\Theta\big) = \lambda \Theta^{-2} g_{\mu\nu}
  \;,
\label{Einstein_singular}
\end{eqnarray}
%
where the conformal factor $\Theta$ is assumed to be some given (smooth) function.
The system \eq{Einstein_singular} does therefore not seem to be convenient to study unphysical spacetimes $(\mcM,g)$ with $\Theta^{-2} g$ being a solution of the  Einstein equations away from conformal infinity.
Serendipitously, H.~Friedrich~\cite{F1,F2,F3} was able to extract a system, the \textit{conformal field equations}, which does remain regular even if $\Theta$ vanishes, and which is equivalent to the vacuum Einstein equations wherever $\Theta$ is non-vanishing.

In a suitable gauge the propagational part of the conformal field equations implies, in $3+1$ dimensions, a symmetric hyperbolic system, the \emph{reduced
conformal field equations}.
Thus equipped with some nice mathematical properties Friedrich's equations  provide a powerful tool to analyse the
asymptotic behaviour of those solutions of the Einstein equations which admit an appropriate conformal structure at infinity.

\subsection{Characteristic initial value problems}

The characteristic initial value problem in general relativity provides a tool to construct systematically general solutions of
Einstein's field equations.
An advantage in comparison with the spacelike Cauchy problem is that the constraint equations can be read as a hierarchical system of ODEs,
which is much more convenient to deal with.
In fact, one may think of  several different types of (asymptotic) characteristic initial value problems, which we want to recall briefly.

One possibility is to take two transversally intersecting null hypersurfaces as initial surface.
This problem was studied by Rendall~\cite{rendall} who established well-posedness results for quasi-linear wave equations as well as for symmetric hyperbolic systems in a neighbourhood of the cross-section of these hypersurfaces.
Using a harmonic reduction of the Einstein equations he then applied his results to prove well-posedness for the Einstein equations.

Another approach is to prescribe data on a light-cone.
There is a well-posedness result for quasi-linear wave equations near the tip of a cone available which is due to Cagnac~\cite{cagnac2} and Dossa~\cite{dossa}.
A crucial assumption in their proof is that the initial data are restrictions to the light-cone of smooth%
\footnote{There is a version for finite differentiability, but here we restrict attention to the smooth case.}
spacetime fields.
Well-posedness of the Einstein equations was investigated in a series of papers~\cite{CCM3,CCM4,CCM2} by
Choquet-Bruhat, Chru\'sciel and Mart\'in-Garc\'ia, and by Chru\'sciel~\cite{C1}.
The authors impose a wave-map gauge condition to obtain a system of wave equations to which the Cagnac-Dossa theorem is applied.
A main difficulty, in the most comprehensive case treated in~\cite{C1}, is to make sure that the Cagnac-Dossa theorem
is indeed applicable. For that one needs to make sure that the initial data for the reduced Einstein equations, which are constructed from suitable free data as solution of the constraint equations, can be extended to smooth spacetime fields.
One then ends up with the result that these free data determine a unique solution (up to isometries) in
some neighbourhood of the tip of the cone $C_O$, intersected  with $J^+(C_O)$.

A third important case arises when the initial surface is, again, given by two transversally intersecting null hypersurfaces,
but now in the unphysical spacetime and with one of the hypersurfaces belonging to $\scri$.
This issue was treated by Friedrich~\cite{F6}, who proved well-posedness for analytic data, and by K\'ann\'ar~\cite{kannar}, who extended Friedrich's result to the smooth case. The basic idea for the proof is to exploit the fact that the reduced conformal field equations form a symmetric hyperbolic system to which  Rendall's local existence result is applicable.

The case we have in mind is when the initial surface is given in the unphysical spacetime by the light-cone $C_{i^-}$
emanating from past timelike infinity~$i^-$.
In order to construct systematically solutions of Einstein's field equations which are compatible with Penrose's notion of asymptotic flatness
and a regular $i^-$,
one would like to prescribe data
on $C_{i^-}$ and predict existence of a solution of Einstein's equations
off  $C_{i^-}$ by solving an appropriate initial value problem.
One way to establish well-posedness near the tip of the cone is to mimic the analysis in~\cite{CCM2,C1}.
To do that, one needs a  system of wave equations
which, when supplemented by an appropriate set of constrain equations, is equivalent to the vacuum Einstein equations wherever $\Theta$ is  non-vanishing and which remains regular when
$\Theta$ vanishes.
Based on a conformal system of equations  due to Choquet-Bruhat and Novello~\cite{novello}, such a regular system of wave equations was employed by Dossa~\cite{dossa2} who states a well-posedness result for suitable initial data for which, however, it is not clear how they can be constructed, nor to what extent his system of wave equations is equivalent to the Einstein equations.

The purpose of this paper is to derive two such systems of wave equations in $3+1$-spacetime dimensions, which we will call \textit{conformal wave equations},
and prove equivalence to Friedrich's conformal system
for solutions of the characteristic initial value problem with initial surface $C_{i^-}$ which satisfy certain constraint equations on $C_{i^-}$.
Our first system will use the same set of unknowns as Friedrich's \textit{metric conformal field equations}~\cite{F3}, while the second system will employ the Weyl and the Cotton tensor rather than the rescaled Weyl tensor (and might be advantageous in view of the construction of solutions with a rescaled Weyl tensor which diverges at $i^-$).
The construction of initial data to which the Cagnac-Dossa theorem is applicable, and thus a well-posedness proof of the Cauchy problem with data on the $C_{i^-}$-cone, is accomplished in \cite{CFP,F9}.

Apart from the application to tackle the characteristic initial value problem with data on $C_{i^-}$,
a regular system of wave equations
 might be interesting for numerics, as well~\cite{kreiss}.

\subsection{Structure of the paper}
In Section~\ref{conformal_field_equations} we recall the metric conformal field equations and address the gauge freedom inherent to them.
In Section~\ref{cwe} we derive the first system of conformal wave equations, \eq{cwe1}-\eq{cwe5}, and prove equivalence
to the conformal field equations and consistency with the gauge condition under the assumption that certain relations hold initially.
In Section~\ref{sec_constraint_equations} we derive the constraint equations induced by the conformal field equations on $C_{i^-}$ in adapted coordinates and imposing a generalized wave-map gauge condition.
We then focus on the case of a light-cone with vertex at past timelike infinity to verify
in Section~\ref{sec_applicability} that the hypotheses needed for the equivalence theorem of Section~\ref{cwe} are indeed satisfied, supposing that the initial data fulfill the constraint equations \eq{constraint_g}-\eq{relation_lambda_omega}.
Our main result, Theorem~\ref{main_result}, states that a solution of the characteristic initial value problem for the conformal wave equations, with initial data on $C_{i^-}$ which have been constructed as solutions of the constraint equations, is also a solution of the conformal field equations in wave-map gauge and vice versa.
In Section~\ref{alternative_system} we then derive an alternative system of  wave equations, \eq{Wweylwave1}-\eq{Wweylwave6}, and study equivalence to the conformal field equations, supposing that certain constraint equations, namely \eq{W_constrainteqn_g}-\eq{W_boundary_values2}, are satisfied, cf.\ Theorem~\ref{main_result2}.  In Section~\ref{sec_conclusions} we   briefly compare both systems of wave equations and give a short summary.
We conclude the article by reviewing some basic properties of cone-smooth functions, which are utilized to prove a lemma stated  in Section~\ref{conformal_field_equations}.

Throughout this work we restrict attention to $3+1$ dimensions, cf.\  footnote~\ref{foot_rem}.

\section{Friedrich's conformal field equations and gauge freedom}
\label{conformal_field_equations}

\subsection{Metric conformal field equations (MCFE)}


As indicated above, the vacuum Einstein equations themselves do not provide a nice evolution system near infinity and are therefore not suitable to tackle the issue at hand,
namely to analyse existence of a solution to the future of $C_{i^-}$.
Nonetheless,  they permit a representation
which does not contain factors of $\Theta^{-1}$ and which is regular everywhere \cite{F1,F2,F3}.
Due to this property the Einstein equations are called \textit{conformally regular}.

The curvature of a spacetime is measured by the Riemann curvature tensor $R_{\mu\nu\sigma}{}^{\rho}$, which can be decomposed into the trace-free \textit{Weyl tensor} $W_{\mu\nu\sigma}{}^{\rho}$,
invariant under conformal transformations, and a term which involves the \textit{Schouten tensor} $L_{\mu\nu}$,
\begin{equation}
  R_{\mu\nu\sigma}{}^{\rho}= W_{\mu\nu\sigma}{}^{\rho} + 2(g_{\sigma[\mu} L_{\nu]}{}^{\rho} -\delta_{[\mu}{}^{\rho}L_{\nu]\sigma} )
 \;.
\end{equation}
The Schouten tensor is defined in terms of the Ricci tensor $R_{\mu\nu}$,
\begin{equation}
  L_{\mu\nu} := \frac{1}{2}R_{\mu\nu} - \frac{1}{12} R g_{\mu\nu}
 \;.
\label{dfn_Schouten}
\end{equation}
 The Weyl tensor is usually considered to represent the radiation part of the gravitational field.
Let us further define the \textit{rescaled Weyl tensor}
\begin{equation}
 d_{\mu\nu\sigma}{}^{\rho} := \Theta^{-1} W_{\mu\nu\sigma}{}^{\rho}
 \label{dfn_rescaled_weyl}
 \;,
\end{equation}
as well as the scalar function ($\Box_g \equiv \nabla^{\mu}\nabla_{\mu}$)
\begin{equation}
 s := \frac{1}{4}\Box_g \Theta + \frac{1}{24}  R\Theta
 \;.
\end{equation}

There exist different versions of the conformal field equations, depending on which fields are regarded as unknowns. Here we present the \textit{metric conformal field
equations (MCFE)}~\cite{F3} which read in $3+1$ spacetime dimensions
\begin{eqnarray}
 && \nabla_{\rho} d_{\mu\nu\sigma}{}^{\rho} =0\;,
 \label{conf1}
\\
 && \nabla_{\mu} L_{\nu\sigma} - \nabla_{\nu}L_{\mu\sigma} = \nabla_{\rho}\Theta \, d_{\nu\mu\sigma}{}^{\rho}\;,
 \label{conf2}
\\
 && \nabla_{\mu}\nabla_{\nu}\Theta = -\Theta L_{\mu\nu} + s g_{\mu\nu}\;,
 \label{conf3}
\\
 && \nabla_{\mu} s = -L_{\mu\nu}\nabla^{\nu}\Theta\;,
 \label{conf4}
\\
 && 2\Theta s - \nabla_{\mu}\Theta\nabla^{\mu}\Theta = \lambda/3  \;,
 \label{conf5}
\\
 && R_{\mu\nu\sigma}{}^{\kappa}[ g] = \Theta d_{\mu\nu\sigma}{}^{\kappa} + 2\left(g_{\sigma[\mu} L_{\nu]}{}^{\kappa}  - \delta_{[\mu}{}^{\kappa}L_{\nu]\sigma} \right)
 \label{conf6}
\;.
\end{eqnarray}
The unknowns are $g_{\mu\nu}$, $\Theta$, $s$, $ L_{\mu\nu}$ and $d_{\mu\nu\sigma}{}^{\rho}$.

Friedrich has shown that the MCFE are equivalent to the vacuum Einstein equations,
\begin{equation*}
 \tilde R_{\mu\nu}[\tilde g] \, = \, \lambda \tilde g_{\mu\nu}\;, \quad \tilde g_{\mu\nu}\, =\, \Theta^{-2} g_{\mu\nu}\;,
\end{equation*}
in the region where $\Theta$ is non-vanishing. 
They give rise to a complicated and highly overdetermined PDE-system.
It turns out that (\ref{conf5}) is a consequence of \eq{conf3} and \eq{conf4} if it is known to hold at just one point (e.g.\ by an appropriate choice of the initial data).
Moreover, Friedrich  has separated constraint and evolution equations from the conformal field equations by working in a spin frame  \cite{F1,F2} .
In Sections~\ref{ss31VIII11},  \ref{constraints_generalizedw} and \ref{kappa0_wavemap}
we shall do the same (if the initial surface is $C_{i^-}$) in a coordinate frame and by imposing a generalized wave-map gauge condition.

A specific property in the $3+1$-dimensional case is that the contracted Bianchi identity is equivalent to the Bianchi identity.
That is the reason why (\ref{conf1}) implies hyperbolic equations;
in higher dimensions this is no longer true~\cite{F3}. The conformal field equations provide a nice, i.e.\  symmetric hyperbolic,  evolution system only in $3+1$ dimensions.

Penrose proposed to distinguish asymptotically flat spacetimes by requiring
the unphysical  metric $g$ to be smoothly extendable across $\scri$.
The Weyl tensor of $g$ is known to vanish on $\scri$~\cite{p2}.
Since by definition $\mathrm{d}\Theta|_{\scri}\ne 0$ the rescaled Weyl tensor can be smoothly continued across $\scri$.
However, there seems to be no reason why
the same should be possible at $i^-$ where $\mathrm{d}\Theta=0$.
When dealing with the MCFE, where  the rescaled Weyl tensor is one of the unknowns, it is convenient to confine attention to the class of solutions with
a regular $i^-$ in the sense that both $g_{\mu\nu}$ and $d_{\mu\nu\sigma}{}^{\rho}$ are smoothly extendable
across $i^-$ (cf.\ Section~\ref{alternative_system}
where this additional assumption is dropped).

\subsection{Gauge freedom and conformal covariance inherent to the MCFE}

The gauge freedom contained in the MCFE comes from the freedom to choose  coordinates supplemented by the freedom
to choose the conformal factor $\Theta$ relating the physical and the unphysical spacetime.
Since $\Theta$ is regarded as an unknown rather than a gauge function, it remains to identify another function which reflects this gauge
freedom. The most convenient choice is the Ricci scalar~$R$:

Let us assume we have been given a smooth solution $(g_{\mu\nu}, \Theta, s, L_{\mu\nu}, d_{\mu\nu\sigma}{}^{\rho})$ of the MCFE.
Then we can compute $R$. For a conformal rescaling $g\mapsto \phi^{2}g$ for some $\phi> 0$,
the Ricci scalars $R$ and $R^*$ of $g$ and $\phi^{2}g$, respectively, are related via
\begin{equation}
 \phi R - \phi^{3} R^* = 6\Box_g\phi
 \;.
  \label{R_R*}
\end{equation}
Now, let us prescribe $R^*$ and read \eq{R_R*} as an equation for $\phi$.
If we think of a characteristic initial value problem with data on a light-cone $C_O$ (including the $C_{i^-}$-case) we are free to prescribe some $\mathring \phi >0$ on $C_O$.%
\footnote{The positivity of $\phi$ at the vertex guarantees any solution of \eq{R_R*} to be positive sufficiently close to the vertex and thereby the positivity of $\Theta^*$ (in the $C_{i^-}$-case just off the cone).
}%
\footnote{Since we are mainly interested in this case, we focus on an initial surface which is a cone. However, an analogous result can be obtained for two transversally intersecting null hypersurfaces.}
Supposing that $\mathring \phi$ is the restriction to $C_O$ of a smooth spacetime function, the Cagnac-Dossa theorem~\cite{cagnac2,dossa} tells us
that there is a solution $\phi>0$ with $\phi|_{C_O}=\mathring \phi$ in some neighbourhood of the tip of the cone.
Due to the \textit{conformal covariance} of the conformal field equations, the conformally rescaled fields
\begin{eqnarray}
 g^* &=& \phi^{2}g\;,
 \label{conformal_covariance1}
 \\
 \Theta^* &=& \phi\,\Theta\;,
 \\
 s^* &=& \frac{1}{4}\Box_{g^*}\Theta^* + \frac{1}{24}R^*\Theta^* \;,
   \\
   L^*_{\mu\nu} &=& \frac{1}{2} R^*_{\mu\nu}[g^*] - \frac{1}{12}R^* g^*_{\mu\nu} \;,
   \\
   d^*_{\mu\nu\sigma}{}^{\rho} &=& \phi^{-1} d_{\mu\nu\sigma}{}^{\rho}\;,
    \label{conformal_covariance5}
  \end{eqnarray}
provide another solution of the MCFE with Ricci scalar $R^*$ which corresponds to the same physical solution $\tilde g_{\mu\nu}$.
These considerations show that if we treat the conformal factor $\Theta$ as unknown, determined by the MCFE,
the curvature scalar $R$ of the unphysical spacetime can be arranged to take any preassigned form.
The function $R$ can therefore be regarded as a \textit{conformal gauge source function} which can be chosen arbitrarily.

There remains the freedom to prescribe $\mathring \phi$ on $C_O$.
On an ordinary cone with nowhere vanishing $\Theta$ this freedom can be employed to prescribe the initial data for the conformal factor, $\Theta|_{C_O}$ (it clearly needs to be the restriction to $C_O$ of a smooth spacetime function).
In this work we are particularly interested in the case where the vertex of the cone is located at past timelike infinity $i^-$, where, by definition, $\Theta =0$
(note that this requires to take $\lambda=0$).
Then the gauge freedom to choose $\mathring \phi$ can be employed to prescribe the function $s$ on $C_{i^-}$.
To see that, let us assume we have been given a smooth solution $(g_{\mu\nu}, \Theta, s, L_{\mu\nu}, d_{\mu\nu\sigma}{}^{\rho})$ of the MCFE
to the future of $C_{i^-}$, at least in some neighbourhood of $i^-$, by which we also mean that the solution admits a smooth extension through $C_{i^-}$.
(When $\Theta$ vanishes e.g.\ on one of two transversally intersecting null hypersurfaces one might put forward a similar argument.)
In particular the function $s$ is smooth. According to \eq{conf5} (with $\lambda=0$), it can be written away from $C_{i^-}$ as
%
\[ s= \frac{1}{2}\Theta^{-1} \nabla_{\mu}\Theta\nabla^{\mu}\Theta\;, \]
with the right-hand side smoothly extendable across $C_{i^-}$.
Under the conformal rescaling
\begin{eqnarray}
\Theta \mapsto \Theta^*:=\phi\,\Theta\;, \quad  g_{\mu\nu} \mapsto g^*_{\mu\nu}:=\phi^{2}g_{\mu\nu}\;, \quad \phi>0\;,
\label{rescaling}
\end{eqnarray}
the function $s$ becomes
\begin{equation}
 s^* = \phi^{-1}\Big( \frac{1}{2}\Theta \phi^{-2}\nabla^{\mu}\phi \nabla_{\mu}\phi + \phi^{-1} \nabla^{\mu}\Theta \nabla_{\mu}\phi + s\Big)\;.
\label{trafo_s_s'}
\end{equation}
Evaluation of this expression on $C_{i^-}$ yields
\begin{eqnarray}
 \overline{\nabla^{\mu}\Theta\nabla_{\mu} \phi +  \phi\,  s - \phi^{2} s^*}=0
  \;.
 \label{ODE_omega_rho_pre}
\end{eqnarray}
Here and henceforth we use an overbar to denote the restriction of a spacetime object to the initial surface.
Note that $\overline{\nabla^{\mu}\Theta}$ is tangent to $\Scri$, so \eq{ODE_omega_rho_pre} does not involve transverse derivatives of $\phi$ on $\scri$.
Let us prescribe $\overline s^*$ (as a matter of course it needs to be the restriction of a smooth spacetime function) and assume for the moment that some positive solution of
\eq{ODE_omega_rho_pre} exists,%
\footnote{In case of a negative $\overline s^*$, the gauge transformation would change the sign of $\Theta$.}
which we denote by $\mathring \phi$.
We take $\mathring \phi$  as initial datum for \eq{R_R*}.
We would like to have a $\mathring \phi$ which is the restriction to $C_{i^-}$ of a smooth spacetime function, so that we can apply the Cagnac-Dossa theorem, which  would supply us with a  function $\phi$ solving \eq{R_R*} and satisfying $\phi|_{C^-}=\mathring\phi$. Via the conformal rescaling \eq{conformal_covariance1}-\eq{conformal_covariance5}
we then would be led to a new solution of the MCFE with preassigned functions $R^*$ and $\overline s^*$ which represents the same physical solution
we started with.

The crucial point, which remains to be checked, is whether a solution of~\eq{ODE_omega_rho_pre} exists with the desired properties.
The following lemma, which is proven in Appendix~\ref{cone_smoothness}, shows that
this is indeed the case
(cf.\ \cite[Appendix~A]{CFP} where an alternative proof is given).
%
\begin{lemma}
\label{extension_s}
Consider any smooth solution of the MCFE in $3+1$  dimensions
in some neighbourhood $\mycal U$ to the future of $i^-$, smoothly extendable through $C_{i^-}$, which satisfies
 \begin{equation}
 s|_{i^-} \ne 0\;.
\end{equation}
 Let $\overline s^*$ be the restriction of a smooth spacetime function on $\mycal U\cap \partial J^+(i^-)$ with $\overline s^*|_{i^-}\ne 0$
and $\lim_{r\rightarrow 0}\partial_r\overline s^* = 0$.%
 \footnote{$r$ is a suitable (e.g.\ an affine) parameter along the null geodesics emanating from $i^-$, see Section~\ref{adapted_null_coord} and Appendix~\ref{cone_smoothness} for more details.}
 Then \eq{ODE_omega_rho_pre} is a Fuchsian
ODE and for every solution $\mathring \phi$  (note that the solution set is non-empty) it holds that
\begin{equation}
\mathrm{sign}(\mathring\phi|_{i^-} ) =   \mathrm{sign}(s|_{i^-}) \mathrm{sign}(s^*|_{i^-})
  \;,
\end{equation}
and $\mathring \phi$ is the restriction to $C_{i^-}$ of a smooth spacetime function.
In particular, if $\mathrm{sign}(s|_{i^-})= \mathrm{sign}(s^*|_{i^-})$ the function $\mathring\phi$ will be positive sufficiently close to $i^-$.
\end{lemma}

\begin{remark}
 Note that solutions with $s|_{i^-}=0$ would satisfy $\mathrm d \Theta =0$ on $\scri^-$, which is why the corresponding class of solutions is not of physical interest.
\end{remark}

To sum it up, due the conformal covariance of the MCFE the functions $R$ and $s|_{C_{i^-}}$ can and will be regarded as gauge source functions.

\section{Conformal wave equations (CWE)}
\label{cwe}
\subsection{Derivation of the conformal wave equations}
 \label{ss31VIII11}

In this section we  derive a system of wave equations from the MCFE (\ref{conf1})-(\ref{conf6}).
Recall that the unknowns are $g_{\mu\nu}$, $\Theta$, $s$, $L_{\mu\nu}$ and $d_{\mu\nu\sigma}{}^{\rho}$,
while the Ricci scalar $ R$ (and the function $\overline s$ or $\overline \Theta$, respectively, depending on the characteristic initial surface) are considered as
gauge functions.
The cosmological constant $\lambda$ is allowed to be non-vanishing in this section.

\subsubsection*{Derivation of an appropriate second-order system}

From (\ref{conf1}) and (\ref{conf2}) we obtain (with $\Box_{ g}\equiv \nabla^{\mu}\nabla_{\mu}$)
\begin{eqnarray*}
\Box_{ g}  L_{\mu\nu} -  R_{\mu\kappa}  L_{\nu}{}^{\kappa} - R_{\alpha\mu\nu}{}^{\kappa} L_{\kappa}{}^{\alpha} - \nabla_{\mu}\nabla_{\alpha} L_{\nu}{}^{\alpha}
 =d_{\mu}{}^{\alpha}{}_{\nu}{}^{\rho} \nabla_{\alpha} \nabla_{\rho}\Theta
 \;.
\end{eqnarray*}
Using the definition \eq{dfn_Schouten} of the Schouten tensor,
together with the contracted Bianchi identity, 
we find
\begin{equation}
 \nabla_{\mu} L_{\nu}{}^{\mu} = \frac{1}{6}\nabla_{\nu} R
 \;,\ \label{BianchiL}
\end{equation}
and thus
\begin{eqnarray*}
 \Box_{ g}  L_{\mu\nu} -  R_{\mu\kappa}  L_{\nu}{}^{\kappa} - R_{\alpha\mu\nu}{}^{\kappa} L_{\kappa}{}^{\alpha}
 -\frac{1}{6}\nabla_{\mu}\nabla_{\nu} R &=& d_{\mu}{}^{\alpha}{}_{\nu}{}^{\rho} \nabla_{\alpha} \nabla_{\rho}\Theta
 \;.
\end{eqnarray*}
We combine the right-hand side with \eq{conf3}, and employ  \eq{dfn_rescaled_weyl} as well as  (\ref{conf6})  to transform the third term on the left-hand side to end up with a wave equation for the Schouten tensor (suppose for the time being that $g_{\mu\nu}$ is given, cf.\ below),
\begin{eqnarray}
  \Box_{ g} L_{\mu\nu} - 4 L_{\mu\kappa} L_{\nu}{}^{\kappa} +  g_{\mu\nu}| L|^2
  + 2\Theta d_{\mu\alpha\nu}{}^{\rho}  L_{\rho}{}^{\alpha}
 = \frac{1}{6}\nabla_{\mu}\nabla_{\nu} R
  \label{waveeqn_L} \label{vconf1}
 \;,
\end{eqnarray}
where we have set
\begin{equation*}
 |L|^2 := L_{\mu}{}^{\nu}L_{\nu}{}^{\mu}
  \;.
\end{equation*}

Next, let us consider the function $ s$. From \eq{conf4}, \eq{BianchiL} and \eq{conf3} we deduce the wave equation
\begin{eqnarray}
 \Box_{ g}  s &=& -\nabla_{\mu} L_{\nu}{}^{\mu}\nabla^{\nu}\Theta - L^{\mu\nu}\nabla_{\mu}\nabla_{\nu}\Theta
 \nonumber
\\
    &=&  \Theta| L|^2  -\frac{1}{6}\nabla_{\nu} R\,\nabla^{\nu}\Theta  - \frac{1}{6} s  R
 \;.
 \label{vconf2}
\end{eqnarray}
The definition of $ s$ provides a wave equation for the conformal factor,
\begin{eqnarray}
 \Box_{ g}\Theta &=& 4 s-\frac{1}{6} \Theta  R
 \;.
 \label{vconf3}
\end{eqnarray}

To obtain a wave equation for the rescaled Weyl tensor $d_{\mu\nu\sigma}{}^{\rho}$ \textit{in 3+1 dimensions} one proceeds as follows:
Due to its algebraic properties
the rescaled Weyl tensor satisfies the relation
\begin{equation*}
 \epsilon_{\mu\nu}{}{}^{\alpha\beta}d_{\alpha\beta\lambda\rho} = \epsilon_{\lambda\rho}{}^{\alpha\beta}d_{\mu\nu\alpha\beta}
 \;,
\end{equation*}
where $\epsilon_{\mu\nu\sigma\rho}$ denotes the totally antisymmetric tensor.
We conclude that (cf.~\cite{pr1})
\begin{eqnarray}
 \nabla_{[\lambda}d_{\mu\nu]\sigma\rho} &=& -\frac{1}{6} \epsilon_{\lambda\mu\nu\kappa} \epsilon^{\alpha\beta\gamma\kappa} \nabla_{\alpha}d_{\beta\gamma\sigma\rho}
 \,=\, \frac{1}{6} \epsilon_{\lambda\mu\nu}{}^{\kappa} \epsilon_{\sigma\rho}{}^{\beta\gamma} \nabla_{\alpha}d_{\beta\gamma\kappa}{}^{\alpha}
  \;.
\label{d_bianchi_divergence}
\end{eqnarray}
This equation  implies the equivalence%
\footnote{\label{foot_rem}
We remark that this equivalence holds only in 4 dimensions.
Any attempt to derive a wave equation for $d_{\mu\nu\sigma}{}^{\rho}$ in dimension $d\geq 5$ seems to lead to singular
terms.
Also, if one uses a different set of variables, like e.g.\ Cotton and Weyl tensor instead of $d_{\mu\nu\sigma}{}^{\rho}$, cf.\ Section~\ref{alternative_system}, the derivation of a regular system of wave equations
seems to be possible merely in the 4-dimensional case.
This is in line with the observation that the conformal field equations provide a good evolution system only in 4 dimensions.
}
\begin{eqnarray*}
 \nabla_{\rho}d_{\mu\nu\sigma}{}^{\rho}=0 \quad \Longleftrightarrow \quad \nabla_{[\lambda}d_{\mu\nu]\sigma\rho}=0
  \;.
\end{eqnarray*}
Equation \eq{conf1} can therefore be replaced by
\begin{eqnarray}
 \nabla_{[\lambda}d_{\mu\nu]\sigma\rho} =0
  \;.
 \label{property_d_tensor}
\end{eqnarray}
Applying $\nabla^{\lambda}$ and commuting the covariant derivatives yields with \eq{conf1}
\begin{eqnarray*}
   \Box_g d_{\mu\nu\sigma\rho}
  + 2R_{\kappa\mu\nu}{}^{\alpha} d_{\sigma\rho\alpha}{}^{\kappa} + 2R_{\alpha[\mu} d_{\nu]}{}^{\alpha}{}_{\sigma\rho}
   + 2R_{\kappa[\mu|\sigma|}{}^{\alpha} d_{\nu]}{}^{\kappa}{}_{\alpha\rho}
   - 2 R_{\alpha\rho\kappa[\mu} d_{\nu]}{}^{\kappa}{}_{\sigma}{}^{\alpha}
   =0
   \;.
\end{eqnarray*}
With \eq{conf6} we end up with a wave equation for the rescaled Weyl tensor,
\begin{eqnarray}
  \Box_g d_{\mu\nu\sigma\rho}
  -\Theta d_{\mu\nu\kappa}{}^{\alpha}d_{\sigma\rho\alpha}{}^{\kappa}
   + 4\Theta d_{\sigma\kappa[\mu} {}^{\alpha}d_{\nu]\alpha\rho}{}^{\kappa}
 + 2 g_{\sigma[\mu}d_{\nu]\alpha\rho\kappa} L^{\alpha\kappa} &&
 \nonumber
 \\
  - 2g_{\rho[\mu} d_{\nu]\alpha\sigma\kappa} L^{\alpha\kappa}
  + 2d_{\mu\nu\kappa[\sigma} L_{\rho]}{}^{\kappa}
    + 2 d_{\sigma\rho\kappa[\mu} L_{\nu]}{}^{\kappa}
 -\frac{1}{3}R d_{\mu\nu\sigma\rho}  & =&0
   \label{vconf4_pre}
   \;.
\end{eqnarray}
It turns out that this equation does not take its simplest form yet.
To see this let us exploit \eq{property_d_tensor} again.
Invoking the Bianchi identity and \eq{conf2} we find
\begin{eqnarray*}
 0&=& \Theta\nabla_{[\lambda}d_{\mu\nu]\sigma\rho}\,=\, \nabla_{[\lambda}W_{\mu\nu]\sigma\rho} -(\nabla_{[\lambda}\Theta) d_{\mu\nu]\sigma\rho}
 \\
  &=& \frac{2}{3}\big( g_{\sigma\nu}\nabla_{[\lambda}L_{\mu]\rho} +g_{\mu\rho}\nabla_{[\lambda}L_{\nu]\sigma}
   + g_{\sigma\mu}\nabla_{[\nu}L_{\lambda]\rho}  + g_{\lambda\rho}\nabla_{[\nu}L_{\mu]\sigma}
  \\
   &&  + g_{\sigma\lambda}\nabla_{[\mu}L_{\nu]\rho} + g_{\nu\rho}\nabla_{[\mu}L_{\lambda]\sigma} \big)
 -(\nabla_{[\lambda}\Theta )d_{\mu\nu]\sigma\rho}
   \\
    &=& g_{\rho[\lambda} d_{\mu\nu]\sigma}{}^{\alpha} \nabla_{\alpha}\Theta - g_{\sigma[\lambda} d_{\mu\nu]\rho}{}^{\alpha} \nabla_{\alpha}\Theta
 -(\nabla_{[\lambda}\Theta )d_{\mu\nu]\sigma\rho}
 \;.
\end{eqnarray*}
Applying $\nabla^{\lambda}$ and using \eq{property_d_tensor}, \eq{conf3} and \eq{vconf3} we are led to
\begin{eqnarray*}
 0 &=& 3\nabla^{\lambda}(g_{\rho[\lambda} d_{\mu\nu]\sigma}{}^{\alpha} \nabla_{\alpha}\Theta - g_{\sigma[\lambda} d_{\mu\nu]\rho}{}^{\alpha} \nabla_{\alpha}\Theta
 -\nabla_{[\lambda}\Theta \,d_{\mu\nu]\sigma\rho})
 \\
&=&
2d_{\mu\nu[\sigma}{}^{\alpha}  \nabla_{\rho]}\nabla_{\alpha}\Theta
     + 2g_{\rho[\mu} d_{\nu]\lambda\sigma}{}^{\alpha}  \nabla^{\lambda}\nabla_{\alpha}\Theta
  - 2g_{\sigma[\mu} d_{\nu]\lambda\rho}{}^{\alpha}\nabla^{\lambda} \nabla_{\alpha}\Theta
  \\
  && -\Box \Theta \,d_{\mu\nu\sigma\rho}
   -\nabla^{\lambda}\nabla_{\nu}\Theta \,d_{\lambda\mu\sigma\rho}    -\nabla^{\lambda}\nabla_{\mu}\Theta \,d_{\nu\lambda\sigma\rho}
   \\
&=&
     2\Theta g_{\sigma[\mu} d_{\nu]\lambda\rho}{}^{\alpha}L_{\alpha}{}^{\lambda}  -2\Theta g_{\rho[\mu} d_{\nu]\lambda\sigma}{}^{\alpha} L_{\alpha}{}^{\lambda}
      +  2\Theta d_{\mu\nu\alpha[\sigma}L_{\rho]}{}^ { \alpha}
  \\
  &&  +2\Theta d_{\sigma\rho\alpha[\mu} L_{\nu]}{}^{ \alpha}
 +\frac{1}{6}\Theta R d_{\mu\nu\sigma\rho}
 \;.
\end{eqnarray*}
This relation simplifies \eq{vconf4_pre} significantly,
\begin{eqnarray}
  &&\Box_g d_{\mu\nu\sigma\rho}
  -\Theta d_{\mu\nu\kappa}{}^{\alpha}d_{\sigma\rho\alpha}{}^{\kappa}
   + 4\Theta d_{\sigma\kappa[\mu} {}^{\alpha}d_{\nu]\alpha\rho}{}^{\kappa}
  -\frac{1}{2}R d_{\mu\nu\sigma\rho}   =0
   \label{vconf4}
   \;.
\end{eqnarray}

We have found a system of wave equations (\ref{vconf1})-(\ref{vconf3}) and  (\ref{vconf4})
for the fields $ L_{\mu\nu}$, $s$, $\Theta$ and $d_{\mu\nu\sigma}{}^{\rho}$, assuming that $g_{\mu\nu}$ is given.
Now, we drop this assumption, so first of all the system needs to be complemented by an equation for the metric tensor. Taking the trace of \eq{conf6} yields
\begin{equation}
R_{\mu\nu}[g] = 2L_{\mu\nu} + \frac{1}{6} R g_{\mu\nu}
 \label{ricci_schouten}
 \;.
\end{equation}
However, the equations
(\ref{vconf1})-(\ref{vconf3}) and (\ref{vconf4})-\eq{ricci_schouten} do
 \textit{not} form a system of wave equations yet: Equation \eq{ricci_schouten} is not a wave equation due to
the fact that the principal part of the Ricci tensor is not a d'Alembert operator.
Moreover, the  principal part of the wave-operator $\Box_g$ is not a  d'Alembert operator when acting on tensors of valence $\geq 1$ and when the metric tensor is part of the unknowns, for the corresponding expression contains second-order derivatives of the metric due to which the principal part is not $g^{\mu\nu}\partial_{\mu}\partial_{\nu}$ anymore.
Consequently  \eq{waveeqn_L} and \eq{vconf4} are no wave equations, as well.

We need to impose an appropriate gauge condition to transform these equations into wave equations, which is accomplished subsequently.

\subsubsection*{Generalized wave-map gauge}
\label{reduced_cwe}

Let us introduce the so-called \textit{$\hat g$-generalized wave-map gauge} (cf.\ \cite{F4, F5,CCM2}),
where $\hat g_{\mu\nu}$ denotes some \textit{target metric}.
For that we define the \textit{wave-gauge vector}
\begin{equation*}
 H^{\sigma} := g^{\alpha\beta}(\Gamma^ {\sigma}_{\alpha\beta} -\hat\Gamma^{\sigma}_{\alpha\beta} ) -  W^{\sigma}
 \;.
\end{equation*}
Herein $\hat\Gamma^{\sigma}_{\alpha\beta}$ are the Christoffel symbols of  $\hat g_{\mu\nu}$.
Moreover,
\begin{equation*}
 W^{\sigma} = W^{\sigma}(x^{\mu}, g_{\mu\nu},  s,\Theta, L_{\mu\nu}, d_{\mu\nu\sigma}{}^{\rho})
\end{equation*}
is an arbitrary vector field, which is allowed to depend upon the coordinates, and possibly upon $g_{\mu\nu}$ as well as all the other fields which appear in the MCFE,%
\footnote{I am grateful to L. Andersson for pointing that out. However, in view of the constraint equations we shall consider later on for convenience merely those $W^{\sigma}$'s which depend just on the coordinates.}
but not upon derivatives thereof.
The freedom to prescribe $W^{\sigma}$ reflects the freedom to choose coordinates off the initial surface.
We then impose the \textit{$\hat g$-generalized wave-map gauge condition}
\begin{equation*}
 H^{\sigma} \,=\, 0\;.
 \end{equation*}
The \textit{reduced Ricci tensor} $R^{(H)}_{\mu\nu}$ is defined as
\begin{equation}
 R^{(H)}_{\mu\nu} :=  R_{\mu\nu} -  g_{\sigma(\mu}\hat\nabla_{\nu)} H^{ \sigma}
 \label{ricci_riccired}
 \;,
\end{equation}
where $\hat\nabla$  denotes the covariant derivative associated to the target metric.
The principal part of the reduced Ricci tensor \textit{is} a d'Alembert operator.

Furthermore, we define a \textit{reduced wave-operator} as follows:
We observe that for any covector field $v_{\lambda}$ we have
\begin{eqnarray*}
 \Box_g v_{\lambda} &=& g^{\mu\nu}\partial_{\mu}\partial_{\nu}v_{\lambda}-g^{\mu\nu}(\partial_{\mu}\Gamma^{\sigma}_{\nu\lambda})v_{\sigma} +f_{\lambda}(g,\partial g, v, \partial v)
 \\
 &=& g^{\mu\nu}\partial_{\mu}\partial_{\nu}v_{\lambda}
  +(R_{\lambda}{}^{\sigma} - \partial_{\lambda}(g^{\mu\nu}\Gamma^{\sigma}_{\mu\nu}))v_{\sigma} +f_{\lambda}(g,\partial g, v, \partial v)
  \\
 &=& g^{\mu\nu}\partial_{\mu}\partial_{\nu}v_{\lambda}
  +(R_{\lambda}{}^{\sigma} - \partial_{\lambda}H^{\sigma})v_{\sigma} +f_{\lambda}(g,\partial g, v, \partial v,\hat g, \partial\hat g, \partial^2 \hat g,\partial W)
   \\
 &=& g^{\mu\nu}\partial_{\mu}\partial_{\nu}v_{\lambda}
  +(R_{\mu\lambda}^{(H)} + g_{\sigma[\lambda}\hat\nabla_{\mu]} H^{ \sigma})v^{\mu} +f_{\lambda}(g,\partial g, v, \partial v,\hat g, \partial\hat g, \partial^2 \hat g,\partial W)
   \;.
\end{eqnarray*}
Similarly, the action on a vector field $v^{\lambda}$ yields
\begin{eqnarray*}
 \Box_g v^{\lambda} &=&
  g^{\mu\nu}\partial_{\mu}\partial_{\nu}v^{\lambda}
  -(R_{(H)}^{\mu\lambda} +  g^{\sigma[\lambda}\hat\nabla_{\sigma} H^{ \mu]})v_{\mu} +f^{\lambda}(g,\partial g, v, \partial v,\hat g, \partial\hat g, \partial^2 \hat g,\partial W)
   \;.
\end{eqnarray*}
This motivates to define a \textit{reduced wave-operator} $\Box_g^{(H)}$ via its action on (co)vector  fields in the following way:
\begin{eqnarray*}
 \Box_g^{(H)}v_{\lambda} &:=& \Box_g v_{\lambda} - g_{\sigma[\lambda}(\hat\nabla_{\mu]} H^{ \sigma})v^{\mu}
  +(2L_{\mu\lambda} -R^{(H)}_{\mu\lambda} + \frac{1}{6}Rg_{\mu\lambda})v^{\mu}
   \;,
   \\
  \Box_g^{(H)}v^{\lambda} &:=& \Box_g v^{\lambda}  +  g^{\sigma[\lambda}(\hat\nabla_{\sigma} H^{ \mu]})v_{\mu}
  -(2L^{\mu\lambda} -R_{(H)}^{\mu\lambda} + \frac{1}{6}Rg^{\mu\lambda})v_{\mu}
   \;.
\end{eqnarray*}
For arbitrary tensor fields we set
\begin{eqnarray*}
 \Box_g^{(H)}v_{\alpha_1\dots\alpha_n}{}^{\beta_1\dots\beta_m} &:=&
   \Box_g v_{\alpha_1\dots\alpha_n}{}^{\beta_1\dots\beta_m}
     - \sum_i g_{\sigma[\alpha_i}(\hat\nabla_{\mu]} H^{ \sigma})v_{\alpha_1\dots}{}^{\mu}{}_{\dots\alpha_n}{}^{\beta_1\dots\beta_m}
 \\
 && +\sum_i(2L_{\mu\alpha_i} -R^{(H)}_{\mu\alpha_i} + \frac{1}{6}Rg_{\mu\alpha_i})v_{\alpha_1\dots}{}^{\mu}{}_{\dots\alpha_n}{}^{\beta_1\dots\beta_m}
 \\
  &&  + \sum_i g^{\sigma[\beta_i}(\hat\nabla_{\sigma} H^{ \mu]})v_{\alpha_1\dots\alpha_n}{}^{\beta_1\dots}{}_{\mu}{}^{\dots\beta_m}
 \\
  && -\sum_i(2L^{\mu\beta_i} -R_{(H)}^{\mu\beta_i} + \frac{1}{6}Rg^{\mu\beta_i})v_{\alpha_1\dots\alpha_n}{}^{\beta_1\dots}{}_{\mu}{}^{\dots\beta_m}
   \;,
\end{eqnarray*}
which is a proper wave-operator even if $g_{\mu\nu}$ is part of the unknowns since $L_{\mu\nu}$ and the gauge source function $R$ are regarded as independent of $g_{\mu\nu}$.
Note that the action of $\Box_g$ and $\Box^{(H)}_g$ coincides on scalars.
Moreover, if $H^{\sigma}=0$, and $L_{\mu\nu}$ and $R$ are known to be the Schouten tensor and the Ricci scalar of $g_{\mu\nu}$, respectively, then the action of
$\Box_g$ and $\Box^{(H)}_g$ coincides on all tensor fields.

\subsubsection*{Conformal wave equations}

Let us reconsider the system \eq{vconf1}, \eq{vconf2}, \eq{vconf3}, \eq{vconf4} and \eq{ricci_schouten}.
We replace the Ricci tensor by the reduced Ricci tensor and the wave-operator by the reduced wave-operator to
end up with a closed regular system of wave equations for $g_{\mu\nu}$, $\Theta$, $s$, $L_{\mu\nu}$ and $d_{\mu\nu\sigma}{}^{\rho}$,
\begin{eqnarray}
 \Box^{(H)}_{ g} L_{\mu\nu}&=&  4 L_{\mu\kappa} L_{\nu}{}^{\kappa} -  g_{\mu\nu}| L|^2
  - 2\Theta d_{\mu\sigma\nu}{}^{\rho}  L_{\rho}{}^{\sigma}
 + \frac{1}{6}\nabla_{\mu}\nabla_{\nu} R
  \label{cwe1}
  \;,
  \\
  \Box_g  s  &=& \Theta| L|^2 -\frac{1}{6}\nabla_{\kappa} R\,\nabla^{\kappa}\Theta  - \frac{1}{6} s  R
  \label{cwe2}
  \;,
  \\
  \Box_{ g}\Theta &=& 4 s-\frac{1}{6} \Theta  R
  \label{cwe3}
  \;,
  \\
  \Box^{(H)}_g d_{\mu\nu\sigma\rho}
  &=& \Theta d_{\mu\nu\kappa}{}^{\alpha}d_{\sigma\rho\alpha}{}^{\kappa}
   - 4\Theta d_{\sigma\kappa[\mu} {}^{\alpha}d_{\nu]\alpha\rho}{}^{\kappa}
  + \frac{1}{2}R d_{\mu\nu\sigma\rho}
  \label{cwe4}
  \;,
  \\
  R^{(H)}_{\mu\nu}[g] &=& 2L_{\mu\nu} + \frac{1}{6} R g_{\mu\nu}
  \label{cwe5}
  \;.
\end{eqnarray}
Henceforth the system \eq{cwe1}-\eq{cwe5} will be called \textit{conformal wave equations}
(CWE).

\begin{remark}Since $ R$ is regarded as a
gauge degree of freedom and not as unknown, there is no need to worry about its second-order derivatives appearing in \eq{cwe1}.
Note, however,  that, unlike $W^{\sigma}$, the gauge source function $R$  cannot be allowed to depend upon the fields  $L_{\mu\nu}$, $d_{\mu\nu\sigma\rho}$, $\Theta$ and $s$,
due to the fact that \eq{cwe1} contains second-order derivatives of $R$.
Since $	\nabla g=0$, $R$ can in principle be allowed to depend upon $g_{\mu\nu}$.
\end{remark}

\subsection{Consistency with the gauge condition}
\label{subsect_consist_gauge}

Let us analyse  now consistency of the CWE with the gauge conditions we imposed.
More concretely, we consider a characteristic initial value problem,
where, for definiteness, we think of two transversally intersecting null hypersurfaces or a light-cone,
 and assume that we have been given  initial
data ($\mathring g_{\mu\nu}$, $\mathring s$, $\mathring \Theta$, $\mathring L_{\mu\nu}$, $\mathring d_{\mu\nu\sigma}{}^{\rho}$). We further assume that there exists a smooth solution ($g_{\mu\nu}$, $s$, $\Theta$, $L_{\mu\nu}$, $d_{\mu\nu\sigma}{}^{\rho}$) of the CWE with gauge source function $R$ which induces these data.
We aim to work out conditions, which need to be satisfied initially, which guarantee consistency
with the gauge conditions in the sense that the solution implies $H^{\sigma}=0$ and $R_g=R$, where $R_g:=R[g]$ denotes the curvature scalar of $g_{\mu\nu}$. (Recall that there is, depending on the type of the characteristic initial surface, the additional gauge freedom to prescribe $\overline \Theta$ or $\overline s$, but here consistency is trivial.)

Let us outline the strategy.
To make sure that  $H^{\sigma}$ and $R-R_g$ vanish we shall derive a linear, homogeneous system of wave equations
for  $H^{\sigma}$  as well as some subsidiary fields, which is fulfilled by any solution of the CWE.
We shall see that it is not necessary to regard $R-R_g$ as an unknown.
We shall assume that all the fields which are regarded as unknowns in this set of equations vanish on the initial surface (in Section~\ref{sec_applicability} these assumptions will be justified).
Due to the uniqueness of solutions of wave equations, which is established by standard energy estimates, cf.\ e.g.\  \cite{friedlander},
we then conclude that the trivial solution is the only one and that the fields involved need to vanish everywhere.

\subsubsection*{Some properties of solutions of the CWE}
Let establish some properties of solutions of the CWE.
First of all we show that
 the tensors $g_{\mu\nu}$ and $L_{\mu\nu}$ are symmetric,
supposing that their initial data are (and that $d_{\mu\nu\sigma\rho}$ satisfies a certain symmetry property on the initial surface).
\begin{lemma}
\label{lem_g_sym}
 Assume that the initial data on a characteristic initial surface $S$  of some smooth solution of the CWE are such that
$g_{\mu\nu}|_S$ is the restriction to $S$ of a Lorentzian metric, that $ L_{[\mu\nu]}|_S=0$ and $ d_{\mu\nu\sigma\rho}|_S=d_{\sigma\rho\mu\nu}$.
 Then the solution has the following properties:
 \begin{enumerate}
  \item $g_{\mu\nu}$ and $L_{\mu\nu}$ are symmetric tensors,
  \item $ d_{\mu\nu\sigma\rho}=  d_{\sigma\rho\mu\nu}$.
 \end{enumerate}
\end{lemma}

\begin{remark}
A priori it might  happen that $g_{\mu\nu}$ becomes non-symmetric away from the initial surface.
However, the lemma shows that the tensor $g_{\mu\nu}$ does indeed define a metric as long as it does not degenerate  (i.e.\ at least sufficiently close to the vertex or the intersection manifold, respectively).
Later on, the initial data will be constructed from certain free data such that  all the hypotheses of Lemma~\ref{lem_g_sym} are satisfied, we thus will  assume throughout that $g_{\mu\nu}$ and $L_{\mu\nu}$ have their usual symmetry properties.
\end{remark}

\begin{proof}
Equation \eq{cwe4} yields%
\footnote{
\label{non-symmetric_g}
The indices are raised and lowered as follows: $v^{\mu} := g^{\mu\nu}v_{\nu}$ and $w_{\mu} := g_{\mu\nu}w^{\nu}$. Note for this that $g_{\mu\nu}$ is non-degenerated sufficiently close to $S$.  The definition of the Ricci tensor, which appears in \eq{cwe5}, in terms of Christoffel symbols which in turn are expressed in terms of $g$ make sense even if $g$ is not symmetric.}
\begin{eqnarray}
  \Box^{(H)}_g (d_{\mu\nu\sigma\rho}- d_{\sigma\rho\mu\nu}) &=& 4\Theta [g^{[\alpha\beta]}d_{\sigma\beta\mu\kappa} d_{\rho\alpha\nu}{}^{\kappa}
 - g^{[\gamma\kappa]}d_{\sigma\beta\mu\kappa}d_{\rho}{}^{\beta}{}_{\nu\gamma}]
 \nonumber
\\
 && \hspace{-10em} + 2\Theta g^{\alpha\beta}g^{\kappa\gamma}  [d_{\rho\alpha\nu\gamma} (d_{\mu\kappa\sigma\beta}-d_{\sigma\beta\mu\kappa})
- d_{\sigma\kappa\mu\beta}(d_{\nu\alpha\rho\gamma} -d_{\rho\gamma\nu\alpha} )]
 \nonumber
\\
&& \hspace{-10em} + \frac{1}{2}R (d_{\mu\nu\sigma\rho}- d_{\sigma\rho\mu\nu})
 \label{symmetry_C}
 \;.
\end{eqnarray}
From \eq{cwe1} and \eq{cwe5} we find
\begin{eqnarray}
 \Box^{(H)}_{ g} L_{[\mu\nu]}&=&  4g_{[\alpha\beta]}L_{\mu}{}^{\alpha}L_{\nu}{}^{\beta} - g_{[\mu\nu]}| L|^2
  + \Theta g^{\rho\gamma}L_{\rho }{}^{\sigma}(d_{\nu\sigma\mu\gamma}-d_{\mu\gamma\nu\sigma})
 \nonumber
\\
 && + 2\Theta g^{\sigma\kappa}d_{\mu}{}^{\rho}{}_{\nu\sigma}L_{[\rho\kappa]}
 - 2\Theta g^{[\sigma\kappa]}d_{\mu\sigma\nu}{}^{\rho}L_{\rho\kappa}
  \;,
 \label{antisym_L}
\\
  R^{(H)}_{[\mu\nu]}[g_{(\sigma\rho)},g_{[\sigma\rho]}] &=&  2L_{[\mu\nu]} + \frac{1}{6} R g_{[\mu\nu]}
 \label{antisym_R}
 \;.
\end{eqnarray}
The equations \eq{symmetry_C}-\eq{antisym_R} are to be read as a linear, homogeneous system of wave equations  satisfied by $g_{[\mu\nu]}$, $L_{[\mu\nu]}$
and $d_{\mu\nu\sigma\rho}- d_{\sigma\rho\mu\nu}$, and with all the other fields regarded as being given.
Since, by assumption, these fields vanish initially they have to vanish everywhere and the assertion follows.
\qed
\end{proof}

It is useful to derive some more properties of the tensor $d_{\mu\nu\sigma\rho}$.
We emphasize that $d_{\mu\nu\sigma\rho}$ is assumed to be part of some given   solution of the CWE and that, a priori, it
neither needs to be the rescaled Weyl tensor nor does it need to have all its algebraic properties.

\begin{lemma}
\label{lemma_properties_d}
 Assume that $d_{\mu\nu\sigma\rho}$ belongs to a solution of the CWE \eq{cwe1}-\eq{cwe5} for which the hypotheses of
Lemma~\ref{lem_g_sym} are fulfilled. Then the tensor $d_{\mu\nu\sigma\rho}$ has the following properties:
\begin{enumerate}
 \item[(i)] $d_{\mu\nu\sigma\rho}=d_{\sigma\rho\mu\nu}$,
\item[(ii)] $d_{\mu\nu\sigma\rho}$ is anti-symmetric in its first two and last two indices,
 \item[(iii)] $d_{\mu\nu\sigma\rho}$ satisfies the first Bianchi identity, i.e.\  $d_{[\mu\nu\sigma]\rho}=0$,
 \item[(iv)] $d_{\mu\nu\sigma\rho}$ is trace-free,
\end{enumerate}
supposing that (i)-(iv) hold initially.
\end{lemma}

\begin{remark}
 The constraint equations we shall impose later on on the initial data guarantee that (i)-(iv) are initially satisfied.
As for $g_{\mu\nu}$ and $L_{\mu\nu}$ we shall therefore use the implications of this lemma
without mentioning it each time.
\end{remark}

\begin{proof}
(i) This is part of  the proof of Lemma~\ref{lem_g_sym}.

(ii) Equation \eq{cwe4} implies a linear, homogeneous wave equation for $d_{(\mu\nu)\sigma\rho}$,
\begin{eqnarray*}
\Box^{(H)}_g d_{(\mu\nu)\sigma\rho} &=&
  \Theta d_{\sigma\rho\alpha}{}^{\kappa}  d_{(\mu\nu)\kappa}{}^{\alpha}
   + \frac{1}{2} R d_{(\mu\nu)\sigma\rho}
 \;,
\end{eqnarray*}
i.e.\ the tensor $d_{\mu\nu\sigma\rho}$ is antisymmetric in its first two (and therefore by (i) in its last two indices) since this is assumed to be  initially the case.

(iii) Due to the (anti-)symmetry properties (i)-(ii), we find the following linear, homogeneous wave equation from \eq{cwe4},
\begin{eqnarray*}
 \Box^{(H)}_g d_{[\mu\nu\sigma]\rho} &=&
  \Theta d_{[\mu\nu|\kappa|}{}^{\alpha}d_{\sigma]\rho\alpha}{}^{\kappa}
   +4\Theta d_{\kappa[\sigma\mu} {}^{\alpha}d_{\nu]\alpha\rho}{}^{\kappa}
    + \frac{1}{2}R d_{[\mu\nu\sigma]\rho}
\\
 &=& 2\Theta d_{\sigma\alpha\rho}{}^{\kappa} d_{[\kappa\mu\nu]} {}^{\alpha}+2\Theta d_{\mu\alpha\rho}{}^{\kappa}  d_{[\kappa\nu\sigma]} {}^{\alpha}
  +2\Theta d_{\nu\alpha\rho}{}^{\kappa}  d_{[\kappa\sigma\mu]} {}^{\alpha}
\\
&& + \Theta d_{\mu\nu\kappa} {}^{\alpha}d_{[\alpha\sigma\rho]}{}^{\kappa}
  +\Theta d_{\nu\sigma\kappa} {}^{\alpha}d_{[\alpha\mu\rho]}{}^{\kappa}
 +\Theta d_{\sigma\mu\kappa} {}^{\alpha}d_{[\alpha\nu\rho]}{}^{\kappa}
 + \frac{1}{2}R d_{[\mu\nu\sigma]\rho}
 \;.
\end{eqnarray*}

(iv)  It remains to be shown that $d_{\mu\rho\sigma}{}^{\rho}=0$.
Employing the properties (i)-(iii) we conclude from \eq{cwe4} that
\begin{eqnarray*}
 \Box^{(H)}_g d_{\mu\rho\sigma}{}^{\rho} =
 - 2\Theta d_{\sigma}{}^{\kappa}{}_{\mu} {}^{\alpha}d_{\kappa\rho\alpha}{}^{\rho}
  + \frac{1}{2}Rd_{\mu\rho\sigma}{}^{\rho}
 \;,
\end{eqnarray*}
which is again a  linear, homogeneous wave equation.
\qed
\end{proof}

\newpage

Next, let us establish another important property:
\begin{lemma}
\label{lemma_properties_L}
 Assume that the hypotheses of Lemma~\ref{lem_g_sym} and \ref{lemma_properties_d} are satisfied and that, in addition, the trace
 \[
  L:= L_{\sigma}{}^{\sigma}
\]
of $L_{\mu\nu}$ coincides on the initial surface with one sixth of the gauge source function~$R$, $\overline L=\frac{1}{6}\overline R$. Then
\begin{equation}
 L=\frac{1}{6} R
 \label{vanishingL}
  \;.
\end{equation}
(This is what one would expect if $L_{\mu\nu}$ was the Schouten tensor and $R$ the Ricci scalar.)
\end{lemma}

\begin{proof}
We observe that in virtue of \eq{cwe1} the tracelessness of $d_{\mu\nu\sigma\rho}$ implies
\begin{eqnarray*}
 \Box_{ g} (L-\frac{1}{6} R) &=& 0
 \;.
\end{eqnarray*}
and the assertion follows again from standard uniqueness results for linear wave equations.
\qed
\end{proof}

\subsubsection*{Gauge consistency}

Let us return to the question of whether we have consistency with the gauge condition in the sense that a solution of the CWE satisfies $H^{\sigma}=0$
 and $R_g=R$.
For that we assume that all the hypotheses of Lemma~\ref{lem_g_sym}, \ref{lemma_properties_d} and \ref{lemma_properties_L}
 are fulfilled.
We consider the identity
\begin{eqnarray}
 R_{\mu\nu} - \frac{1}{2} R_g g_{\mu\nu} \equiv R^{(H)}_{\mu\nu} - \frac{1}{2}R^{(H)}g_{\mu\nu} + g_{\sigma(\mu}\hat\nabla_{\nu)} H^{ \sigma}
 -\frac{1}{2} g_{\mu\nu} \hat\nabla_{\sigma} H^{ \sigma}
 \label{einstein_tensor}
 \;.
\end{eqnarray}
Invoking \eq{cwe5} and Lemma~\ref{lemma_properties_L} we deduce that
\begin{eqnarray}
  && \hspace{-2em} R_{\mu\nu}   - \frac{1}{2} R_ gg_{\mu\nu} \,= \,2 L_{\mu\nu}
  -\frac{1}{3} R g_{\mu\nu} + g_{\sigma(\mu}\hat\nabla_{\nu)} H^{ \sigma} -\frac{1}{2} g_{\mu\nu} \hat\nabla_{\sigma} H^{ \sigma}
 \nonumber
\\
 \overset{\mathrm{Bianchi}}{\Longrightarrow} && \nabla^{\nu} \hat\nabla_{\nu} H^{ \alpha}+2g^{\mu\alpha} \nabla_{[\sigma} \hat\nabla_{\mu]} H^{ \sigma}
 + 4(\nabla^{\nu} L_{\nu}{}^{\alpha}- \frac{1}{6}\nabla^{\alpha}R ) \,=\,0
 \label{wave_H}
 \;.
\end{eqnarray}
Be aware that at this stage it is not known whether $L_{\mu\nu}$ coincides with the Schouten tensor and thus  satisfies the contracted Bianchi identity \eq{BianchiL} such that the term in brackets in \eq{wave_H} drops out.  That is the reason why we cannot immediately deduce $H^{\sigma}=0$ as in~\cite{CCM2} supposing that this is initially the case.

Given two covariant derivative operators $\nabla$ and $\hat\nabla$ (associated to the metrics $g$ and $\hat g$, respectively), there exists a tensor field $C^{\sigma}_{\mu\nu}=C^{\sigma}_{\nu\mu}$,
which depends on $g$, $\partial g$, $\hat g$ and $\partial\hat g$,
such that
\begin{equation}
   \nabla_{\mu}v^{\sigma}-\hat\nabla_{\mu}v^{\sigma} = C^{\sigma}_{\mu\nu}v^{\nu}
 \label{nabla_hat}
 \;,
\end{equation}
for any vector $v^{\sigma}$, and similar formulae hold for tensor fields of other types.
Setting
\begin{equation}
 \zeta_{\mu}:= -4(\nabla_{\nu}L_{\mu}{}^{\nu} -\frac{1}{6} \nabla_{\mu} R)
 \;,
\end{equation}
the equation \eq{wave_H} can therefore be written as
\begin{eqnarray}
  \Box_gH^{\alpha}
 &=& \zeta^{\alpha}  + f^{\alpha}(g,\hat g; H,\nabla H)
 \label{wave_H*}
 \;,
\end{eqnarray}
which is a linear wave equation satisfied by the wave-gauge vector $H^{\sigma}$.%
\footnote{Note that in this part the metric is regarded as being given, so $\Box_g$ is a wave-operator
and there is no need to work with the reduced wave-operator $\Box_g^{(H)}$.
}
In \eq{wave_H*}, as in what follows, the generic smooth field $ f^{\alpha}(g,\hat g; H,\nabla H)$, or more general  $ f_{\alpha_1\dots\alpha_p}{}^{\beta_1\dots\beta_q}(v_1,\dots, v_m; w_1,\dots, w_n)$,
represents a sum of fields, each of which contains precisely 
 one multiplicative factor from the set $\{w_i\}$ as well as further factors which may depend on the $v_j$'s and also higher-order derivatives
of the $v_j$'s. The latter does not cause any problems  since the $v_j$'s will be regarded as given fields rather than unknowns of the system we are about to derive.
In most cases we will therefore simply write  $ f_{\alpha_1\dots\alpha_p}{}^{\beta_1\dots\beta_q}(x; w_1,\dots, w_n)$.

Taking the trace of \eq{einstein_tensor} and inserting \eq{cwe5},  yields (note that $L=R/6$)
\begin{eqnarray}
 R_g &\equiv&   R^{(H)}  + \hat\nabla_{\sigma} H^{ \sigma} =  R +   \hat\nabla_{\sigma} H^{ \sigma}
 \label{eqn_curvscalar}
 \;.
\end{eqnarray}
The vanishing of $H^{\sigma}$ would  therefore immediately ensure that $R_g=R$.

The tensor $d_{\mu\nu\sigma}{}^{\rho}$ is supposed to be part of a solution of the CWE.
Note, again, that  at this stage
it is by no means clear whether it, indeed, represents the rescaled Weyl tensor of $g_{\mu\nu}$ and $\Theta$.
As before, we denote by $W_{\mu\nu\sigma}{}^{\rho}$  the Weyl tensor associated  to $g_{\mu\nu}$, defined via the decomposition
\begin{eqnarray}
  R_{\mu\nu\sigma\rho} = W_{\mu\nu\sigma\rho}  + g_{\sigma[\mu}R_{\nu]\rho} - g_{\rho[\mu}R_{\nu]\sigma} - \frac{1}{3}R_g g_{\sigma[\mu}g_{\nu]\rho}
 \label{riemann_ricci_weyl}
 \;.
\end{eqnarray}

As outlined above we want to derive a closed, linear, homogeneous system of wave equations for a certain set of fields
in order to establish the vanishing of $H^{\sigma}$.
First of all, we need a wave equation for $\zeta_{\mu}$.
Making use of the Bianchi identity, \eq{vanishingL} and \eq{cwe1}, we obtain
\begin{eqnarray}
 \Box_g \zeta_{\mu}
 &\equiv &  -4 \nabla_{\nu} \Box_g L_{\mu}{}^{\nu}  +\frac{2}{3}\Box_g\nabla_{\mu} R
   -8 \nabla^{\nu}(W_{\mu\sigma\nu}{}^{\rho}L_{\rho}{}^{\sigma})+8R_{\nu}{}^{\kappa}\nabla_{\kappa}L_{\mu}{}^{\nu}
    \nonumber
\\
 &&  -4 R_{\nu}{}^{\kappa}  \nabla_{\mu}L_{\kappa}{}^{\nu} -R_{\mu}{}^{\nu}\zeta_{\nu}+\frac{1}{3} R_g \zeta_{\mu}   -4 R_{\mu}{}^{\nu} \nabla_{\nu}(L-\frac{1}{6} R)
 \nonumber
  \\
 && +\frac{4}{3} R_g  \nabla_{\mu}(L-\frac{1}{6}R)  +\frac{8}{3}L_{\mu}{}^{\nu}\nabla_{\nu} R_g  -\frac{2}{3}L \nabla_{\mu} R_g
  \nonumber
\\
 &=& (4  L_{\mu}{}^{\nu} -R_{\mu}{}^{\nu})\zeta_{\nu}
    + 4(2L_{\nu\sigma}- R_{\nu\sigma}+\frac{1}{6} R g_{\nu\sigma})( \nabla_{\mu} L^{\nu\sigma}-2\nabla^{\sigma} L_{\mu}{}^{\nu})
       \nonumber
 \\
 &&
  -8 \nabla^{\nu}[(W_{\mu\sigma\nu}{}^{\rho}-\Theta d_{\mu\sigma\nu}{}^{\rho})L_{\rho}{}^{\sigma}]
+\frac{1}{3}(\zeta_{\mu} + 8L_{\mu}{}^{\nu}\nabla_{\nu}  - 2L \nabla_{\mu})( R_g-R)
  \nonumber
\\
 &&   -4  L_{\nu}{}^{\lambda}\nabla^{\nu}\nabla_{[\lambda}H_{\mu]}
    -4 L_{\mu}{}^{\lambda}\nabla^{\nu}\nabla_{[\lambda}H_{\nu]}
+ f_{\mu}(x;H,\nabla H)
 \;.
  \label{equation_Box_zeta}
\end{eqnarray}
We employ \eq{ricci_riccired}, \eq{cwe5}, \eq{eqn_curvscalar} and \eq{wave_H*} to end up with
\begin{eqnarray}
 \Box_g \zeta_{\mu}
  &=& 4  L_{\mu}{}^{\nu}\zeta_{\nu}-\frac{1}{6}R\zeta_{\mu}
      -8 \nabla^{\nu}[(W_{\mu\sigma\nu}{}^{\rho}-\Theta d_{\mu\sigma\nu}{}^{\rho})L_{\rho}{}^{\sigma}]
       -\frac{2}{3}L \nabla_{\mu}\nabla_{\nu} H^{\nu}
       \nonumber
 \\
 &&   +\frac{2}{3}L_{\mu}{}^{\nu}\nabla_{\nu}\nabla_{\sigma} H^{\sigma}
  -4  L_{\nu}{}^{\lambda}\nabla^{\nu}\nabla_{[\lambda}H_{\mu]}
   + f_{\mu}(x;H,\nabla H)
    \;.
 \label{wave_zetavec_prov}
\end{eqnarray}
In order to get rid of the undesired second-order derivatives in $H^{\sigma}$, we introduce
the tensor  field
\begin{equation}
 K_{\mu}{}^{\nu} :=\nabla_{\mu}H^{\nu}
 \label{dfn_K_diffH}
\end{equation}
as another unknown for which we need to derive a wave equation, as well.
We employ the fact that the right-hand side of \eq{wave_H*} does not contain derivatives of $\zeta^{\alpha}$:
Differentiating \eq{wave_H*} we are straightforwardly led to the desired equation,
\begin{eqnarray}
  \Box_g K_{\mu\nu} &\equiv & \nabla_{\mu}\Box_g H_{\nu} + R_{\mu}{}^{\kappa}\nabla_{\kappa}H_{\nu}
 + H^{\kappa}  \nabla_{\sigma}R_{\kappa\nu\mu}{}^{\sigma}
 +2 R_{\kappa\nu\mu}{}^{\sigma}  \nabla_{\sigma}H^{\kappa}
 \nonumber
\\
 &=&  \nabla_{\mu}\zeta_{\nu}
 +f_{\mu\nu}(x;H,\nabla H,\nabla K)
 \label{wave_DiffH}
 \;.
\end{eqnarray}
Moreover, \eq{wave_zetavec_prov} becomes a wave equation for $\zeta_{\mu}$,
\begin{eqnarray}
 \Box_g \zeta_{\mu} &=&   4  L_{\mu}{}^{\nu}\zeta_{\nu}-\frac{1}{6}R\zeta_{\mu}
 -8 \nabla^{\nu}[(W_{\mu\sigma\nu\rho}- \Theta d_{\mu\sigma\nu\rho})L^{\sigma\rho}]
\nonumber
\\
&&  + f_{\mu}(x;H,\nabla H,\nabla K)
 \label{wave_zetavec}
 \;.
\end{eqnarray}
We observe that we need a wave equation for $W_{\mu\sigma\nu\rho}- \Theta d_{\mu\sigma\nu\rho}$
(actually just for its contraction with $L^{\sigma\rho}$, but for later purposes it is useful to show that
$\Theta d_{\mu\sigma\nu\rho}$ coincides with the Weyl tensor, which would follow, supposing, as usual, that it is initially true).
For this purpose let us introduce the tensor field $\zeta_{\mu\nu\sigma}$,
\begin{eqnarray*}
 \zeta_{\mu\nu\sigma}:= 4\nabla_{[\sigma} L_{\nu]\mu}
 \;.
\end{eqnarray*}
Note that $ \zeta_{[\mu\nu\sigma]}=0$ for a symmetric $L_{\mu\nu}$.

Starting from the second Bianchi identity, 
we find with \eq{riemann_ricci_weyl}, \eq{ricci_riccired}, \eq{cwe5} and \eq{eqn_curvscalar}
\begin{eqnarray}
  \nabla_{\alpha}W_{\mu\nu\sigma\rho}
 & \equiv& -2 \nabla_{[\mu}W_{\nu]\alpha\sigma\rho}  + 2\nabla_{[\alpha}R_{\nu][\sigma}g_{\rho]\mu}
 - 2\nabla_{[\alpha}R_{\mu][\sigma} g_{\rho]\nu}
  - 2\nabla_{[\mu}R_{\nu][\sigma} g_{\rho]\alpha}
 \nonumber
\\
 &&  + \frac{2}{3}g_{\mu[\sigma}g_{\rho][\nu}\nabla_{\alpha]}R_g
 - \frac{1}{3}g_{\alpha[\sigma}g_{\rho]\nu}\nabla_{\mu}R_g
 \nonumber
\\
 & =&  g_{\mu[\sigma}\zeta_{\rho]\alpha\nu} + g_{\nu[\sigma}\zeta_{\rho]\mu\alpha} -g_{\alpha[\sigma}\zeta_{\rho]\mu\nu}
   - 2\nabla_{[\mu}W_{\nu]\alpha\sigma\rho}
 \nonumber
\\
 &&  + \frac{2}{3}g_{\mu[\sigma}g_{\rho][\nu}\nabla_{\alpha]}\nabla_{\kappa}H^{\kappa}
 - \frac{1}{3}g_{\alpha[\sigma}g_{\rho]\nu}\nabla_{\mu}\nabla_{\kappa}H^{\kappa}
   + g_{\alpha[\sigma}\nabla_{\rho]}\nabla_{[\mu}H_{\nu]}
 \nonumber
\\
 && + g_{\mu[\sigma}\nabla_{\rho]}\nabla_{[\nu}H_{\alpha]}
 +g_{\nu[\sigma}\nabla_{\rho]}\nabla_{[\alpha}H_{\mu]}
  + f_{\alpha\mu\nu\sigma\rho}(x; H, \nabla H)
 \label{nabla_weyl}
 \;.
\end{eqnarray}
Applying $\nabla^{\alpha}$ yields
\begin{eqnarray}
 \Box_g W_{\mu\nu\sigma\rho}
 & =&  2 \nabla_{[\nu}\nabla^{\alpha}W_{\mu]\alpha\sigma\rho}
 + W_{\mu\nu\alpha}{}^{\kappa}W_{\sigma\rho\kappa} {}^{\alpha}  - 4W_{\sigma\kappa[\mu} {}^{\alpha}W_{\nu]\alpha\rho}{}^{\kappa} +\frac{1}{3}RW_{\mu\nu\sigma\rho}
  \nonumber
\\
 &&+2 ( g_{\rho[\mu} W_{\nu]\alpha\sigma}{}^{\kappa}- g_{\sigma[\mu}W_{\nu]\alpha\rho}{}^{\kappa})L_{\kappa}{}^{\alpha}
 -2L_{[\mu}{}^{\kappa} W_{\nu]\kappa\sigma\rho}  - 2L_{[\sigma}{}^{\kappa}  W_{\rho]\kappa\mu\nu}
   \nonumber
\\
 && + \nabla_{[\sigma}\zeta_{\rho]\nu\mu}
 + g_{\sigma[\mu}\nabla^{\alpha}\zeta_{|\rho\alpha|\nu]}  - g_{\rho[\mu}\nabla^{\alpha}\zeta_{|\sigma\alpha|\nu]}
  + \frac{1}{3}g_{\mu[\sigma}g_{\rho]\nu}\nabla_{\kappa}\Box_g H^{\kappa}
   \nonumber
\\
 &&+ \frac{1}{6}g_{\mu[\sigma}\nabla_{\rho]}\nabla_{\nu}\nabla_{\alpha}H^{\alpha}- \frac{1}{6} g_{\nu[\sigma}\nabla_{\rho]}\nabla_{\mu}\nabla_{\alpha}H^{\alpha}
  - \frac{1}{2} g_{\mu[\sigma}\nabla_{\rho]}\Box_g H_{\nu}
    \nonumber
\\
 &&     +\frac{1}{2} g_{\nu[\sigma}\nabla_{\rho]}\Box_g H_{\mu}
  + f_{\mu\nu\sigma\rho}(x; H, \nabla H,\nabla K)
 \;.
  \label{wave_W_sub}
\end{eqnarray}
Before we manipulate this expression any further it is useful to compute
\begin{eqnarray}
 \nabla^{\alpha}\zeta_{\mu\nu\alpha} &\equiv& 2  \Box_g L_{\mu\nu}-   2\nabla_{\nu} \nabla_{\alpha}L_{\mu}{}^{\alpha}
  - 2 R_{\alpha\nu\mu}{}^{\kappa}L_{\kappa}{}^{\alpha} - 2  R_{\nu\kappa}L_{\mu}{}^{\kappa}
 \nonumber
\\
 & =&  2\Box^{(H)}_g L_{\mu\nu}
   + 2 L_{\alpha}{}^{\kappa}W_{\mu\kappa\nu}{}^{\alpha}
     - 3 L_{\mu}{}^{\kappa} R^{(H)}_{\nu\kappa}  -L_{\nu}{}^{\alpha}R^{(H)}_{\mu\alpha}
  +  g_{\mu\nu} L^{\alpha\kappa}R^{(H)}_{\alpha\kappa}
   \nonumber
 \\
 && -   \frac{1}{2} \nabla_{\mu} \nabla_{\nu}  (R_g-\frac{1}{3}R)   +L R^{(H)}_{\mu\nu}
    + \frac{1}{3}  L_{\mu\nu}R_g   - \frac{1}{3}  L R_g g_{\mu\nu}
  \nonumber
  \\
  &&   + \frac{1}{2}   \nabla_{\nu} \nabla_{\kappa}\hat\nabla_{\mu} H^{ \kappa}
 + \frac{1}{2} g_{\mu\kappa} \nabla_{\nu}  \nabla^{\alpha}\hat\nabla_{\alpha} H^{ \kappa}
   +f_{\mu\nu}(x;H,\nabla H)
 \nonumber
\\
  & =&  2( W_{\mu\alpha\nu}{}^{\kappa} - 2\Theta d_{\mu\alpha\nu}{}^{\kappa} ) L_{\kappa}{}^{\alpha}
   + \frac{1}{2}\nabla_{\nu} \Box_g H_{ \mu}
\nonumber
\\
&&  +f_{\mu\nu}(x;H,\nabla H, \nabla K)
 \label{contracted_zeta_B}
 \;,
\end{eqnarray}
which follows from \eq{ricci_riccired}, \eq{cwe1}, \eq{cwe5}, \eq{vanishingL}, \eq{wave_H}, \eq{eqn_curvscalar} and \eq{riemann_ricci_weyl}.
Due to the 
Bianchi identity, \eq{ricci_riccired}, \eq{cwe5} and \eq{eqn_curvscalar}, we also have
\begin{eqnarray}
 \nabla_{\alpha}W_{\mu\nu\sigma}{}^{\alpha} \equiv- \nabla_{[\mu}R_{\nu]\sigma}  - \frac{1}{6} g_{\sigma[\mu}\nabla_{\nu]}R_g
\phantom{xxxxxxxxxxxxxxxxxxxxxxxxxx}
  \nonumber
 \\
  =  \frac{1}{2}\zeta_{\sigma\mu\nu}
 -  \frac{1}{2}\nabla_{\sigma} \nabla_{[\mu}H_{\nu]}
 - \frac{1}{6} g_{\sigma[\mu}\nabla_{\nu]}\nabla_{\kappa}H^{\kappa}
 + f_{\mu\nu\sigma}(x; H,\nabla H)
  \;.
  \label{contr_Bian_W}
\end{eqnarray}
Invoking \eq{contracted_zeta_B} and \eq{contr_Bian_W} we rewrite \eq{wave_W_sub} to obtain
\begin{eqnarray}
 \Box_g W_{\mu\nu\sigma\rho} &=&
  \nabla_{[\sigma}\zeta_{\rho]\nu\mu} -\nabla_{[\mu}\zeta_{\nu]\sigma\rho}
  + W_{\mu\nu\alpha}{}^{\kappa}W_{\sigma\rho\kappa} {}^{\alpha}  - 4W_{\sigma\kappa[\mu} {}^{\alpha}W_{\nu]\alpha\rho}{}^{\kappa} +\frac{1}{3}RW_{\mu\nu\sigma\rho}
   \nonumber
\\
 &&-2L_{[\mu}{}^{\kappa} W_{\nu]\kappa\sigma\rho}  - 2L_{[\sigma}{}^{\kappa}  W_{\rho]\kappa\mu\nu}
  +4L_{\kappa}{}^{\alpha}(  W_{\rho\alpha[\mu}{}^{\kappa} - \Theta d_{\rho\alpha[\mu}{}^{\kappa} )  g_{\nu]\sigma}
   \nonumber
\\
 &&  -4L_{\kappa}{}^{\alpha} ( W_{\sigma\alpha[\mu}{}^{\kappa} - \Theta d_{\sigma\alpha[\mu}{}^{\kappa} ) g_{\nu]\rho}
  + \frac{1}{2} g_{\rho[\mu}\nabla_{\nu]} \Box_g H_{ \sigma} - \frac{1}{2}g_{\sigma[\mu}\nabla_{\nu]} \Box_g H_{ \rho}
   \nonumber
 \\
 && +\frac{1}{2} g_{\nu[\sigma}\nabla_{\rho]}\Box_g H_{\mu}
   - \frac{1}{2} g_{\mu[\sigma}\nabla_{\rho]}\Box_g H_{\nu}
   + \frac{1}{3}g_{\mu[\sigma}g_{\rho]\nu}\nabla_{\kappa}\Box_g H^{\kappa}
   \nonumber
\\
 &&    + f_{\mu\nu\sigma\rho}(x; H, \nabla H,\nabla K)
  \;.
 \label{wave_W1}
\end{eqnarray}
We insert \eq{wave_H*},
\begin{eqnarray}
 \Box_g W_{\mu\nu\sigma\rho} &=&
  \nabla_{[\sigma}\zeta_{\rho]\nu\mu} -\nabla_{[\mu}\zeta_{\nu]\sigma\rho}
  + W_{\mu\nu\alpha}{}^{\kappa}W_{\sigma\rho\kappa} {}^{\alpha}  - 4W_{\sigma\kappa[\mu} {}^{\alpha}W_{\nu]\alpha\rho}{}^{\kappa}
   \nonumber
\\
 &&-2L_{[\mu}{}^{\kappa} W_{\nu]\kappa\sigma\rho}  - 2L_{[\sigma}{}^{\kappa}  W_{\rho]\kappa\mu\nu}
  +4L_{\kappa}{}^{\alpha}(  W_{\rho\alpha[\mu}{}^{\kappa} - \Theta d_{\rho\alpha[\mu}{}^{\kappa} )  g_{\nu]\sigma}
   \nonumber
\\
 &&  -4L_{\kappa}{}^{\alpha} ( W_{\sigma\alpha[\mu}{}^{\kappa} - \Theta d_{\sigma\alpha[\mu}{}^{\kappa} ) g_{\nu]\rho}
   +\frac{1}{3}RW_{\mu\nu\sigma\rho}
  + \frac{1}{3}g_{\mu[\sigma}g_{\rho]\nu}\nabla_{\kappa}\zeta^{\kappa}
   \nonumber
 \\
 && + \frac{1}{2} g_{\rho[\mu}\nabla_{\nu]} \zeta_{ \sigma}
 - \frac{1}{2}g_{\sigma[\mu}\nabla_{\nu]} \zeta_{ \rho} +\frac{1}{2} g_{\nu[\sigma}\nabla_{\rho]}\zeta_{\mu}
   - \frac{1}{2} g_{\mu[\sigma}\nabla_{\rho]}\zeta_{\nu}
   \nonumber
\\
 &&   + f_{\mu\nu\sigma\rho}(x; H, \nabla H,\nabla K)
  \;.
 \label{wave_W1II}
\end{eqnarray}
It proves useful to make the following definitions:
\begin{eqnarray}
 \varkappa_{\mu\nu\sigma} &:=& \frac{1}{2}\zeta_{\mu\nu\sigma} - \nabla_{\kappa}\Theta d_{\nu\sigma\mu}{}^{\kappa}
  \;,
  \\
 \Xi_{\mu\nu} &:=& \nabla_{\mu}\nabla_{\nu}\Theta +\Theta L_{\mu\nu} -sg_{\mu\nu}
  \;.
\end{eqnarray}
We observe the relation
\begin{eqnarray*}
 \nabla_{\rho} \zeta_{\mu\nu\sigma} \,=\,  2\nabla_{\rho}\varkappa_{\mu\nu\sigma} +2\nabla^{\kappa}\Theta \nabla_{\rho} d_{\nu\sigma\mu\kappa}
  +2\Xi_{\rho\kappa} d_{\nu\sigma\mu}{}^{\kappa} -2 L_{\rho}{}^{\kappa} \Theta d_{\nu\sigma\mu\kappa} + 2s d_{\nu\sigma\mu\rho}
   \;.
\end{eqnarray*}
Then, due to  the (anti-)symmetry properties of the tensor $d_{\mu\nu\sigma\rho}$ derived above, \eq{wave_W1II} yields
\begin{eqnarray}
 \Box_g W_{\mu\nu\sigma\rho} &=&
  2\nabla^{\kappa}\Theta \nabla_{[\sigma} d_{\rho]\kappa\nu\mu} -2\nabla^{\kappa}\Theta \nabla_{[\mu} d_{\nu]\kappa\sigma\rho}
   +2\nabla_{[\sigma}\varkappa_{\rho]\nu\mu}   -2\nabla_{[\mu}\varkappa_{\nu]\sigma\rho}
  \nonumber
 \\
 &&  + 2d_{\nu\mu[\rho}{}^{\kappa}\Xi_{\sigma]\kappa} -2 d_{\sigma\rho[\nu}{}^{\kappa}\Xi_{\mu]\kappa}
  + W_{\mu\nu\alpha}{}^{\kappa}W_{\sigma\rho\kappa} {}^{\alpha}  - 4W_{\sigma\kappa[\mu} {}^{\alpha}W_{\nu]\alpha\rho}{}^{\kappa}
  \nonumber
  \\
 && +\frac{1}{3}RW_{\mu\nu\sigma\rho}  + 2L_{[\mu}{}^{\kappa} (\Theta d_{\nu]\kappa\sigma\rho}-W_{\nu]\kappa\sigma\rho})
 - 2L_{[\sigma}{}^{\kappa}( \Theta d_{\rho]\kappa\nu\mu}-W_{\rho]\kappa\nu\mu})
  \nonumber
\\
 &&  +4L_{\kappa}{}^{\alpha}(  W_{\rho\alpha[\mu}{}^{\kappa} - \Theta d_{\rho\alpha[\mu}{}^{\kappa} )  g_{\nu]\sigma}
  -4L_{\kappa}{}^{\alpha} ( W_{\sigma\alpha[\mu}{}^{\kappa} - \Theta d_{\sigma\alpha[\mu}{}^{\kappa} ) g_{\nu]\rho}
  \nonumber
 \\
 && + \frac{1}{2} g_{\rho[\mu}\nabla_{\nu]} \zeta_{ \sigma}
 - \frac{1}{2}g_{\sigma[\mu}\nabla_{\nu]} \zeta_{ \rho} +\frac{1}{2} g_{\nu[\sigma}\nabla_{\rho]}\zeta_{\mu}
   - \frac{1}{2} g_{\mu[\sigma}\nabla_{\rho]}\zeta_{\nu}
  \nonumber
\\
 &&   + \frac{1}{3}g_{\mu[\sigma}g_{\rho]\nu}\nabla_{\kappa}\zeta^{\kappa}
  + 4s d_{\mu\nu\sigma\rho}
  + f_{\mu\nu\sigma\rho}(x; H, \nabla H,\nabla K)
  \;.
  \label{wave_W1III}
\end{eqnarray}
On the other hand, in virtue of \eq{cwe3} and \eq{cwe4}, we have
\begin{eqnarray}
 \Box_g (\Theta d_{\mu\nu\sigma\rho}) &\equiv& d_{\mu\nu\sigma\rho}\Box_g\Theta + \Theta \Box_g d_{\mu\nu\sigma\rho}
  +2\nabla^{\kappa}\Theta \nabla_{\kappa}d_{\mu\nu\sigma\rho}
   \nonumber
 \\
 &=& 4 s d_{\mu\nu\sigma\rho}  +2\nabla^{\kappa}\Theta \nabla_{\kappa}d_{\mu\nu\sigma\rho}
  +  \Theta^2 d_{\mu\nu\kappa}{}^{\alpha}d_{\sigma\rho\alpha}{}^{\kappa}
   - 4\Theta^2 d_{\sigma\kappa[\mu} {}^{\alpha}d_{\nu]\alpha\rho}{}^{\kappa}
      \nonumber
 \\
 &&  + \frac{1}{3}R \Theta d_{\mu\nu\sigma\rho}
    + f_{\mu\nu\sigma\rho}(x; H, \nabla H)
 \;.
  \label{wave_Theta_d}
\end{eqnarray}
Combining \eq{wave_W1III} and \eq{wave_Theta_d}, and invoking \eq{d_bianchi_divergence}, we are led to the wave equation
\begin{eqnarray}
 &&\hspace{-3em}\Box_g( W_{\mu\nu\sigma\rho} - \Theta d_{\mu\nu\sigma\rho})
  \,=\,     2\nabla_{[\sigma}\varkappa_{\rho]\nu\mu}  - 2\nabla_{[\mu}\varkappa_{\nu]\sigma\rho}
  + 2d_{\mu\nu[\sigma}{}^{\kappa}\Xi_{\rho]\kappa} + 2 d_{\sigma\rho[\mu}{}^{\kappa}\Xi_{\nu]\kappa}
\nonumber
    \\
 &&+W_{\mu\nu\alpha}{}^{\kappa}(W_{\sigma\rho\kappa} {}^{\alpha}-\Theta d_{\sigma\rho\kappa}{}^{\alpha} )
 +\Theta d_{\sigma\rho\kappa}{}^{\alpha} (W_{\mu\nu\alpha}{}^{\kappa} -  \Theta  d_{\mu\nu\alpha}{}^{\kappa})
\nonumber
\\
&& - 4W_{\sigma\kappa[\mu} {}^{\alpha}(W_{\nu]\alpha\rho}{}^{\kappa}  - \Theta d_{\nu]\alpha\rho}{}^{\kappa})
 - 4(W_{\sigma\kappa[\mu} {}^{\alpha}-\Theta d_{\sigma\kappa[\mu} {}^{\alpha})\Theta d_{\nu]\alpha\rho}{}^{\kappa}
\nonumber
\\
 && - 2L_{[\mu}{}^{\kappa} ( W_{\nu]\kappa\sigma\rho} - \Theta d_{\nu]\kappa\sigma\rho})
 + 2L_{[\sigma}{}^{\kappa}( W_{\rho]\kappa\nu\mu} -  \Theta d_{\rho]\kappa\nu\mu})
\nonumber
 \\
 && +4L_{\kappa}{}^{\alpha}(  W_{\rho\alpha[\mu}{}^{\kappa} - \Theta d_{\rho\alpha[\mu}{}^{\kappa} )  g_{\nu]\sigma}
 -4L_{\kappa}{}^{\alpha} ( W_{\sigma\alpha[\mu}{}^{\kappa} - \Theta d_{\sigma\alpha[\mu}{}^{\kappa} ) g_{\nu]\rho}
\nonumber
\\
  && + \frac{1}{3}R(W_{\mu\nu\sigma\rho} - \Theta d_{\mu\nu\sigma\rho})
 - \frac{1}{2}\nabla^{\kappa}\Theta(\epsilon_{\kappa\sigma\rho}{}^{\delta} \epsilon_{\mu\nu}{}^{\beta\gamma}  + \epsilon_{\kappa\mu\nu}{}^{\delta} \epsilon_{\sigma\rho}{}^{\beta\gamma}) \nabla_{\alpha}d_{\beta\gamma\delta}{}^{\alpha}
\nonumber
\\
 && + \frac{1}{2} g_{\rho[\mu}\nabla_{\nu]} \zeta_{ \sigma}
 - \frac{1}{2}g_{\sigma[\mu}\nabla_{\nu]} \zeta_{ \rho} +\frac{1}{2} g_{\nu[\sigma}\nabla_{\rho]}\zeta_{\mu}
   - \frac{1}{2} g_{\mu[\sigma}\nabla_{\rho]}\zeta_{\nu}
  \nonumber
\\
 &&   + \frac{1}{3}g_{\mu[\sigma}g_{\rho]\nu}\nabla_{\kappa}\zeta^{\kappa}
 + f_{\mu\nu\sigma\rho}(x;H, \nabla H,\nabla K)
 \label{wave_W-C}
 \;,
\end{eqnarray}
which is fulfilled by any solution of the CWE.

In order to  end up with a homogeneous system of wave equations,
it remains to derive wave equations for $\varkappa_{\mu\nu\sigma}$, $\Xi_{\mu\nu}$ and $\nabla_{\rho} d_{\mu\nu\sigma}{}^{\rho} $.
Let us start with $\nabla_{\rho} d_{\mu\nu\sigma}{}^{\rho} $,
%
\begin{eqnarray}
 \Box_g \nabla_{\rho}d_{\mu\nu\sigma}{}^{\rho}
 &\equiv &
 \nabla_{\rho}\Box_g d_{\mu\nu\sigma}{}^{\rho}
 -4 W_{\kappa\rho[\mu}{}^{\alpha} \nabla^{\kappa}d_{\nu]\alpha\sigma}{}^{\rho}
  +2 W_{\kappa\rho\sigma}{}^{\alpha} \nabla^{\kappa}d_{\mu\nu\alpha}{}^{\rho}
 \nonumber
\\
 &&- 2d_{\mu\nu\rho}{}^{\alpha}\nabla_{[\sigma}R_{\alpha]}{}^{\rho}
  - 2d_{\sigma\rho\nu}{}^{\alpha}\nabla_{[\mu}R_{\alpha]}{}^{\rho}
  +2d_{\sigma\rho\mu}{}^{\alpha}\nabla_{[\nu}R_{\alpha]}{}^{\rho}
  \nonumber
\\
 &&  +2 R_{\rho[\mu} \nabla^{\alpha}d_{\nu]\alpha\sigma}{}^{\rho}
  + R_{\sigma}{}^{\rho}\nabla_{\alpha}d_{\mu\nu\rho}{}^{\alpha}
 +3 R_{\rho}{}^{\alpha} \nabla_{[\mu}d_{\alpha\nu]\sigma}{}^{\rho}
\nonumber
\\
&& - \frac{1}{2}d_{\mu\nu\sigma}{}^{\alpha}  \nabla_{\alpha}R_g
  \nonumber
\\
 &=& 2  d_{\mu\nu\rho}{}^{\alpha}\varkappa^{\rho}{}_{\sigma\alpha}
 - 4  d_{\sigma\rho[\mu} {}^{\alpha} \varkappa^{\rho}{}_{\nu]\alpha}
 + (W_{\kappa\sigma\rho}{}^{\alpha}  -\Theta d_{\kappa\sigma\rho}{}^{\alpha})\nabla^{\kappa}d_{\mu\nu\alpha}{}^{\rho}
  \nonumber
\\
 && -4( W_{\kappa\rho[\mu}{}^{\alpha}  -  \Theta d_{\kappa\rho[\mu}{}^{\alpha})\nabla^{\kappa}d_{\nu]\alpha\sigma}{}^{\rho}
 +\frac{1}{2} R^{\rho\alpha}\epsilon_{\mu\alpha\nu}{}^{\delta} \epsilon_{\sigma\rho}{}^{\beta\gamma} \nabla_{\lambda}d_{\beta\gamma\delta}{}^{\lambda}
  \nonumber
\\
 &&  + 2R_{[\mu}{}^{\alpha} \nabla_{|\rho}d_{\sigma\alpha|\nu]}{}^{\rho}
 + \Theta d_{\mu\nu\kappa}{}^{\alpha} \nabla_{\rho} d_{\alpha}{}^{\kappa}{}_{\sigma}{}^{\rho}
 + 4 \Theta d_{\sigma}{}^{\kappa}{}_{[\mu} {}^{\alpha}\nabla_{|\rho|}d_{\nu]\alpha\kappa}{}^{\rho}
  \nonumber
\\
&&
+( R_{\sigma}{}^{\alpha} +\frac{1}{2}R\delta_{\sigma}{}^{\alpha})\nabla_{\rho}d_{\mu\nu\alpha}{}^{\rho}
 + f_{\mu\nu\sigma}(x; H,\nabla H, \nabla K)
\label{divergence_d}
 \;.
\end{eqnarray}
The validity of the last equality follows from \eq{ricci_riccired}, \eq{cwe4}, \eq{cwe5}, \eq{eqn_curvscalar} and \eq{d_bianchi_divergence}.
Note that  to establish  \eq{d_bianchi_divergence} one just needs the algebraic properties of $d_{\mu\nu\sigma}{}^{\rho}$ which are ensured by Lemma~\ref{lemma_properties_d}.

Next, let us derive a wave equation for $\Xi_{\mu\nu}$.
With \eq{cwe1}-\eq{cwe3}, \eq{cwe5}, \eq{ricci_riccired}, \eq{vanishingL} and \eq{eqn_curvscalar} the following relation is verified,
\begin{eqnarray}
 \Box_g\Xi_{\mu\nu}
 &\equiv&  \nabla_{\mu}\nabla_{\nu}\Box_g\Theta +2\nabla_{(\mu}R_{\nu)\kappa}\nabla^{\kappa}\Theta
 + 2 R_{\kappa(\mu}\nabla_{\nu)}\nabla^{\kappa}\Theta + 2R_{\sigma\mu\nu}{}^{\kappa}\nabla^{\sigma}\nabla_{\kappa}\Theta
 \nonumber
\\
 && - \nabla_{\kappa}R_{\mu\nu}\nabla^{\kappa}\Theta
 +  L_{\mu\nu}\Box_g \Theta + \Theta \Box_g L_{\mu\nu} +2 \nabla^{\sigma}\Theta \nabla_{\sigma} L_{\mu\nu}- g_{\mu\nu}\Box_g s
  \nonumber
\\
 &=&  2(2 L_{(\mu}{}^{\kappa}\delta_{\nu)}{}^{\sigma} -g_{\mu\nu}L^{\sigma\kappa}   - W_{\mu}{}^{\sigma}{}_{\nu}{}^{\kappa}) \Xi_{\sigma\kappa}
   + 2\Theta L_{\sigma\kappa}(W_{\mu}{}^{\sigma}{}_{\nu}{}^{\kappa}-\Theta d_{\mu}{}^{\sigma}{}_{\nu}{}^{\kappa})
  \nonumber
\\
 &&
  + 4\nabla_{(\mu}\Upsilon_{\nu)}  + \frac{1}{6} R\Xi_{\mu\nu}
  + f_{\mu\nu}(x;H,\nabla H,\nabla K)
 \label{conf3**_wave}
 \;,
\end{eqnarray}
where we have set
\begin{equation}
 \Upsilon_{\mu}\, := \, \nabla_{\mu} s + L_{\mu\nu}\nabla^{\nu} \Theta
 \;.
\end{equation}

Of course, we also need a wave equation for $\Upsilon_{\mu}$.
Using again \eq{cwe1}-\eq{cwe3}, \eq{cwe5} as well as \eq{ricci_riccired} and \eq{eqn_curvscalar} we find that
\begin{eqnarray}
 \Box_g\Upsilon_{\mu} &\equiv& \nabla_{\mu}\Box_g s + R_{\mu}{}^{\kappa}\nabla_{\kappa}s + \Box_g L_{\mu\nu}\nabla^{\nu}\Theta + L_{\mu}{}^{\nu}\nabla_{\nu}\Box_g\Theta
 + L_{\mu}{}^{\nu}R_{\nu}{}^{\kappa}\nabla_{\kappa}\Theta
 \nonumber
\\
 &&+2\nabla_{\sigma}L_{\mu\nu}\nabla^{\sigma}\nabla^{\nu}\Theta
 \nonumber
\\
 &=& 6L_{\mu}{}^{\kappa}\Upsilon_{\kappa}  + 2\Theta L^{\rho\kappa}\varkappa_{\rho\kappa\mu} +2\Xi_{\nu\sigma}\nabla^{\sigma}L_{\mu}{}^{\nu}
- \frac{1}{6}\Xi_{\mu}{}^{\nu} \nabla_{\nu} R
 \nonumber
\\
 && + f_{\mu}(x;H,\nabla H,\nabla K)
  \label{conf4**_wave}
 \;.
\end{eqnarray}

Finally, let us derive a wave equation which is satisfied by $\varkappa_{\mu\nu\sigma}\equiv \frac{1}{2} \zeta_{\mu\nu\sigma} - \nabla_{\kappa}\Theta d_{\nu\sigma\mu}{}^{\kappa}$.
The definition of the Weyl tensor \eq{riemann_ricci_weyl} together with the Bianchi identities yield
\begin{eqnarray*}
\frac{1}{2} \Box_g \zeta_{\mu\nu\sigma}
  &\equiv& 2\nabla_{[\sigma}\Box_g L_{\nu]\mu}
 - 2W_{\nu\sigma\kappa\rho}\nabla^{\rho}L_{\mu}{}^{\kappa}+4W_{\mu\kappa\rho[\sigma} \nabla^{\rho}L_{\nu]}{}^{\kappa}
  -  2R_{\kappa[\nu} \nabla_{\sigma]}L_{\mu}{}^{\kappa}
\\
 && + 2R_{\kappa[\sigma} \nabla_{|\mu|}L_{\nu]}{}^{\kappa} - 2R_{\mu[\sigma}  \nabla_{|\kappa|}L_{\nu]}{}^{\kappa}    -2R_{\rho\kappa} g_{\mu[\sigma} \nabla^{\rho}L_{\nu]}{}^{\kappa}   +  \frac{1}{6}R_g\zeta_{\mu\nu\sigma}
\\
 &&  + \frac{2}{3}R_g g_{\mu[\sigma}\nabla^{\kappa}L_{\nu]\kappa}  + 2L_{\mu}{}^{\kappa}\nabla_{[\nu}R_{\sigma]\kappa}
 +2  L_{\nu}{}^{\kappa}\nabla_{[\mu}R_{\kappa]\sigma} +  2  L_{\sigma}{}^{\kappa}\nabla_{[\kappa}R_{\mu]\nu}
\\
 &=& 2\zeta_{\mu\kappa[\sigma} L_{\nu]}{}^{\kappa}    + 3\zeta_{\alpha[\nu\sigma}  g_{\kappa]\mu} L^{\alpha\kappa}
 + 4  L_{\rho}{}^{\kappa}\nabla_{[\nu}(\Theta d_{\sigma]\kappa\mu}{}^{\rho}) + 2\Theta \zeta_{\alpha\kappa[\nu}d_{\sigma]}{}^{\kappa}{}_{\mu}{}^{\alpha}
\\
 &&  +4(W_{\mu}{}^{\rho}{}_{[\nu}{}^{\kappa} - \Theta d_{\mu}{}^{\rho}{}_{[\nu}{}^{\kappa}) \nabla_{|\kappa|} L_{\sigma]\rho}
 - \zeta_{\mu\alpha\kappa}W_{\nu}{}^{\alpha}{}_{\sigma}{}^{\kappa}
+ \frac{1}{3} L_{\mu[\nu}\nabla_{\sigma]}R
\\
 &&  + \frac{1}{6} (R_{\sigma\nu\mu}{}^{\kappa} + 2g_{\mu[\nu} L_{\sigma]}{}^{\kappa} )\nabla_{\kappa}R
+  \frac{1}{12}R_g\zeta_{\mu\nu\sigma}
    + f_{\mu\nu\sigma}(x; H, \nabla H,\nabla K)
 \;,
\end{eqnarray*}
where the last equality follows from  \eq{ricci_riccired}, \eq{cwe1}, \eq{cwe5} and \eq{eqn_curvscalar}.
We employ \eq{cwe3}-\eq{cwe5} and \eq{ricci_riccired} to deduce that
\begin{eqnarray*}
 \Box_ g (\nabla_{\kappa}\Theta d_{\nu\sigma\mu}{}^{\kappa})
 &\equiv &   d_{\nu\sigma\mu}{}^{\kappa}  (\nabla_{\kappa}\Box_g \Theta+  R_{\kappa}{}^{\rho}\nabla_{\rho}\Theta )
 +\nabla_{\kappa}\Theta  \Box_ g  d_{\nu\sigma\mu}{}^{\kappa}
 + 2 \nabla_{\alpha}\nabla_{\kappa}\Theta \nabla^{\alpha}d_{\nu\sigma\mu}{}^{\kappa}
\\
&=&   4 \Upsilon_{\kappa}d_{\nu\sigma\mu}{}^{\kappa}
 - 2 L_{\kappa}{}^{\rho}\nabla_{\rho}(\Theta  d_{\nu\sigma\mu}{}^{\kappa})
 + 2\Xi_{\lambda\kappa} \nabla^{\lambda}d_{\nu\sigma\mu}{}^{\kappa}
   +2s\nabla_{\kappa}d_{\nu\sigma\mu}{}^{\kappa}
\\
 &&   +\Theta  (\frac{1}{2}\zeta_{\mu\lambda}{}^{\alpha} - \varkappa_{\mu\lambda}{}^{\alpha})d_{\nu\sigma\alpha}{}^{\lambda}
   - \Theta  (2\zeta_{\alpha\lambda[\sigma} - 4\varkappa_{\alpha\lambda[\sigma})d_{\nu]}{}^{\lambda}{}_{\mu} {}^{\alpha}
\\
 && + \frac{1}{2}R d_{\nu\sigma\mu}{}^{\kappa}\nabla_{\kappa}\Theta -\frac{1}{6} \Theta  d_{\nu\sigma\mu}{}^{\kappa}  \nabla_{\kappa}   R
 + f_{\mu\nu\sigma}(x; H, \nabla H)
 \;.
\end{eqnarray*}
With \eq{eqn_curvscalar} we are led to
\begin{eqnarray}
  \Box_g \varkappa_{\mu\nu\sigma}
 &=&   4\nabla^{\beta} \Theta \big\{ g_{\beta[\nu}  d_{\sigma]\kappa\mu}{}^{\alpha} L_{\alpha}{}^{\kappa} - g_{\mu[\nu}  d_{\sigma]\kappa\beta}{}^{\alpha} L_{\alpha}{}^{\kappa}
- d_{\mu\beta\kappa[\nu} L_{\sigma]}{}^{\kappa} \big.
\nonumber
\\
 && \hspace{4em}  -\big.  d_{\nu\sigma\kappa[\mu} L_{\beta]}{}^{\kappa}  - \frac{1}{12}R d_{\nu\sigma\mu\beta}  \big\}
  \nonumber
\\
 &&  +6\Theta   L_{\rho}{}^{\kappa} \nabla_{[\nu}  d_{\sigma\kappa]\mu}{}^{\rho}
 - 2\Xi_{\lambda\kappa} \nabla^{\lambda}d_{\nu\sigma\mu}{}^{\kappa} -4\Upsilon_{\kappa}d_{\nu\sigma\mu}{}^{\kappa}
  \nonumber
\\
 &&  +4(W_{\mu}{}^{\rho}{}_{[\nu}{}^{\kappa}- \Theta d_{\mu}{}^{\rho}{}_{[\nu}{}^{\kappa}) \nabla_{|\kappa|} L_{\sigma]\rho}
  -  \frac{1}{2}\zeta_{\mu\kappa}{}^{\alpha}(W_{\nu\sigma\alpha}{}^{\kappa} - \Theta d_{\nu\sigma\alpha}{}^{\kappa}  )
  \nonumber
\\
 && -  4\varkappa_{\mu\kappa[\nu}  L_{\sigma]}{}^{\kappa} + 6\varkappa_{\alpha[\nu\sigma}  g_{\kappa]\mu} L^{\alpha\kappa}
 +\Theta \varkappa_{\mu\lambda}{}^{\alpha}d_{\nu\sigma\alpha}{}^{\lambda}  - 4\Theta \varkappa_{\alpha\lambda[\sigma}d_{\nu]}{}^{\lambda}{}_{\mu} {}^{\alpha}
  \nonumber
\\
&& - 2s \nabla_{\kappa}d_{\nu\sigma\mu}{}^{\kappa}
-\frac{1}{6}(W_{\nu\sigma\mu}{}^{\kappa}  - \Theta  d_{\nu\sigma\mu}{}^{\kappa} ) \nabla_{\kappa}   R
  + \frac{1}{6}R_g\varkappa_{\mu\nu\sigma}
 \nonumber
\\
 &&   +f_{\mu\nu\sigma}(x; H, \nabla H,\nabla K)
\label{prov_wave_xi}
 \;.
\end{eqnarray}
The term in braces needs to be eliminated.
To this end let us consider the expression (we use the implications of Lemma~\ref{lemma_properties_d})
\begin{eqnarray*}
3 \nabla^{\lambda}\nabla_{[\lambda}d_{\mu\nu]\sigma\rho} &=&\Box_g d_{\mu\nu\sigma\rho}
  - \nabla_{\nu}\nabla_{\lambda}d_{\sigma\rho\mu}{}^{\lambda}   + \nabla_{\mu}\nabla_{\lambda}d_{\sigma\rho\nu}{}^{\lambda}
\\
 && - W_{\mu\nu\lambda}{}^{\kappa}d_{\sigma\rho\kappa}{}^{\lambda} +2 W_{\lambda\nu[\sigma}{}^{\kappa}d_{\rho]\kappa\mu}{}^{\lambda}
 -2 W_{\lambda\mu[\sigma}{}^{\kappa}d_{\rho]\kappa\nu}{}^{\lambda}
\\
 &&- d_{\sigma\rho[\mu}{}^{\kappa}R_{\nu]\kappa}-d_{\mu\nu[\sigma}{}^{\lambda} R_{\rho]\lambda}
 + g_{\mu[\sigma}d_{\rho]\kappa\nu\lambda}R^{\lambda\kappa} - g_{\nu[\sigma}d_{\rho]\kappa\mu\lambda}R^{\lambda\kappa}
\;.
\end{eqnarray*}
We take \eq{d_bianchi_divergence}, \eq{ricci_riccired}, \eq{cwe4} and \eq{cwe5} into account to rewrite this equation as
\begin{eqnarray}
&&  \hspace{-4em}2g_{\nu[\sigma}d_{\rho]\kappa\mu}{}^{\alpha}L_{\alpha}{}^{\kappa}
  -2g_{\mu[\sigma}d_{\rho]\kappa\nu}{}^{\alpha}L_{\alpha}{}^{\kappa}
 - 2d_{\mu\nu\kappa[\sigma} L_{\rho]}{}^{\kappa} -  2d_{\sigma\rho\kappa[\mu}L_{\nu]}{}^{\kappa} - \frac{1}{6}R d_{\mu\nu\sigma\rho}
 \nonumber
\\
 &\equiv& 2 (W_{\sigma\kappa[\mu} {}^{\alpha}- \Theta d_{\sigma\kappa[\mu} {}^{\alpha})d_{\nu]\alpha\rho}{}^{\kappa}
   -2( W_{[\mu|\alpha\rho}{}^{\kappa} - \Theta d_{[\mu|\alpha\rho}{}^{\kappa}  )d_{\sigma\kappa|\nu]} {}^{\alpha}
 \nonumber
\\
   &&- (W_{\mu\nu\kappa}{}^{\alpha} - \Theta d_{\mu\nu\kappa}{}^{\alpha})d_{\sigma\rho\alpha}{}^{\kappa}
   + 2\nabla_{[\mu}\nabla_{|\lambda}d_{\sigma\rho|\nu]}{}^{\lambda}
 \nonumber
\\
 &&- \frac{1}{2}  \epsilon_{\lambda\mu\nu}{}^{\kappa} \epsilon_{\sigma\rho}{}^{\beta\gamma} \nabla^{\lambda}\nabla_{\alpha}d_{\beta\gamma\kappa}{}^{\alpha}
 + f_{\mu\nu\sigma\rho}(x; H, \nabla H)
\label{rln_Schouten_resc_Weyl}
\;.
\end{eqnarray}
Combining \eq{rln_Schouten_resc_Weyl} with \eq{prov_wave_xi} and  \eq{d_bianchi_divergence}  yields a wave equation for $\varkappa_{\mu\nu\sigma}$,
\begin{eqnarray}
  \Box_g \varkappa_{\mu\nu\sigma}
  &=& 4 \nabla^{\beta} \Theta [ ( W_{\nu\kappa[\mu}{}^{\alpha} - \Theta d_{\nu\kappa[\mu}{}^{\alpha}  ) d_{\beta]\alpha\sigma}{}^{\kappa}
   -(W_{\sigma\kappa[\mu} {}^{\alpha}- \Theta d_{\sigma\kappa[\mu} {}^{\alpha})  d_{\beta]\alpha\nu}{}^{\kappa}]
 \nonumber
\\
   &&+2 (W_{\mu\beta\kappa}{}^{\alpha} - \Theta d_{\mu\beta\kappa}{}^{\alpha}) \nabla^{\beta} \Theta d_{\sigma\nu\alpha}{}^{\kappa}
  + 4 \nabla^{\beta} \Theta \nabla_{[\beta}(\nabla_{\lambda}d_{|\sigma\nu|\mu]}{}^{\lambda})
 \nonumber
\\
 && +   \epsilon_{\lambda\mu\beta}{}^{\kappa} \epsilon_{\sigma\nu}{}^{\delta\gamma}  \nabla^{\beta} \Theta \nabla^{\lambda}(\nabla_{\alpha}d_{\delta\gamma\kappa}{}^{\alpha})
 +\Theta   L_{\rho}{}^{\kappa} \epsilon_{\sigma\kappa\nu}{}^{\delta} \epsilon_{\mu}{}^{\rho\beta\gamma} \nabla_{\alpha}d_{\beta\gamma\delta}{}^{\alpha}
  \nonumber
\\
 &&   - 2\Xi_{\lambda\kappa} \nabla^{\lambda}d_{\nu\sigma\mu}{}^{\kappa} -4\Upsilon_{\kappa}d_{\nu\sigma\mu}{}^{\kappa}
 +4(W_{\mu}{}^{\rho}{}_{[\nu}{}^{\kappa}- \Theta d_{\mu}{}^{\rho}{}_{[\nu}{}^{\kappa}) \nabla_{|\kappa|} L_{\sigma]\rho}
  \nonumber
\\
 &&    -  4\varkappa_{\mu\kappa[\nu}  L_{\sigma]}{}^{\kappa} + 6\varkappa_{\alpha[\nu\sigma}  g_{\kappa]\mu} L^{\alpha\kappa}
 + \frac{1}{2} \zeta_{\mu\kappa}{}^{\alpha}(W_{\nu\sigma\alpha}{}^{\kappa}-\Theta d_{\nu\sigma\alpha}{}^{\kappa} )
  \nonumber
\\
 &&   +\Theta \varkappa_{\mu\lambda}{}^{\alpha}d_{\nu\sigma\alpha}{}^{\lambda}
 +  4\Theta  \varkappa_{\alpha\lambda[\nu}d_{\sigma]}{}^{\lambda}{}_{\mu}{}^{\alpha}
  -\frac{1}{6}(W_{\nu\sigma\mu}{}^{\kappa}  - \Theta  d_{\nu\sigma\mu}{}^{\kappa} ) \nabla_{\kappa}   R
  \nonumber
\\
 &&  - 2s \nabla_{\kappa}d_{\nu\sigma\mu}{}^{\kappa}
  + \frac{1}{6}R_g\varkappa_{\mu\nu\sigma}
+f_{\mu\nu\sigma}(x; H, \nabla H,\nabla K)
 \label{wave_zeta-xi}
 \;.
\end{eqnarray}

The equations \eq{wave_H*}, \eq{wave_DiffH}, \eq{wave_zetavec}, \eq{wave_W-C}, \eq{divergence_d}, \eq{conf3**_wave}, \eq{conf4**_wave} and \eq{wave_zeta-xi}
form a closed, linear, homogeneous system of linear wave equations satisfied by the fields $H^{\sigma}$, $K_{\mu\nu}$, $\zeta_{\mu}$, $W_{\mu\nu\sigma\rho} -\Theta d_{\mu\nu\sigma\rho} $, $\nabla_{\rho}d_{\mu\nu\sigma}{}^{\rho}$, $\Xi_{\mu\nu}$, $\Upsilon_{\mu}$ and $\varkappa_{\mu\nu\sigma}$,
with all other quantities regarded as being given.
An application of standard uniqueness results, cf.\ e.g.\  \cite{friedlander}, establishes that
all the fields vanish, supposing that this is initially the case.
In particular this guarantees consistency with the gauge condition, i.e.\ $H^{\sigma}=0$ and, by \eq{ricci_schouten},  $R_g=R$, for solutions of the CWE.
In fact we have proven more, and that will be of importance in the next section.

\subsection{Equivalence issue between the CWE and the MCFE}

Recall the CWE \eq{cwe1}-\eq{cwe5} and the MCFE \eq{conf1}-\eq{conf6}.
Let us tackle the equivalence issue between them.
A look at the derivation of the CWE reveals that any solution of the MCFE which satisfies the gauge condition $H^{\sigma}=0$
will be a solution of the CWE with gauge source function $R=R_g$.
The other direction is the more interesting albeit more involved one.
We therefore devote ourselves subsequently to the issue whether (or rather under which conditions) a solution of the CWE is also  a solution
of the MCFE.
We shall demonstrate that a solution of the CWE is a solution of the MCFE supposing that it satisfies certain relations on the initial surface.
In fact, most of the work has already been done in the previous section.

We have the following intermediate result; we emphasize that the conformal factor is allowed to have zeros, or vanish, on the initial surface:
\begin{theorem}
\label{inter-thm}
Assume we have been given data ($\mathring g_{\mu\nu}$, $\mathring s$, $\mathring \Theta$, $\mathring L_{\mu\nu}$, $\mathring d_{\mu\nu\sigma}{}^{\rho}$) on a characteristic initial surface $S$ (for definiteness we think either of two transversally intersecting null hypersurfaces or a light-cone) and a gauge source function $R$, such that $\mathring g_{\mu\nu}$ is the restriction to $S$ of a Lorentzian metric, $\mathring L_{\mu\nu}$ is symmetric, $\mathring L_{\mu}{}^{\mu} \equiv  \mathring L = \overline R/6$, and such that $\mathring d_{\mu\nu\sigma}{}^{\rho}$   satisfies all the algebraic properties of the Weyl tensor (cf.\ the assumptions of Lemma~\ref{lemma_properties_d}).
Suppose further that there exists a solution ($g_{\mu\nu}$, $s$, $\Theta$, $L_{\mu\nu}$, $d_{\mu\nu\sigma}{}^{\rho}$)  of the CWE \eq{cwe1}-\eq{cwe5} with gauge source function $R$ which induces the above data on $S$ and fulfills the following conditions:
\begin{enumerate}
 \item The MCFE \eq{conf1}-\eq{conf4} are fulfilled on $S$;
 \item  equation \eq{conf5} holds at one point on $S$;
\item the Weyl tensor $W_{\mu\nu\sigma}{}^{\rho}[g]$  coincides on $S$ with $\mathring \Theta \mathring d_{\mu\nu\sigma}{}^{\rho}$;
 \item  the wave-gauge vector $H^{\sigma}$ and its covariant derivative $K_{\mu}{}^{\sigma}\equiv \nabla_{\mu}H^{\sigma}$ vanish on $S$;
\item the covector  field $\zeta_{\mu}\equiv -4(\nabla_{\nu}L_{\mu}{}^{\nu} - \frac{1}{6}\nabla_{\mu}R)$  vanishes on $S$.
\end{enumerate}
Then
\begin{enumerate}
 \item[a)] $H^{\sigma}=0$ and $R_g=R$;
\item[b)] $L_{\mu\nu}$ is the Schouten  tensor of $g_{\mu\nu}$;
\item[c)] $\Theta d_{\mu\nu\sigma}{}^{\rho}$ is the Weyl tensor of $g_{\mu\nu}$;
\item[d)] ($g_{\mu\nu}$, $s$, $\Theta$, $L_{\mu\nu}$, $d_{\mu\nu\sigma}{}^{\rho}$)  solves the MCFE \eq{conf1}-\eq{conf6} with $H^{\sigma}=0$ and $R_g=R$.
\end{enumerate}
The conditions 1-5 are  necessary for d) to be true.
\end{theorem}
\begin{proof}
 The conditions 1 and 3-5 make sure that the fields $H^{\sigma}$, $K_{\mu\nu}$, $\zeta_{\mu}$, $W_{\mu\nu\sigma\rho} -\Theta d_{\mu\nu\sigma\rho} $, $\nabla_{\rho}d_{\mu\nu\sigma}{}^{\rho}$, $\Xi_{\mu\nu}$, $\Upsilon_{\mu}$ and $\varkappa_{\mu\nu\sigma}$ vanish on $S$.
In the previous section we have seen that they provide a solution of the
 closed, linear, homogeneous system of wave equations \eq{wave_H*}, \eq{wave_DiffH}, \eq{wave_zetavec}, \eq{wave_W-C}, \eq{divergence_d}, \eq{conf3**_wave}, \eq{conf4**_wave} and \eq{wave_zeta-xi}, so  that all these fields  need to vanish identically.
In particular that implies $H^{\sigma}=0$, that $\Theta d_{\mu\nu\sigma}{}^{\rho}$ is the Weyl tensor of $g_{\mu\nu}$, and that \eq{conf1}-\eq{conf4} hold.
 The vanishing of $H^{\sigma}$ guarantees that the Ricci tensor coincides with the reduced Ricci tensor and  by \eq{eqn_curvscalar} that $R$ is the curvature scalar $R_g$ of $g_{\mu\nu}$. Equation \eq{cwe5} then tells us that $L_{\mu\nu}$  is the Schouten tensor. Hence \eq{conf6}  is an identity and automatically satisfied.
To establish \eq{conf5}, it suffices to check that it is satisfied at one point, which is ensured by condition~2.
\qed
\end{proof}

In the following we shall investigate to what extent the conditions~1-5 are satisfied if the fields  $\mathring g_{\mu\nu}$,
$\mathring L_{\mu\nu}$, $\mathring d_{\mu\nu\sigma}{}^{\rho}$,  $\mathring \Theta$ and $\mathring s$
are constructed as solutions of the constraint equations induced by the MCFE on the initial surface.

\section{Constraint equations induced by the MCFE on the $\mathbf{C_{i^-}}$-cone}
\label{sec_constraint_equations}
\label{adapted_null_coord}

\subsection{Adapted null coordinates and another gauge freedom}
The aim of this section is to derive the set of constraint equations  induced by the MCFE,
\begin{eqnarray}
 && \nabla_{\rho} d_{\mu\nu\sigma}{}^{\rho} =0
  \;,
 \label{mconf1}
\\
 && \nabla_{\mu}L_{\nu\sigma} - \nabla_{\nu}L_{\mu\sigma} = \nabla_{\rho}\Theta \, d_{\nu\mu\sigma}{}^{\rho}
  \;,
 \label{mconf2}
\\
 && \nabla_{\mu}\nabla_{\nu}\Theta = -\Theta L_{\mu\nu} +  s  g_{\mu\nu}
  \;,
 \label{mconf3}
\\
 && \nabla_{\mu}  s = - L_{\mu\nu}\nabla^{\nu}\Theta
 \;,
 \label{mconf4}
\\
 && 2\Theta  s -  \nabla_{\mu}\Theta\nabla^{\mu}\Theta = 0
  \;,
 \label{mconf5}
\\
 &&  R_{\mu\nu\sigma}{}^{\kappa}[ g] = \Theta  d_{\mu\nu\sigma}{}^{\kappa} + 2\left( g_{\sigma[\mu} L_{\nu]}{}^{\kappa}  - \delta_{[\mu}{}^{\kappa} L_{\nu]\sigma} \right)
 \label{mconf6}
\;,
\end{eqnarray}
on the initial surface $S$, where we assume henceforth
\begin{equation}
\lambda=0\;.
\label{eq_vanishing_lambda}
\end{equation}
 By \textit{constraint equations} we mean intrinsic equations on the initial surface which determine the fields
$ g_{\mu\nu}|_S$, $ L_{\mu\nu}|_S$, $  d_{\mu\nu\sigma}{}^{\rho}|_S$,  $  \Theta|_S$ and $ s|_S$
starting from suitable free ``reduced'' data.
We shall do this in adapted null coordinates
and imposing a generalized wave-map gauge condition.
To avoid too many case distinctions we shall derive them in the case where the initial surface is the light-cone $C_{i^-}$ on which the conformal factor $\Theta$ vanishes (this requires  \eq{eq_vanishing_lambda}, cf.\ \eq{conf5} evaluated on $C_{i^-}$), which is completely sufficient for our purposes.

\textit{Adapted null coordinates}  $(u,r, x^A)$ are defined in such a way that  $\{x^0\equiv u=0\}=\Scri^- \equiv C_{i^-}\setminus \{i^-\}$, $x^1\equiv r>0$ parameterizes
the null rays emanating from $i^-$, and $x^A$, $A=2,3$, are local coordinates on the level sets $\{r=\mathrm{const},u=0\}\cong S^2$ (note that these coordinates are singular at the tip, see~\cite{CCM2} for more~details).

First we shall sketch how the constraint equations are obtained in a generalized wave-map gauge with arbitrary gauge functions.
We shall write them down explicitly in a specific gauge afterwards.

We use the same notation as in~\cite{CCM2}, i.e. $\nu_0:=\overline g_{01}$, $\nu_A:=\overline g_{0A}$. The function $\chi_A{}^B:=\frac{1}{2}\overline g^{BC}\partial_1 \overline g_{AC}$ denotes the \textit{null second fundamental form}, the function $\tau$, which describes the \textit{expansion} of the cone, its trace, and the \textit{shear tensor} $\sigma_A{}^B$ its traceless part.
The symbols $\tilde\nabla_A$, $\tilde\Gamma^C_{AB}$ and $\tilde R_{AB}$ refer to the $r$-dependent Riemannian  metric $\tilde g:= \overline g_{AB}\mathrm{d}x^A\mathrm{d}x^B$.


The equation \eq{eqn_gAB} below together with regularity conditions at the tip of the cone imply that
$\tilde g$ is conformal to the   the standard metric $s_{AB}$ on the 2-sphere $S^2$.
 It therefore makes sense  to take as reduced data the $\tilde g$ -trace-free part of $L_{AB}$ on $C_{i^-}$ (which coincides with its $s$-trace-free part). It will be denoted by $\breve{\overline L}_{AB}=:\omega_{AB}$.

The field $\omega_{AB}$ is an $r$-dependent tensor on $S^2$.
Here and in what follows  $\breve{.}$ denotes the $\tilde g$-trace-free part of the corresponding  2-tensor on $S^2$.
As before, overlining is used to indicate restriction to the initial surface.
The gauge degrees of freedom are comprised by $R$, $W^{\lambda}$, $\overline s$  (cf.\ Sections~\ref{conformal_field_equations} and \ref{reduced_cwe}) and $\kappa$.
The function $\kappa$ is given by
\begin{equation}
 \kappa:= \nu^0\partial_1\nu_0 -\frac{1}{2} \tau -\frac{1}{2}\nu_0 (\overline g^{\mu\nu}\overline{\hat \Gamma}{}^0_{\mu\nu} + \overline{ W}{}^0)
 \label{dfn_kappa}
\;,
\end{equation}
where  $\nu^0:= \ol g^{01}=(\nu_0)^{-1}$.
It reflects the freedom to parameterize the null geodesics
generating the initial surface~\cite{CCM2}; the choice $\kappa=0$ corresponds to an affine parameterization.

\subsection{Constraint equations in a generalized wave-map gauge}
\label{constraints_generalizedw}

We show that, in the case where the initial surface is $C_{i^-}$,
 the constraint equations form a  hierarchical system of  algebraic equations and ODEs along the generators of  $C_{i^-}$.
In doing so, we merely consider those gauge choices $W^{\lambda}$ which depend just upon the coordinates and none of the fields appearing in the CWE (cf.\ footnote~\ref{footnote_W_coord}).
To derive the constraint equations we assume we have been given a smooth solution of the
MCFE in a generalized wave-map gauge $H^{\sigma}=0$, smoothly extendable through $C_{i^-}$
We then evaluate the MCFE on $C_{i^-}$ and eliminate the transverse derivatives.
For this we shall assume that the solution satisfies $s_{i^-}\ne 0$, which implies that $\overline s^{-1}$ and $(\overline{\partial_0\Theta})^{-1}$ exist near $i^-$
(the existence of the latter one follows e.g.\ from \eq{expression_s} below).
The function $\tau^{-1}$ needs to exist  anyway close to $i^-$~\cite{CCM2}.
It should be emphasized that, on a light-cone, the initial data for the ODEs cannot be specified freely but follow from regularity conditions at the vertex.
For sufficiently regular gauges the behaviour of the relevant fields near the vertex is computed in \cite{CCM2}.
When stating this behaviour below we shall always tacitly assume that the gauge is sufficiently regular.
\label{suff_reg_gauge}

In the following we shall frequently make use of the formulae (A.8)-(A.25) in~\cite{CCM2} for the Christoffel symbols computed in adapted null coordinates
on a cone.

We consider \eq{mconf3} for $(\mu\nu)=(10),(AB)$ on $C_{i^-}$, where we take the $\overline g{}^{AB}$-trace of the latter equation,
\begin{eqnarray}
 \partial_1\overline{\partial_0\Theta} + (\kappa - \nu^0\partial_1\nu_0)\overline{\partial_0 \Theta} &=& \nu_0\overline s \;,
 \label{eqn_transTheta}
\\
 \overline s &=&\frac{1}{2}\tau\nu^0\overline{\partial_0 \Theta}
 \label{expression_s}
\end{eqnarray}
(note that $\overline H{}^{0}=0$ implies $\kappa = \overline \Gamma{}^1_{11}$~\cite{CCM2}).
Differentiating \eq{expression_s} and inserting the result  into \eq{eqn_transTheta} we obtain an equation for $\tau$,
\begin{eqnarray}
 \partial_1\tau - (\kappa +\partial_1\log|\overline s|)\tau + \frac{1}{2}\tau^2 =0
  \;.
 \label{constraint_tau}
\end{eqnarray}
The boundary behaviour 
is given by $\tau=2r^{-1} +O(r)$~\cite{CCM2}.

Due to our assumption $s_{i^-}\ne 0$ the (AB)-component of \eq{mconf3}, together with \eq{expression_s}, provides an equation for $\overline g_{AB}$
(at least sufficiently close to the vertex),%
\begin{eqnarray}
 \overline s(\partial_1\overline g_{AB} -\tau \overline  g_{AB})=0 \quad \Longleftrightarrow \quad \sigma_{AB}=0
  \label{eqn_gAB}
   \;.
\end{eqnarray}
The boundary condition 
is   $\overline g_{AB}=r^2 s_{AB} + O(r^4)$~\cite{CCM2}, with $s_{AB}$ the round sphere metric.

Using the definition of $L_{\mu\nu}$, which can be recovered from \eq{mconf6}, as well as \eq{eqn_gAB}, we find that
\begin{equation}
 \overline L_{11} \equiv-\frac{1}{2}\big(\partial_1\tau -\overline \Gamma{}^1_{11}\tau + \chi_A{}^B\chi_B{}^A\big)
  = -\frac{1}{2}\partial_1\tau +\frac{1}{2}\kappa \tau - \frac{1}{4}\tau^2
   \label{eqn_L11}
   \;.
\end{equation}

The gauge condition $\overline H{}^0 =0$ provides an equation for $\nu_0$,%
\footnote{
\label{footnote_W_coord}
Recall that we assume $W^{\lambda}$ to depend just upon the coordinates, otherwise one would have to be careful here and specify upon which components of which fields $W^{\lambda}$ is allowed to depend in order to get the hierarchical system we are about to derive.}
\begin{equation}
  \partial_1\nu^0 + \nu^0(\frac{1}{2}\tau + \kappa) + \frac{1}{2}\overline V{}^0 =0\;.
  \label{gauge0}
\end{equation}
%
Here we have set
\begin{eqnarray*}
   V^{\lambda} :=g^{\mu\nu} \hat\Gamma^{\lambda}_{\mu\nu} +  W^{\lambda}
    \;.
\end{eqnarray*}
%
The boundary condition is $\nu_0=1+O(r^2)$~\cite{CCM2}.
Equation \eq{expression_s} then determines $\overline{\partial_0\Theta}$.
%
The function $\overline{\partial_0 g_{11}}$ is computed from $\kappa=\overline \Gamma{}^1_{11}$, 
\begin{equation}
 \overline{\partial_0 g_{11}} = 2\partial_1\nu_0 - 2\nu_0\kappa
 \;.
\end{equation}
We remark that the values of certain transverse derivatives are needed on the way to derive the constraint equations.
As a matter of course the constraint equations themselves will not involve any transverse derivatives, for they are not part of the
characteristic  initial data for the CWE.

Let us introduce the field
\begin{eqnarray}
 \xi_A \,:=\,-2\nu^0\partial_1\nu_A + 4\nu^0\nu_B\chi_A{}^B + \nu_A\overline V{}^0 + \overline g_{AB}\overline V{}^B - \overline g_{AD} \overline g^{BC}\tilde\Gamma^D_{BC}
  \;.
    \label{eqn_xiA}
\end{eqnarray}
%
In a generalized wave-map gauge we have
\cite{CCM2}
\begin{eqnarray}
 \xi_A \,=\, -2\overline\Gamma{}^1_{1A}
 \label{relation_xi_Gamma}
  \;.
\end{eqnarray}
Invoking \eq{expression_s} and \eq{eqn_gAB},  equation \eq{mconf3} with $(\mu\nu)=(0A)$ can  be written as an equation for $\xi_A$,
\begin{equation}
 \xi_A \,=\, 2\partial_A\log|\overline{\partial_0\Theta}| -2\nu^0\partial_A\nu_0
\label{relation_xi_Theta}
 \;.
\end{equation}
The definition of $\xi_A$ can then be employed to compute $\nu_A$,
\begin{equation}
 \nu^0\partial_1\nu_A - \tau\nu^0\nu_A - \frac{1}{2}\nu_A \overline V{}^0 - \frac{1}{2}\overline g_{AB} \overline V{}^B + \frac{1}{2}\overline g_{AD} \overline g^{BC}\tilde\Gamma^D_{BC} + \frac{1}{2}\xi_A =0
 \label{dfn_xiA}
 \;.
\end{equation}
%
The boundary condition is given by $\nu_A =O(r^3)$~\cite{CCM2}.
The equation $\xi_A = -2\overline\Gamma{}^1_{1A}$ then determines $\overline{\partial_0 g_{1A}}$ algebraically,
\begin{equation}
 \overline{\partial_0 g_{1A}} = (\partial_A+\xi_A)\nu_0 + (\partial_1 - \tau)\nu_A
 \;.
\end{equation}
From  \eq{eqn_transTheta}, \eq{expression_s} and \eq{relation_xi_Theta} we obtain the relation
\begin{equation*}
 \partial_1\xi_A \,=\, \partial_A(\tau - 2\kappa)
 \;,
\end{equation*}
which yields
\begin{eqnarray}
 \overline L_{1A} &\equiv& \frac{1}{2}(\partial_1 +\tau)\overline\Gamma{}^1_{1A} + \frac{1}{2}\tilde\nabla_B\chi_A{}^B
  - \frac{1}{2}\partial_A\overline \Gamma{}^1_{11} - \frac{1}{2}\partial_A\tau
   \nonumber
\\
 &=& -\frac{1}{4}\tau\xi_A - \frac{1}{2}\partial_A\tau
  \label{eqn_L1A}
  \;.
\end{eqnarray}

We define the function
\begin{equation}
 \zeta \,:=\, 2(\partial_1 +\kappa +\frac{1}{2}\tau )\overline g^{11}
   + 2\overline V{}^1
    \label{dfn_zeta}
    \;.
\end{equation}
For a solution which satisfies the generalized wave-map gauge condition $H^{\sigma}=0$ it holds~\cite{CCM2} that
\begin{equation}
 \zeta \,=\, 2\overline g^{AB} \overline \Gamma{}^1_{AB} + \tau\overline g{}^{11}
  \label{zeta_H_gauge}
  \;.
\end{equation}
We find that
\begin{eqnarray}
  \overline g^{AB} \overline R_{ACB}{}^C &\equiv& \tilde R -\frac{1}{2} \overline g^{1A}\partial_A\tau
 + \tau \overline g^{AB}\overline \Gamma{}^1_{AB}   + \frac{1}{2}\tau\overline g^{1A}\overline \Gamma{}^1_{1A}    +  \frac{1}{2} \tau^2 \overline g{}^{11}
 \nonumber
\\
 &=& \tilde R -\frac{1}{2} \overline g^{1A}(\partial_A+\frac{1}{2}\xi_A)\tau
 + \frac{1}{2}\tau\zeta
 \;.
 \label{Riemann_one}
\end{eqnarray}
On the other hand, the $\overline g^{AB} \overline R_{ACB}{}^C$-part of \eq{mconf6} yields (we set $\xi^A:=\overline g^{AB}\xi_B$)
\begin{eqnarray}
   \overline g^{AB} \overline R_{ACB}{}^C  &=& \overline g^{1A}\overline L_{1A}  +  2\overline g^{AB}\overline L_{AB}
 \nonumber
\\
 &=&  (\tilde\nabla^A- \frac{1}{2}\xi^A -\frac{1}{4}\tau \ol g^{1A} )\xi_A - \frac{1}{2}\ol g^{1A}\partial_A\tau + (\partial_1 + \tau + \kappa)\zeta
\nonumber
\\
 && + \tilde R - \frac{1}{3}   R
\;,
 \label{Riemann_two}
\end{eqnarray}
where we took into account that
\begin{eqnarray}
  2\overline g^{AB}\overline L_{AB}   &\equiv & (\partial_1 + \tau + \kappa)\zeta   + (\tilde\nabla_A- \frac{1}{2}\xi_A)\xi^A + \tilde R  - \frac{1}{3}   \overline R
 \;.
  \label{eqn_gABLAB}
\end{eqnarray}
Combining \eq{Riemann_one} and \eq{Riemann_two}, we end up with an equation for $\zeta$,
\begin{eqnarray}
 (\partial_1 + \frac{1}{2}\tau + \kappa)\zeta   + (\tilde\nabla_A- \frac{1}{2}\xi_A)\xi^A - \frac{1}{3}  \overline  R &=& 0
 \;,
\label{ODE_zeta}
\end{eqnarray}
where the boundary condition 
is $\zeta + 2r^{-1}=O(1)$.
Then \eq{eqn_gABLAB} becomes
\begin{eqnarray}
  \overline g^{AB}\overline L_{AB} &=& \frac{1}{4}\tau\zeta  + \frac{1}{2}\tilde R
 \;.
\end{eqnarray}

The definition \eq{dfn_zeta} of $\zeta$ can be employed to compute $\overline g_{00}$,
since  $ \overline g_{00} = \ol g^{AB}\nu_A\nu_B -(\nu_0)^2\overline g^{11} $.
The boundary condition is \cite{CCM2} $\overline g^{11}=1 + O(r^2)$.
The equation $\zeta = 2\overline g^{AB} \overline \Gamma{}^1_{AB} + \tau\overline g^{11}$ can then be read as an equation for $\overline g^{AB}\overline {\partial_0 g_{AB}}$,
\begin{equation}
 \overline g^{AB}\overline{\partial_0 g_{AB}} = 2 \tilde\nabla^A\nu_A - \nu_0(\tau \overline g^{11} +\zeta)
 \;.
\end{equation}

An expression for  $\overline L_{01}$ follows from the relation $g^{\mu\nu}L_{\mu\nu} = \frac{1}{6}R$,
which yields
\begin{eqnarray}
  \overline L_{01}
 &=& -\frac{1}{2}\nu^A(\partial_A+ \frac{1}{2}\xi_A)\tau+ \frac{1}{4}\nu_0\overline g^{11}[\partial_1\tau -\kappa\tau + \frac{1}{2}\tau^2]
 \nonumber
\\
 && - \frac{1}{8}\nu_0(\tau\zeta  + 2\tilde R)  + \frac{1}{12}\nu_0 \overline R
 \label{eqn_L01}
 \;,
\end{eqnarray}
where $\nu^A:=\ol g^{AB} \nu_B$.
On the other hand we have
\begin{eqnarray*}
 2\overline L_{01}&\equiv &\overline R_{01} - \frac{1}{6}\nu_0\overline R
\\
 &\equiv& \overline{\partial_0\Gamma^0_{01}} - \partial_1\overline \Gamma{}^0_{00} + (\tilde{\nabla}_A + \frac{1}{2}\tau  \nu^0\nu_A)\overline \Gamma{}^A_{01} +(\nu^0\partial_1 \nu_0-\kappa + \tau) \overline\Gamma{}^1_{01}
\\
 && - (\partial_1-\nu^0\partial_1\nu_0+\kappa  + \frac{1}{2}\tau)\overline\Gamma{}^A_{0A}
  - \frac{1}{6}\nu_0\overline R
\;.
\end{eqnarray*}
Combining this with the gauge condition $\overline{\partial_0 H^0}= 0$ one determines
$\overline {\partial_0 g_{01}}$ and $\overline{\partial^2_{00}g_{11}}$,
with boundary condition  $\overline{\partial_0 g_{01}}=O(r)$ \cite{CCM2}.

Note that up to this stage the initial data $\omega_{AB}$ have not entered yet, i.e.\ all the field components computed so far have a pure gauge-character.
Note further that $(\overline{\partial_0 g_{AB}})\breve{}$
can be computed in terms of $\omega_{AB}\equiv \breve{\overline{L}}_{AB} = \frac{1}{2}\breve{\overline R}_{AB}$ and those quantities computed so far. (Recall that $(\overline{\partial_0 g_{AB}})\breve{}$ denotes the trace-free part of $ \overline{\partial_0 g_{AB}} $ with respect to $\overline g_{AB}$, and note that $(.)\breve{}$ always refers to the two free angular indices.)

Equation \eq{mconf6} with $(\mu\nu\sigma\kappa)=(0ABC)$, contracted with $\overline g{}^{AB}$, gives an equation for $\overline L_{0A}$,
\begin{eqnarray}
 \overline L_{0A}  &=& -\overline g_{AC}\overline g^{BD}(\partial_B\overline\Gamma{}^C_{0D} - \overline{\partial_0\Gamma{}^C_{BD}}
    + \overline\Gamma{}^{\alpha}_{0D}\overline\Gamma{}^C_{\alpha B} - \overline\Gamma{}^{\alpha}_{BD}\overline\Gamma{}^C_{\alpha 0})
 \nonumber
\\
  && -\nu^0\nu_A\nu^{B}\overline L_{1B}  + \nu^{B} \overline L_{AB} +2\nu^0\nu_A\overline L_{01}
  -\frac{1}{6}\nu_A\tilde R
\end{eqnarray}
(the right-hand side contains only known quantities).
From the definition of $\overline L_{0A}$ and the gauge condition $\overline {\partial_0 H{}^C }= 0$ one then computes $\overline{\partial_0 g_{0A}}$ and $\overline{\partial^2_{00}g_{1A}}$.
The relevant boundary condition is $\overline{\partial_0g_{0A}}=O(r^2)$.
The $\tilde g$-trace-free  part of \eq{mconf6} for $(\mu\nu\sigma\kappa)=(0A0B)$ yields $(\overline{\partial^2_{00}g_{AB}})\breve{}$.


The 10 independent components of the rescaled Weyl tensor in adapted null coordinates are
\begin{equation*}
 \overline d_{0101}\;, \quad \overline d_{011A}\;,\quad \overline d_{010A}\;, \quad \overline d_{01AB}\;, \quad
 \breve{\overline { d}}_{1A1B}\;, \quad \breve{\overline {d}}_{0A0B}\;.
\end{equation*}
The $\tilde g$-trace-free part of \eq{mconf2} with $(\mu\nu\sigma)=(A1B)$ determines $\breve{\overline d}_{1A1B}$,
\begin{eqnarray}
 \breve{\overline d}_{1A1B} &= &\nu_0 (\overline{\partial_0\Theta})^{-1}\Big[(\partial_1-\frac{1}{2}\tau)\omega_{AB} + \overline L_{11} \breve{\overline\Gamma}{}^1_{AB} - \tilde\nabla_A\overline L_{1B} + \frac{1}{2}\xi_B\overline L_{1A}
 \nonumber
\\
 && + \frac{1}{2}\ol g_{AB}(\tilde\nabla^C- \frac{1}{2}\xi^C)\overline L_{1C} \Big]
 \;.
\end{eqnarray}
All the remaining components of the rescaled Weyl tensor can be determined from \eq{mconf1}.
We will be rather sketchy here.
For $(\mu\nu\sigma)=(1A1)$ one finds
\begin{eqnarray}
 \overline{\nabla_1{d}_{011A}} +\nu^B\overline{\nabla_1 d_{1A1B}} - \nu_0\overline g^{CD}\overline{\nabla_C d_{1A1D}} =0
 \;,
\end{eqnarray}
which is an ODE for $\overline d_{011A}$, since the term $\overline g^{AB}\overline d_{1ABC}$, which appears when expressing the covariant derivatives in terms
of partial derivatives and connection coefficients, can be written as
\begin{eqnarray*}
 \overline g^{AB}\overline d_{1ABC} &=&   \nu^0\overline d_{011C} - \overline g^{1B}\overline d_{1B1C}
 \;.
\end{eqnarray*}
Any bounded solution of the MCFE satisfies $\overline d_{011A} = O(r)$ for small $r$.

For $(\mu\nu\sigma)=(AB1)$ one obtains an ODE for $\overline d_{01AB}$,
\begin{eqnarray}
 \overline{\nabla_1d_{01AB}} + \nu^C\overline{\nabla_1 d_{1CAB}} - \nu_0\overline g^{CD}\overline{\nabla_Dd_{1CAB}}=0
 \;,
\label{ODE_d01AB}
\end{eqnarray}
the boundary condition is given  by the requirement $\overline d_{01AB}=O(r^2)$.
Note for this that $\overline d_{1ABC}$ and $\overline d_{0[AB]1}$, both of which are  hidden in the covariant derivatives appearing in \eq{ODE_d01AB}, can be expressed in terms of
$\overline d_{1A1B}$ and $\overline d_{011A}$. Indeed,  symmetries of the rescaled Weyl tensor imply that
\begin{eqnarray*}
 \overline d_{1ABC} &=& 2\overline g{}^{EF}\overline d_{1EF[C}\overline g_{B]A}
 \,=\,  2 (  \nu^0\overline d_{011[C} - \overline  g^{1D}\overline d_{1D1[C}  )\overline g_{B]A}
\;,
\\
2\overline d_{0[AB]1} &=&- \overline d_{01AB}
 \;.
\end{eqnarray*}

The $(\mu\nu\sigma)=(101)$-component of \eq{mconf1} can be employed to determine $\overline d_{0101}$,
\begin{eqnarray}
 \overline{ \nabla_1 d_{0101}} + \nu^C\overline{\nabla_1 d_{011C} }- \nu_0\overline g^{CD}\overline{\nabla_Cd_{011D}}=0
\;.
\end{eqnarray}
For that purpose one needs to express  $\overline g^{AB} \overline d_{0AB1}$ in terms of known components  and $\overline d_{0101}$,
\begin{eqnarray*}
  \overline g^{AB} \overline d_{0AB1} &=&  \overline g^{1A}\overline d_{011A} - \nu^0\overline d_{0101}
 \;.
\end{eqnarray*}
The boundary condition for bounded solutions is $\overline d_{0101}=O(1)$.

The function $\overline d_{010A}$ is obtained from \eq{mconf1} with $(\mu\nu\sigma)=(0A1)$,
\begin{eqnarray}
 \overline{\nabla_1 d_{010A} }+ \nu^C\overline{\nabla_1 d_{0A1C}} - \nu_0\overline g^{CD}\overline{\nabla_C d_{0A1D}}=0
 \;,
\end{eqnarray}
and $\overline d_{010A}=O(r)$.
To obtain the desired ODE one needs to use the following relations, which, again, follow from the symmetry properties of the rescaled Weyl tensor:
\begin{eqnarray*}
\overline g^{AB} \overline d_{0AB1}  &=& \overline g^{1A}\overline d_{011A} - \nu^0\overline d_{0101}
 \;,
\\
2\nu^0\overline d_{0(AB)1} &=&  \overline g^{11} \overline d_{1A1B}- 2\overline g{}^{1C}\overline d_{1(AB)C} - \overline g^{CD}\overline d_{CABD}
\;,
\\
 \overline g^{AB} \overline g^{CD}\overline d_{CABD}&=& -2 \overline g^{AB}( \overline g^{1C}\overline d_{1ABC} + \nu^0\overline d_{0AB1}  )
 \;,
\\
\overline g^{CD}\overline d_{CABD}&=& \frac{1}{2}\overline g_{AB}\overline g^{CD} \overline g^{EF} \overline d_{CEFD}
\;,
\\
 \overline d_{ABCD} &=& \overline g^{EF}(\overline g_{C[B}\overline d_{A]EFD} - \overline g_{D[B}\overline d_{A]EFC})
\;,
\\
 \overline g^{AB}\overline d_{0ABC} &=& -\nu^0\overline d_{010C} - \overline g^{11}\overline d_{011C} - \overline g^{1B}(\overline d_{01BC} - \overline d_{0(BC)1}  - \overline d_{0[BC]1})
\;,
\\
 \overline d_{0ABC} &=& 2\overline g^{EF} \overline d_{0EF[C}\overline g_{B]A}
 \;.
\end{eqnarray*}

To gain an equation for $\breve{\overline d}_{0A0B}$ we observe that due to the tracelessness of the rescaled Weyl tensor we have
\begin{eqnarray*}
 0 &=& \overline g^{\mu\nu} \overline{\nabla_0 d_{\mu AB\nu}} - \frac{1}{2}\overline g_{AB}\overline g^{CD} \overline g^{\mu\nu} \overline{\nabla_0 d_{\mu CD\nu}}
\\
 &=&  2\nu^0\overline{\nabla_0 \breve  d_{0( AB)1}}  - \overline g^{11}\overline{\nabla_0 \breve d_{1A1B}} + 2(\overline g^{1C}\overline{\nabla_0 d_{1( AB)C}})\breve{}
 \;.
\end{eqnarray*}
Two of the transverse derivatives can be eliminated via the  following relations,
%
\begin{eqnarray*}
0=  \nu_0\overline{\nabla_{\rho} \breve d_{1(AB)}{}^{\rho}} &\equiv&  -\overline{\nabla_{0} \breve d_{1A1B}}
 + \overline{\nabla_{1} \breve{ d}_{0(AB)1}} - \nu_0\overline g^{11}\overline{\nabla_{1} \breve { d}_{1A1B}}
\\
&& - (\nu^C\overline{\nabla_{1} { d}_{1(AB)C}})\breve{} + \nu^C \overline{\nabla_{C} \breve{d}_{1A1B}}
 + \nu_0(\overline g^{CD}\overline{\nabla_{D}  d_{1(AB)C}})\breve{}
 \;,
\\
 0  = \nu_0 \overline{\nabla_{\rho}  d_{ABC}{}^{\rho}} &\equiv&  \overline{\nabla_{0}  d_{ABC1}}
+ \overline{\nabla_{1}  d_{ABC0}} +  \nu_0\overline g^{11} \overline{\nabla_{1}  d_{ABC1}}
\\
 &&   -\nu^D \overline{\nabla_{1}  d_{ABCD}} - \nu^D\overline{\nabla_{D}  d_{ABC1}}
 + \nu_0 \overline g^{DE} \overline{\nabla_{E}  d_{ABCD}}
 \;,
\end{eqnarray*}
so that we end up with an expression for $\overline{\nabla_{0}\breve{ d}_{0(AB)1}}$.
The trace-free and symmetrized part of equation \eq{mconf1} with $(\mu\nu\sigma)=(0AB)$ reads
\begin{eqnarray}
0&=& \nu^0 \overline{\nabla_{0}\breve{ d}_{0(AB)1}}+  \nu^0 \overline{\nabla_{1}\breve{d}_{0AB0}}
 + \overline  g^{11} \overline{\nabla_{1}\breve{d}_{0(AB)1}} +  \overline  g^{1C} \overline{\nabla_{C} \breve d_{0(AB)1}}
 \nonumber
\\
 &&+  (\overline  g^{1C} \overline{\nabla_{1}d_{0(AB)C}})\breve{}
 +  (\overline  g^{CD}\overline{ \nabla_{D} d_{0(AB)C}})\breve{}
 \;,
\end{eqnarray}
which  thus provides an ODE for $\breve{\overline d}_{0A0B}$ with boundary condition $\breve{\overline d}_{0A0B}=O(r^2)$.

Finally, one determines $\overline L_{00}$ from equation \eq{mconf2} with $(\mu\nu\sigma)=(100)$
and the contracted Bianchi identity \eq{BianchiL},
\begin{equation}
  2\nu^0\overline{\nabla_1 L_{00}}  +\overline g^{11}\overline{\nabla_1  L_{01}}
  +2\ol g^{1A}\overline{\nabla_{(1} L_{A)0}}
+\overline g^{AB}\overline{\nabla_A L_{0B}}
-\frac{1}{6}\overline{\partial_0 R}
= (\nu^0)^2 \overline{\partial_0\Theta}\,\overline d_{0101}
  \;.
\end{equation}
The boundary condition, satisfied by any bounded solution, is $\overline L_{00}=O(1)$.

\subsection{Constraint equations in the $ (R=0,\overline s =-2,\kappa=0, \hat g = \eta)$-wave-map gauge}
\label{kappa0_wavemap}

To simplify computations significantly  let us choose a specific gauge.
The CWE take their simplest form  if we impose the gauge condition
\begin{equation}
  R \,=\,0
 \;,
\label{gauge_curv_scalar}
\end{equation}
which we shall do henceforth.
Moreover, we assume the wave-map gauge condition and an affinely parameterized cone, meaning that
\begin{equation}
\kappa=0  \quad \text{and} \quad W^{\sigma}= 0\;.
\end{equation}
Furthermore, we set
\begin{equation}
\overline s = - 2\;,
\label{initialdata_s}
\end{equation}
(the negative sign of $\overline s$ makes sure that $\Theta$ will be positive inside the cone),
and use a Minkowski target $\hat g_{\mu\nu} = \eta_{\mu\nu}$.
This way many of the constraint equations can be solved explicitly.
From now on all equalities are meant to hold in this particular gauge, if not stated otherwise.

The relevant boundary conditions for the ODEs, which follow from regularity conditions at the vertex, have been specified in the previous section.
Recall that the free initial data are given by the $\tilde g$-trace-free tensor $\omega_{AB}$
and that we treat the case where the initial surface is $C_{i^-}$,  i.e.\ we have
\begin{eqnarray}
 \overline\Theta=0\;.
 \label{initialdata_theta}
\end{eqnarray}
Regularity for the Schouten tensor  requires $\omega_{AB}=O(r^2)$.
However, regularity for the rescaled Weyl tensor requires the stronger condition (cf.\ \eq{initialdata_d1} below)
%
\begin{equation}
 \omega_{AB}=O(r^4)
 \label{behaviour_initial_data}
 \;.
\end{equation}
Many of the above equations can be solved straightforwardly,  we just present the results,
\begin{eqnarray}
 \overline g_{\mu\nu} = \eta_{\mu\nu} \;, \quad\overline L_{1\mu}=0\;, \quad \overline g^{AB}\overline L_{AB}=0 \;, \quad
 \overline L_{0A} = \frac{1}{2}\tilde\nabla^B\overline{\partial_0 g_{AB}}
 \;.
 \label{initialdata_g_L}
  \label{expr_omegaAB_transgAB}
\end{eqnarray}
Note that $L_{\mu\nu}$ is trace-free as required by Lemma~\ref{lemma_properties_L}.
On the way to compute these quantities we have found
\begin{eqnarray}
 & \tau=2/r \;, \quad  \overline{\partial_0\Theta}=-2r\;, \quad  \overline{\partial_0 g_{1\mu}}=0\:, \quad \overline g^{AB} \overline{\partial_0 g_{AB}}=0
 \;,
 \nonumber &
\\
 &  \xi_A=0\;, \quad \zeta=-2/r\;,
\nonumber &
\\
&  (\partial_1- r^{-1})\overline{\partial_0g_{AB}} = - 2 \omega_{AB}  \,\text{ with }\, \overline{\partial_0 g_{AB}}=O(r^3)
 \label{transgAB_omega}
 \;.
&
\end{eqnarray}
We further obtain (note that $\overline \Gamma{}^B_{0A} = \frac{1}{2}\overline g^{BC}\overline{\partial_0 g_{AC}}$)
\begin{eqnarray}
  \overline d_{1A1B} &=& -\frac{1}{2}\partial_1(r^{-1}\omega _{AB})
 \;,
 \label{initialdata_d1}
\\
 (\partial_1+ 3 r^{-1})\overline d_{011A} &=& \tilde\nabla^B\overline d_{1A1B}
 \;,
  \label{initialdata_d2}
\\
 (\partial_1+ 3r^{-1})\overline d_{0101} &=& \tilde\nabla^A\overline d_{011A} - \frac{1}{2}\overline{\partial_0 g^{AB}}\overline d_{1A1B}
  \label{initialdata_d3}
 \;,
\\
 (\partial_1+ r^{-1})\overline d_{01AB} &=& 2\tilde\nabla_{[A} \overline d_{B]110} - 2\overline \Gamma{}^C_{0[A}\overline d_{B]11C}
  \label{initialdata_d4}
 \;,
\\
 (\partial_1+ r^{-1})\overline d_{010A}   &=&  \frac{1}{2} \tilde\nabla^B( \overline d_{01AB} - \overline d_{1A1B} )+ \frac{1}{2} \tilde\nabla_A \overline d_{0101}
 + r^{-1} \overline d_{011A}
 \nonumber
\\
&& +  2\overline \Gamma{}^B_{0A}\overline d_{011B}
  \label{initialdata_d5}
 \;,
\end{eqnarray}
with $ \overline d_{011A}=O(r)$, $\overline d_{0101}=O(1)$, $\overline d_{01AB}=O(r^2)$ and $\overline d_{010A}=O(r)$.
To derive \eq{initialdata_d1}-\eq{initialdata_d5}  we have used the following relations, which  follow from algebraic symmetry properties of the Weyl tensor,
\begin{eqnarray}
 \overline d_{ABCD} &=&   -2  \overline g_{A[C} \overline g_{D]B}d_{0101}     \;,
     \label{relation_d2}
\\
  2d_{0[AB]1} &=& - d_{01AB}\;,
    \label{relation_d3}
\\
 2\overline d_{0(AB)1} &=&
   \overline d_{1A1B} -  \overline g_{AB}\overline d_{0101}
 \;,
    \label{relation_d4}
\\
  \overline d_{1ABC} &=&  -2 \overline d_{011[B}\overline  g_{C]A} \;,
    \label{relation_d5}
\\
    \overline d_{0ABC} &=&
 2\overline d_{010[B} \overline g_{C]A} + 2 \overline d_{011[B} \overline g_{C]A}
 \;.
    \label{relation_d6}
\end{eqnarray}

Before we proceed let us establish some relations:
\begin{lemma}
 \label{lemma_tracelessness}
 \begin{enumerate}
  \item [(i)] \( (\overline g^{CD} \overline{\partial_0 g_{AC}}\,  \overline{\partial_0 g_{BD}})\breve{}=0 \),
  \item [(ii)] \( (\overline g^{CD}\overline{\partial_0 g_{C(A}}\omega_{B)D})\breve{} =0 \),
  \item [(iii)] \( (\overline g^{CD}\overline{\partial_0 g_{C(A}}\,\overline d_{B)1D1} )\breve{} =0 \).
 \end{enumerate}
\end{lemma}
\begin{proof}
 This follows from the constraint equations \eq{transgAB_omega}-\eq{initialdata_d1}, together with the $\tilde g$-tracelessness of $\overline{\partial_0 g_{AB}}$.
 \qed
\end{proof}

The lemma can be employed to simplify the ODE which determines $\breve{\overline d}_{0A0B}$,
%
\begin{eqnarray}
  2( \partial_{1} - r^{-1} )\breve {\overline d}_{0A0B}
&=& 3(\partial_{1} - r^{-1})  \breve{\overline  d}_{0(AB)1}   - ( \partial_{1} - r^{-1} ) {\overline d}_{1A1B}
 \nonumber
\\
 && +(\tilde \nabla^{C} \overline d_{1(AB)C})\breve{} +   2 (\tilde \nabla^{C} \overline d_{0(AB)C} )\breve{} -(\overline{\partial_0g^{CD}}\overline  d_{ACBD })\breve{}
  \nonumber
\\
 &&  +[2\overline \Gamma{}^C_{0(A}( \overline d_{B)C01} - \overline d_{B)01C}  +\frac{1}{2} \overline d_{B)1C1})]\breve{}
  \nonumber
\\
&=& \frac{1}{2}(\partial_{1} - r^{-1}) {\overline d}_{1A1B}  +(\tilde \nabla_{(A} \overline d_{B)110})\breve{}
 +   2 (\tilde \nabla_{(A} \overline d_{B)010}    )\breve{}
  \nonumber
\\
 &&  +3\overline \Gamma{}^C_{0(A}\overline d_{B)C01}
   +\frac{3}{2}\overline d_{0101}\overline{ \partial_0 g_{AB}}
   \label{initialdata_d6}
 \;,
\end{eqnarray}
with $\breve{\overline d}_{0A0B}=O(r^2)$.
Finally, one shows that the missing component of the Schouten tensor satisfies
\begin{eqnarray}
  2(\partial_1+ r^{-1})\overline L_{00}  = \frac{1}{2} \omega^{AB}\overline{\partial_0 g_{AB}} -2r \overline d_{0101} - \tilde \nabla^A \overline L_{0A}
 \label{initialdata_L00}
\;,
\end{eqnarray}
with $\overline L_{00}=O(1)$.

We aim now to find explicit solutions to some of the remaining ODEs \eq{initialdata_d2}-\eq{initialdata_d5}.
The key observation to solve \eq{initialdata_d2} is that, due to \eq{transgAB_omega}, we have
\begin{eqnarray}
 \overline d_{1A1B} &=& -\frac{1}{2}r^{-1}\partial_1\Big(\omega _{AB} + \frac{1}{2}r^{-1}\overline{\partial_0 g_{AB}}\Big)
 \;.
\end{eqnarray}
Hence we find
\begin{eqnarray}
   \partial_1(r^3\overline d_{011A}) &=& -\frac{1}{2}\partial_1\Big( r^2\tilde\nabla^B \omega _{AB} +  r\overline L_{0A}\Big)
\nonumber
\\
  \overset{\overline d_{011A}=O(r)}{\Longrightarrow} \quad    \overline d_{011A} &=& -\frac{1}{2} r^{-1}\tilde\nabla^B \omega _{AB}  -\frac{1}{2}  r^{-2}\overline L_{0A}
 \label{expression_d011A}
\\
 &\overset{\eq{transgAB_omega}}{=}&   \frac{1}{2}r^{-1}\partial_1\overline L_{0A}
\label{diff_expression_d011A}
 \;.
\end{eqnarray}
The equations \eq{transgAB_omega} and  \eq{diff_expression_d011A} can be used to  rewrite \eq{initialdata_d4},
\begin{eqnarray}
  \partial_1 (r \overline d_{01AB}) &=&  \partial_1\tilde\nabla_{[A}\overline L_{B]0}
  - r \overline \Gamma{}^C_{0[A}\partial_{|1|}(r^{-1}\omega_{B]C})
 \nonumber
\\
 &=& \partial_1( \tilde\nabla_{[A}\overline L_{B]0}  -  \overline\Gamma{}^C_{0[A}\omega_{B]C} )
 \nonumber
\\
  \overset{\overline d_{01AB}=O(r^2)}{\Longrightarrow} \quad
 \overline d_{01AB}&=&
  r^{-1}\tilde\nabla_{[A}\overline L_{B]0}- r^{-1}\overline \Gamma{}^C_{0[A}\omega_{B]C}
 \label{expression_d01AB}
 \;.
\end{eqnarray}

The constraint equations in the $ (R=0,\overline s =-2,\kappa=0, \hat g = \eta)$-wave-map gauge are summed up in \eq{constraint_g}-\eq{relation_lambda_omega}
below.

\section{Applicability of Theorem~\ref{inter-thm} on the $\mathbf{C_{i^-}}$-cone}
\label{sec_applicability}

Let us suppose we have been given initial data $\omega_{AB} \equiv \breve{\mathring L}_{AB}$ on $C_{i^-}$,
supplemented by a gauge choice for $R$, $\mathring s$, $ W^{\sigma}$ and $\kappa$.
Then we  solve the hierarchical system of constraint equations derived above;
the solutions are denoted by
$\mathring g_{\mu\nu}$, $\mathring L_{\mu\nu}$ and $\mathring d_{\mu\nu\sigma\rho}$.
Let us further assume that there exists a smooth solution of the CWE in some neighbourhood to the future of $i^-$,
smoothly extendable  through $C_{i^-}$,
 which induces
the data $\overline \Theta=0$, $\overline s=\mathring s$, $\overline g_{\mu\nu}=\mathring g_{\mu\nu}$,
$\overline L_{\mu\nu} = \mathring L_{\mu\nu}$ and $\overline d_{\mu\nu\sigma\rho} = \mathring d_{\mu\nu\sigma\rho}$ on $C_{i^-}$.
The purpose of this section is to investigate to what extent the hypotheses made in Theorem~\ref{inter-thm} are satisfied in the case of
initial data which have been constructed as a solution of the constraint equations.
For convenience and to make computations significantly easier
we shall not do it in an arbitrary generalized wave-map gauge but
prefer to work within the specific gauge \eq{gauge_curv_scalar}-\eq{initialdata_s}.

\subsection{$(R=0,\overline s=-2,\kappa=0, \hat g = \eta)$-wave-map gauge}
We restrict attention to the $\kappa=0$-wave-map gauge with $ W^{\sigma}= 0$;
moreover, we set $R=0$ and $\mathring s =-2$,
and use a Minkowski target $\hat g_{\mu\nu} = \eta_{\mu\nu}$.
All equalities are meant to hold in this specific gauge.
For reasons of clarity let us recall the CWE in an $(R=0)$-gauge, where they take their simplest form,
\begin{eqnarray}
 \Box^{(H)}_{ g} L_{\mu\nu}&=&  4 L_{\mu\kappa} L_{\nu}{}^{\kappa} -  g_{\mu\nu}| L|^2
  - 2\Theta d_{\mu\sigma\nu}{}^{\rho}  L_{\rho}{}^{\sigma}
  \label{cwe1*}
  \;,
  \\
  \Box_{ g}  s  &=& \Theta| L|^2
  \label{cwe2*}
  \;,
  \\
  \Box_{ g}\Theta &=& 4 s
  \label{cwe3*}
  \;,
  \\
  \Box^{(H)}_g d_{\mu\nu\sigma\rho}
  &=& \Theta d_{\mu\nu\kappa}{}^{\alpha}d_{\sigma\rho\alpha}{}^{\kappa}
   - 4\Theta d_{\sigma\kappa[\mu} {}^{\alpha}d_{\nu]\alpha\rho}{}^{\kappa}
  \label{cwe4*}
  \;,
  \\
  R^{(H)}_{\mu\nu}[g] &=& 2L_{\mu\nu}
  \label{cwe5*}
  \;.
\end{eqnarray}

The constraint equations, from which the initial data for the CWE are determined from given free data $\omega_{AB}\equiv \breve{\mathring L}_{AB}=O(r^4)$
read:
\begin{eqnarray}
\mathring g_{\mu\nu} &=& \eta_{\mu\nu} \;,    \label{constraint_g}
\\
\mathring L_{1\mu}&=&0\;, \quad
 \mathring L_{0A} \,=\, \frac{1}{2}\tilde\nabla^B\lambdahere_{AB}\;, \quad
  \mathring g^{AB} \mathring L_{AB}\,=\,0\;,
  \label{constraint_L}
\\
   \mathring d_{1A1B} &=& -\frac{1}{2}\partial_1(r^{-1}\omega _{AB})
    \label{initial_d1A1B_spec_gauge}
 \;,
\\
    \mathring d_{011A} &=&  \frac{1}{2}r^{-1}\partial_1\mathring L_{0A}\;,
     \label{initial_d011A_spec_gauge}
\\
 \mathring d_{01AB}&=& r^{-1}\tilde\nabla_{[A}\mathring L_{B]0} - \frac 12 r^{-1}\lambda_{[A}{}^C \omega_{B]C}
  \label{initial_d01AB_spec_gauge}
 \;,
\\
 (\partial_1+\ 3r^{-1})\mathring d_{0101} &=& \tilde\nabla^A\mathring d_{011A} + \frac{1}{2}\lambdahere^{AB}\mathring d_{1A1B}  \;,
  \label{initial_d0101_spec_gauge}
\\
 2(\partial_1+ r^{-1})\mathring d_{010A}   &=&  \tilde\nabla^B( \mathring d_{01AB} - \mathring d_{1A1B} )
 + \tilde\nabla_A \mathring d_{0101} + 2r^{-1} \mathring d_{011A}
  \nonumber
\\
&&
 +  2\lambda_{A}{}^B \mathring d_{011B}\;,
 \label{initial_d010A_spec_gauge}
 \\
  4( \partial_{1} - r^{-1} )\breve {\mathring d}_{0A0B}
&=& (\partial_{1} - r^{-1}) {\mathring d}_{1A1B}  + 2(\tilde \nabla_{(A} \mathring d_{B)110})\,\breve{}
 +   4 (\tilde \nabla_{(A} \mathring d_{B)010}    )\,\breve{}
 \nonumber
\\
 &&
  + 3 \lambda_{(A}{}^C\mathring d_{B)C01}
   +3\mathring d_{0101}\lambdahere_{AB}
   \label{initial_d0A0B_spec_gauge}
\;,
 \phantom{xxxx}
\\
  4(\partial_1+ r^{-1})\mathring L_{00}  &=& \lambdahere^{AB} \omega_{AB} -4r \mathring d_{0101} - 2\tilde \nabla^A \mathring L_{0A}
   \label{initial_L00_spec_gauge}
\;,
\end{eqnarray}
with
\begin{equation}
 \mathring d_{0101}=O(1)\;, \quad \mathring d_{010A} =O(r)\;, \quad \breve {\mathring d}_{0A0B}=O(r^2)\;, \quad
\mathring L_{00}=O(1)\;,
\label{constraints_boundary_data}
\end{equation}
and where $\lambda_{AB}$ is the unique solution of
\begin{equation}
 (\partial_1- r^{-1})\lambda_{AB} = - 2 \omega_{AB} \quad \text{with} \quad \lambda_{AB}=O(r^5)\;.
  \label{relation_lambda_omega}
\end{equation}
Note that the expansion $\tau$ satisfies
\begin{equation}
 \tau \,=\, 2/r \;.
\end{equation}
All the other components of $\mathring g_{\mu\nu}$, $\mathring L_{\mu\nu}$ and $\mathring d_{\mu\nu\sigma\rho}$ follow from their usual symmetry properties which they are required to satisfy.

\subsection{Vanishing of $\overline H{}^{\sigma}$}
\label{subsec_vanish_H}
Inserting the definition of the reduced Ricci tensor \eq{ricci_riccired} equation \eq{cwe5*} becomes
\begin{eqnarray}
 R_{\mu\nu} - g_{\sigma(\mu}\hat\nabla_{\nu)}H^{\sigma} =2L_{\mu\nu}
  \label{wavered*}
  \;.
\end{eqnarray}
Utilizing the constraint equations \eq{constraint_g} and the identities~\cite{CCM2}
\begin{eqnarray*}
 \overline R_{11} &\equiv& -\partial_1 \tau + \tau\overline\Gamma{}^1_{11} - |\sigma|^2-\frac{1}{2}\tau^2
  \, = \, \tau \overline\Gamma{}^1_{11}
  \;,
 \\
 \overline \Gamma{}^1_{11} &\equiv& \kappa -\frac{1}{2}\nu_0 \overline H{}^0
  \, = \, -\frac{1}{2} \overline H{}^0
  \;,
\end{eqnarray*}
the latter one follows from the definitions of $H^{\sigma}$ and $\kappa$, 
we conclude that the solution satisfies the ODE
\begin{eqnarray}
  \hat\nabla_{1}\overline H{}^{0} + \frac{1}{2}\tau\overline H{}^0 = 0 \quad \Longleftrightarrow \quad
  (\partial_{1} + r^{-1})\overline H{}^{0} = 0
  \;.
   \label{eqn_H0}
\end{eqnarray}
%
For any regular solution of the CWE
 the function  $\overline H{}^0$
will be bounded near the vertex. We observe that
\begin{equation}
 \overline H{}^0 \,=\, 0
 \label{vanishing_H0}
\end{equation}
is the only solution of \eq{eqn_H0} where this is the case.
Then we immediately obtain
\begin{equation}
 \overline \Gamma{}^1_{11} \,=\, \kappa \,=\, 0
  \label{relation_kappa}
  \;.
\end{equation}

Recall the definition of the field $\xi_A$, which vanishes in our gauge,
\begin{eqnarray*}
 \xi_A &\equiv&  -2\nu^0\partial_1\nu_A + 4\nu^0\nu_B\chi_A{}^B + \nu_A\overline V{}^0 + \overline g_{AB}\overline V{}^B - \overline g_{AD} \overline g^{BC}\tilde\Gamma^D_{BC}
\,=\,  0
  \;.
\end{eqnarray*}
From the constraint equations, \eq{wavered*} and the identities~\cite{CCM2}
\begin{eqnarray}
 \overline R_{1A} &\equiv& (\partial_1 + \tau)\overline \Gamma{}^1_{1A} + \tilde\nabla_B\chi_A{}^B -\partial_A\overline\Gamma{}^1_{11} -\partial_A\tau
  \,=\, (\partial_1 + \tau)\overline \Gamma{}^1_{1A}
  \;,
 \\
  \xi_A   &\equiv & -2\overline \Gamma{}^1_{1A} - \overline H_A -\nu_A \overline H{}^0
   \,=\,  -2\overline \Gamma{}^1_{1A} - \overline H_A
    \label{rel_xiA_HA}
 \;,
\end{eqnarray}
we find that $\overline H_A:=\overline g_{AB}\overline H{}^B$ fulfills the ODE
\begin{eqnarray*}
      \partial_1\overline H_A   \,=\,0
  \;.
\end{eqnarray*}
Any regular  solution necessarily satisfies $\overline H_A = O(r)$ and we infer
\begin{eqnarray}
 \overline H{}^A =0 \quad \text{and} \quad \overline\Gamma{}^1_{1A}=0 
 \label{vanishing_HA}
  \;.
\end{eqnarray}

We have introduced the function
\begin{equation*}
 \zeta \,\equiv\, 2(\partial_1 +\kappa +\frac{1}{2}\tau )\overline g^{11}
   + 2\overline V{}^1
  \, =\, -\tau
   \;.
\end{equation*}
From \eq{wavered*}, the constraint equation $\overline g^{AB} \overline L_{AB}=0$
and the identities~\cite{CCM2}
\begin{eqnarray}
 \overline g^{AB} \overline R_{AB} & \equiv& 2(\partial_1 + \overline\Gamma{}^1_{11} + \tau)[\underbrace{(\partial_1 + \overline \Gamma{}^1_{11} + \frac{1}{2}\tau)\overline g^{11} + \overline g^{\mu\nu}\overline \Gamma{}^1_{\mu\nu}}_{\equiv \overline g^{AB}\overline \Gamma{}^1_{AB} + \frac{1}{2}\tau\overline{}g^{11}}]
 \nonumber
\\
 &&  + \tilde R - 2\overline g^{AB}\overline\Gamma{}^1_{1A}\overline\Gamma{}^1_{1B} - 2\overline g^{AB}\tilde\nabla_A\overline\Gamma{}^1_{1B}
\nonumber
 \\
 & =& 2(\partial_1  + \tau)[\overline g^{AB}\overline \Gamma{}^1_{AB} + \frac{1}{2}\tau]
   +\frac{1}{2}\tau^2
 \;,
\\
 \zeta &\equiv&  2\overline g^{AB}\overline\Gamma{}^1_{AB} + \tau\overline g^{11} + \nu_0\overline g^{11}\overline H{}^0 - 2 \overline H{}^1
  \nonumber
   \\
  &=& 2\overline g^{AB}\overline\Gamma{}^1_{AB} + \tau  - 2 \overline H{}^1
 \label{rel_zeta_H}
 \:,
\end{eqnarray}
we deduce that
\begin{eqnarray*}
  (\partial_1 + r^{-1} ) \overline H{}^1 \,=\, 0
  \;.
\end{eqnarray*}
Our solution is supposed to be regular at $i^-$, whence $ \overline H{}^1 =O(1)$ and we conclude
\begin{eqnarray}
  \overline H{}^{1}=0 \quad \text{and} \quad  \overline g^{AB}\overline\Gamma{}^1_{AB} =   - \tau
 \label{vanishing_H1}
 \;.
\end{eqnarray}
Altogether we have proven that
\begin{equation}
 \overline H{}^{\sigma}=0
 \label{vanishing_H}
 \:.
\end{equation}

Note that once we know the values of the wave-gauge vector on $C_{i^-}$, we can compute the values of
certain components of the transverse derivative of the metric on $C_{i^-}$. More concretely, we find
that the solution satisfies
\begin{equation*}
 \overline{\partial_0 g_{11}}=0
\;, \quad  \overline{\partial_0 g_{1A}}=0 \:, \quad \overline g^{AB} \overline{\partial_0 g_{AB}}=0
 \;.
 \label{certain_relations}
\end{equation*}
We also have
\begin{eqnarray*}
 \overline R_{AB} &\equiv& \overline{\partial_{\alpha}\Gamma^{\alpha}_{AB}}  - \partial_A\overline \Gamma{}^{\alpha}_{\alpha B} + \overline \Gamma{}^{\alpha}_{AB} \overline \Gamma{}^{\beta}_{\beta\alpha}
  - \overline \Gamma}{^{\alpha}_{\beta A} \overline \Gamma{}^{\beta}_{\alpha B}
  \\
   &=& \tilde R_{AB} -\frac{1}{4}\tau^2 \overline g_{AB}
  -\frac{1}{2} (\partial_{1}-\tau) \overline{\partial_0 g_{AB}}
   +\overline{\partial_{0}\Gamma^{0}_{AB}} -\frac{1}{2}\tau \overline g_{AB} \overline \Gamma{}^{0}_{00}
 \\
   &=&  - (\partial_{1}-r^{-1}) \overline{\partial_0 g_{AB}}
   \;,
\end{eqnarray*}
where we employed the relation
\begin{equation*}
 \overline{\partial_{0}\Gamma^{0}_{AB}} = \frac{1}{2}\tau \overline g_{AB} \overline{\partial_0 g_{01}} - \frac{1}{2}\partial_1\overline{\partial_0 g_{AB}}
  \;.
\end{equation*}
The vanishing of $\overline H{}^{\sigma}$ implies via \eq{wavered*} and \eq{constraint_g}
\begin{equation*}
 \overline R_{AB} \, = \, 2\overline L_{AB} \, = 2\omega_{AB}
  \;,
\end{equation*}
and thus by \eq{relation_lambda_omega}
\begin{equation*}
 (\partial_1- r^{-1})(\lambda_{AB} - \overline{\partial_0 g_{AB}}) =0
 \;.
\end{equation*}
For initial data of the form $\omega_{AB}=O(r^4)$ we have $\lambda_{AB}=O(r^5)$.
Since regularity requires~\cite{CCM2} $\overline{\partial_0 g_{AB}}=O(r^3)$,
we discover  the expected relation
\begin{eqnarray*}
  \lambda_{AB} \,=\, \overline{\partial_0 g_{AB}}
   \;.
\end{eqnarray*}

\subsection{Vanishing of $\overline {\nabla_{\mu} H^{\sigma}}$ and $\overline \zeta_{\mu}$}
\label{sect_vanishing_nablaH}

We know that the wave-gauge vector satisfies the wave equation \eq{wave_H},
\begin{equation}
 \nabla^{\nu} \hat\nabla_{\nu} H^{ \alpha}+2g^{\mu\alpha} \nabla_{[\sigma} \hat\nabla_{\mu]} H^{ \sigma}
 + 4\nabla^{\nu} L_{\nu}{}^{\alpha} =0
 \label{wave_H+}
 \;.
\end{equation}
Let us first consider the $\alpha=0$-component evaluated on $\scri^-$,
\begin{eqnarray}
 (\partial_1+r^{-1})\overline{\partial_0H^0}  + 2\overline{\partial_0 L_{11}} =0
 \label{ODE_transH0}
 \;.
\end{eqnarray}
We need to show that the source term vanishes.
Equation \eq{cwe1*} provides an expression for $\overline{\partial_0 L_{11}}$,
\begin{eqnarray}
 \label{5X12.1}
&& \overline{\Box^{(H)}_gL_{11}} = 0 \quad \Longleftrightarrow \quad
 (\partial_1+ r^{-1})\overline{\partial_0L_{11}} =0
 \;.
\end{eqnarray}
Any regular solution  satisfies
$\overline{\partial_0 L_{11}} = \overline{\nabla_0 L_{11}}  = O(1)$.
There is precisely one bounded  solution  of \eq{5X12.1},
which is
\begin{equation}
 \overline{\partial_0L_{11}}=0
 \label{vanishing_transL11}
 \;.
\end{equation}
The function $\overline {\nabla_0 H^0}= \overline{\partial_0 H^0}$ needs to be bounded as well, and the only bounded
solution of \eq{ODE_transH0} is
\begin{equation}
 \overline{\partial_0 H^0}=0
 \label{vanishing_transH0}
 \;.
\end{equation}
Taking the trace of  \eq{wavered*}  then shows that the curvature scalar vanishes initially,
\begin{eqnarray}
  \overline  R_g =0
 \label{transH_R}
  \;.
\end{eqnarray}
Using \eq{vanishing_transH0} as well as the relation $ \overline R_{01}=2\overline L_{01}=0 $, which follows from \eq{wavered*},
one verifies that
\begin{equation*}
 \overline{\partial_0 g_{01}} =0 \quad \text{and} \quad  \overline{\partial^2_{00}g_{11}}=0\;.
\end{equation*}

The $\alpha=A$-component of \eq{wave_H+} yields
\begin{eqnarray}
 (\partial_1 + 2r^{-1})\overline{\partial_0H^A}
 + 2\overline g^{AB}( \overline{\partial_0 L_{1B}}  + \partial_1\overline L_{0B} +\tau\overline L_{0B}+ \tilde\nabla^C \omega_{BC}  )=0
 \label{ODE_transHA}
 \;.
\end{eqnarray}
We employ \eq{cwe1*} to compute the source term,
\begin{eqnarray}
 \overline{ \Box^{(H)}_{ g} L_{1A}}=0  \quad \Longleftrightarrow \quad
 2\partial_1\overline{\partial_0L_{1A}} -\tau\tilde\nabla^B  \omega_{AB} -\tau^2 \overline L_{0A}=0
 \label{wave_L1A}
 \;.
\end{eqnarray}
Equation \eq{relation_lambda_omega} implies
\begin{eqnarray}
   2\tilde\nabla^B\omega_{AB} = -\tilde\nabla^B\partial_1\lambda_{AB}  + \tau \overline L_{0A}
 = -2\partial_1\overline L_{0A}  -  \tau \overline L_{0A}
 \label{omega_L0A}
 \;.
\end{eqnarray}
From \eq{wave_L1A} and \eq{omega_L0A} we derive the ODE
\begin{equation}
 \partial_1(\overline{\partial_0 L_{1A}} + r^{-1} \overline L_{0A}) =0
 \;.
 \label{transL1A_L0A}
\end{equation}
For any sufficiently regular  solution we have $\overline{\partial_0 L_{1A}}=\overline{\nabla_0 L_{1A}}=O(r)$.
Since the initial data satisfy $\omega_{AB}=O(r^4)$, we have $\overline L_{0A}=O(r^2)$ by \eq{wave_L1A}.
We then conclude from \eq{transL1A_L0A} that
\begin{equation}
 \overline{\partial_0 L_{1A}} = -r^{-1} \overline L_{0A} = - \frac{1}{4}\tau \tilde\nabla^B \lambda_{AB}
 \label{transverseL1A}
 \;.
\end{equation}
With \eq{constraint_g}, \eq{omega_L0A} and \eq{transverseL1A} equation \eq{ODE_transHA} becomes
\begin{eqnarray}
 (\partial_1 + 2r^{-1})\overline{\partial_0H^A} =0
 \label{ODE_transHA2}
 \;.
\end{eqnarray}
Any  solution which is regular at $i^-$ fulfills $\overline{\partial_0 H^A}=\overline{\nabla_0 H^A}=O(r^{-1})$.
 The ODE \eq{ODE_transHA2} admits precisely one such solution, namely
\begin{equation}
 \overline{\partial_0 H^A} =0
 \;.
\end{equation}
We have
\begin{eqnarray*}
 \tilde\nabla^B\lambda_{AB}&=&2\overline L_{0A} \,=\, \overline R_{0A} \,=\, \frac{1}{2}\overline{\partial^2_{00}g_{1A}} - \frac{1}{2}(\partial_1 - \tau)\overline{\partial_0g_{0A}} + \frac{1}{2}\tilde\nabla^B\lambda_{AB}\;,
\\
 0&=& \overline g_{AB}\overline {\partial_0 H^B} \,=\, \overline{\partial^2_{00} g_{1A}} + (\partial_1+ \tau)\overline{\partial_0 g_{0A}} + \tilde\nabla^B\lambda_{AB}
 \;.
\end{eqnarray*}
The combination of both equations yields
\begin{eqnarray}
  \partial_1\overline{\partial_0g_{0A}} + \tilde\nabla^B\lambda_{AB}=0 \quad \text{and} \quad
 \overline{\partial^2_{00}g_{1A}} =- \tau \overline{\partial_0 g_{0A}}
\label{expr_00g1A}
 \;.
\end{eqnarray}


Utilizing the previous results of this section the $\alpha=1$-component of \eq{wave_H+} can be written in our gauge as
\begin{eqnarray}
 (\partial_1 + r^{-1})\overline{\partial_0 H^1}   +\underbrace{ 2(\partial_1 + \tau)\overline  L_{00}  + 2\tilde\nabla^A\overline L_{0A}
 - \overline g^{AB}\overline{\partial_0 L_{AB}}}_{=:f} =0
 \label{ODE_transH1}
 \;,
\end{eqnarray}
where we took into account that owing to Lemma~\ref{lemma_properties_L} we have
\begin{eqnarray}
0\,=\, \overline{\partial_{0}L} &=& 2\overline{\partial_0 L_{01}} + \overline g^{AB}\overline{\partial_0 L_{AB}}
  -  \omega^{AB} \lambda_{AB}
 \;.
 \label{eqn_transtrace}
\end{eqnarray}
%
We show that the source $f$ vanishes.
To do that we compute the $\tilde g$-trace of the $(\mu\nu)=(AB)$-component of \eq{cwe1*} on $\scri^-$.
With \eq{relation_lambda_omega}
we obtain
\begin{eqnarray}
 &&\overline g^{AB} \overline{\Box^{(H)}_g L_{AB}} = 2\overline L_A{}^B \overline L_B{}^A  \quad \Longleftrightarrow \quad
 \nonumber
\\
  &&2( \partial_1 + r^{-1})(\overline g^{AB}\overline{\partial_0L_{AB}}) - 2\lambda^{AB}(\partial_1- r^{-1})\omega_{AB}
 \nonumber
\\
 &&\phantom{xxxxxxxxxxxxxxxxx} + 2  \tau\tilde\nabla^A\overline L_{0A} + \tau^2 \overline L_{00}
 + 2|\omega|^2 =0
 \label{gAB_transLAB}
 \;,
\end{eqnarray}
where we have set $|\omega|^2:= \omega_A{}^B\omega_B{}^A$.

As another intermediate step it is useful to derive a second-order equation for $\ol L_{00}$,
so let us  differentiate \eq{initial_L00_spec_gauge} with respect to $r$,
\begin{eqnarray*}
 (4\partial^2_{11} + 2\tau \partial_1 - \tau^2)\overline L_{00}
= 8 \overline d_{0101} -  4r(\partial_1 +3 r^{-1}) \overline d_{0101} - 2\partial_1\tilde \nabla^A \overline L_{0A} +  \partial_1(\lambda^{AB} \omega_{AB})
 \;.
\end{eqnarray*}
With \eq{initial_d1A1B_spec_gauge}, \eq{initial_d011A_spec_gauge}
\eq{initial_d0101_spec_gauge}, \eq{relation_lambda_omega} and again \eq{initial_L00_spec_gauge} that yields
\begin{eqnarray}
 2(\partial^2_{11} + 3r^{-1} \partial_1 + r^{-2})\overline L_{00}
   =  \lambda^{AB}(\partial_1 - r^{-1}) \omega_{AB}
  - |\omega|^2 - 2 \tilde \nabla^A\partial_1 \overline L_{0A}
 \label{ODE_L00}
 \;.
\end{eqnarray}
Let us return to the source term $f$ in \eq{ODE_transH1}. It satisfies the ODE
\begin{eqnarray*}
 2(\partial_1+r^{-1})f
 &=& 4\partial^2_{11}\overline L_{00} + 6\tau\partial_1 \overline L_{00} -2\tau\tilde\nabla^A \overline L_{0A} + 4\tilde\nabla^A\partial_1 \overline L_{0A}
\\
 && - 2(\partial_1+r^{-1})(\overline g^{AB}\overline{\partial_0L_{AB}})
\\
 &\overset{\eq{gAB_transLAB}}{=}& 4\partial^2_{11}\overline L_{00} + 6\tau\partial_1\overline L_{00} + \tau^2\overline  L_{00}+ 4\tilde\nabla^A\partial_1\overline L_{0A}
\\
 &&- 2\lambda^{AB}(\partial_1- r^{-1})\omega_{AB}   +  2|\omega|^2
\\ &\overset{\eq{ODE_L00}}{=}&  0
 \;.
\end{eqnarray*}
We conclude that
\begin{eqnarray}
  f\equiv 2(\partial_1 + \tau)\overline  L_{00}  + 2\tilde\nabla^A\overline L_{0A}
 - \overline g^{AB}\overline{\partial_0 L_{AB}} = c(x^A) r^{-1}
 \label{ODE_L002}
\end{eqnarray}
for some angle-dependent function $c$.
Regularity at $i^-$ implies $\overline L_{00}=O(1)$ and $\partial_1\overline L_{00}=\overline{\nabla_1 L_{00}}
=O(1)$.
Furthermore, we have (note that $\lambda^{AB}\omega_{AB}=O(r^5)$)
\begin{eqnarray*}
 O(1) &=& \overline{\nabla^A L_{0A}} \,=\, \tilde\nabla^A\overline L_{0A} - \frac{1}{2}\lambda^{AB} \omega_{AB}
 +\tau\overline L_{00}
 \;,
\\
 O(1) &=& \overline g^{AB} \overline{\nabla_0 L_{AB}} \,= \, \overline g^{AB} \overline{\partial_0 L_{AB}} - \lambda^{AB} \omega_{AB}
 \\
   \Longrightarrow && \tilde\nabla^A\overline L_{0A}+\tau\overline L_{00} \,=\, O(1)\;,
    \quad \overline g^{AB} \overline{\partial_0 L_{AB}}\, =\, O(1)
 \;.
\end{eqnarray*}
Therefore the problematic $r^{-1}$-term in the expansion of $f$ needs to vanish, and we
conclude $c=0$.
Then \eq{ODE_transH1} enforces $\overline{\partial_0 H^1}$ to vanish in order to be bounded, i.e.\ altogether we have proven that
\begin{equation}
 \overline{\nabla_{\mu}H^{\nu}}=0
 \label{vanishing_transH}
 \;.
\end{equation}

Recall that $\zeta_{\mu} \,\equiv\,  -4(\nabla_{\nu} L_{\mu}{}^{\nu} - \nabla_{\mu}R/6) =  -4\nabla_{\nu} L_{\mu}{}^{\nu}$.
If we evaluate \eq{wave_H+} on $\scri^-$ (which, as a matter of course, is to be read as an equation for $\zeta_{\mu}$)  and insert \eq{vanishing_transH}, we immediately observe that
\begin{equation}
 \overline \zeta_{\mu} =0
 \label{vanishing_zeta}
 \;.
\end{equation}

\subsection{Vanishing of $\overline W_{\mu\nu\sigma\rho}$}
\label{sect_vanishing_W}

We want to show that the Weyl tensor $W_{\mu\nu\sigma\rho}$ of $g_{\mu\nu}$  vanishes on
$C_{i^-}$, and thus coincides there
with the tensor $\Theta d_{\mu\nu\sigma\rho}$.
The 10 independent components are
\begin{eqnarray*}
 \overline W_{0101}\;, \quad    \overline W_{011A}\;, \quad   \overline W_{010A}\;, \quad   \overline W_{01AB}\;, \quad  \breve { \overline  W}_{1A1B}\;, \quad  \breve { \overline W}_{0A0B}\;.
\end{eqnarray*}
Due to the vanishing of $\overline H{}^{\sigma}$, $\overline{\nabla_{\mu} H^{\sigma}}$ and $\overline R_g$,  \eq{cwe5*} tells us that the tensor $L_{\mu\nu}$
coincides on $C_{i^-}$ with the Schouten tensor. We thus have the formula:
\begin{eqnarray}
 \overline W_{\mu\nu\sigma\rho} = \overline R_{\mu\nu\sigma \rho} -2( \overline g_{\sigma[\mu}\overline L_{\nu]\rho}
 -\overline g_{\rho[\mu}\overline L_{\nu]\sigma} )
 \;.
 \label{Weyl_Riemann_Schouten}
\end{eqnarray}
The following list of Christoffel symbols, or rather of their transverse derivatives, will be useful:
\begin{eqnarray*}
 \overline{\partial_0\Gamma^0_{01}} &=& \overline{\partial_0\Gamma^1_{11}} \, = \, 0
  \;,
\\
 \overline{\partial_0\Gamma^0_{0A}} &\overset{\eq{expr_00g1A}}{=}& - \frac{1}{2}(\partial_1+ \tau)\overline{\partial_0g_{0A}}
  \;,
\\
 \overline{\partial_0\Gamma^0_{AB}} &=& -\frac{1}{2}\partial_1\lambda_{AB}
  \;,
\\
 \overline{\partial_0\Gamma^1_{1A}} &=&  \frac{1}{2}\partial_1\overline{\partial_0g_{0A}}
  \;,
\\
 \overline{\partial_0\Gamma^1_{AB}} &=& \frac{1}{2}\tau \overline g_{AB}\overline{\partial_0 g_{00}}
  + \tilde\nabla_{(A}\overline{\partial_{|0}g_{0|B)}} -\frac{1}{2}\overline{\partial^2_{00}g_{AB}}
   -\frac{1}{2}\partial_1\lambda_{AB}
  \;,
\\
 \overline{\partial_0\Gamma^C_{0A}} &=&   \frac{1}{2}\overline g^{CD}\overline{\partial^2_{00} g_{AD}} -\frac{1}{2}\lambda_{A}{}^{D}\lambda_{D}{}^C
  +\overline g^{CD}\tilde\nabla_{[A}\overline{\partial_{|0} g_{0|D]}}
  \;,
\\
 \overline{\partial_0\Gamma^C_{1A}} &=& \frac{1}{2}\partial_1\lambda_A{}^C
  \;,
\\
 \overline{\partial_0\Gamma^C_{AB}} &=&
  \frac{1}{2}\tau\overline g_{AB}\overline g{}^{CD} \overline{\partial_0 g_{0D}}
 + \tilde\nabla_{(A}\lambda_{B)}{}^C
   -  \frac{1}{2} \tilde\nabla^C\lambda_{AB}
  \;,
\\
 \overline{\partial^2_{00}\Gamma^0_{AB}}  &\overset{\eq{expr_00g1A}}{=}& \frac{1}{2}\tau\overline g_{AB}\overline{\partial^2_{00}g_{01}}
  -\tau \tilde\nabla_{(A}\overline{\partial_{|0}g_{0|B)}}   -\frac{1}{2}\partial_1\overline{\partial^2_{00}g_{AB}}
  \;.
\end{eqnarray*}
We compute the relevant components of the Riemann tensor
$R_{\mu\nu\sigma}{}^{\rho}\equiv \partial_{\nu}\Gamma^{\rho}_{\mu\sigma} - \partial_{\mu}\Gamma^{\rho}_{\nu\sigma}
 +\Gamma^{\alpha}_{\mu\sigma}\Gamma^{\rho}_{\nu\alpha} - \Gamma^{\alpha}_{\nu\sigma}\Gamma^{\rho}_{\mu\alpha}$,
\begin{eqnarray}
 \overline R_{0101} &=&
     0
  \;, \quad
  \overline R_{011A} \,=\,
  0
  \;, \quad
  \overline R_{01AB} \,=\,
  0
  \;, \quad
  \overline R_{1A1B} \,=\,
  0
  \;,
 \label{expressions_Riemann1}
 \\
  \overline R_{010A} &=&
   \frac{1}{2}(\partial_1-\tau)\overline{\partial_0 g_{0A}} - \frac{1}{2}\overline{\partial^2_{00}g_{1A}}
   \,\overset{\eq{expr_00g1A}}{=}\, -\frac{1}{2}\tilde\nabla^B \lambda_{AB}
  \;,
 \\
  \breve{\overline R}_{0A0B} &=&
  (\tilde\nabla_{(A}\overline{\partial_{|0} g_{0|B)}})\breve{}
   - \frac{1}{2}(\overline{\partial^2_{00}g_{AB}})\breve{}
 \label{expressions_Riemann4}
  \;.
\end{eqnarray}
Next, we determine the independent components of the Weyl tensor on $\scri^-$ via \eq{Weyl_Riemann_Schouten} and by taking into account the values
we have found for $\overline L_{\mu\nu}$,
\begin{eqnarray}
 \overline W_{0101} &=&  
 0
 \;, \quad
 \overline W_{011A} \,=\, 
  0
 \;, \quad
 \overline  W_{010A} \,=\,  
  0
 \;,
\\
 \overline W_{01AB} &=& 0
 \;, \quad
 \overline  W_{1A1B} \,=\,  
 0
 \;,
\\
\breve{ \overline  W}_{0A0B} &=& \omega_{AB} + (\tilde\nabla_{(A}\overline{\partial_{|0} g_{0|B)}})\breve{}
   - \frac{1}{2}(\overline{\partial^2_{00}g_{AB}})\breve{}
\;.
 \label{weyl_component_0A0B}
\end{eqnarray}
It remains to determine $\overline{\partial^2_{00}g_{AB}}$. Note that according to \eq{wavered*} the
vanishing of $\overline H{}^{\sigma}$ and $\overline{\nabla_{\mu} H^{\sigma}}$ implies
\begin{eqnarray*}
&& \overline{\partial_0 R_{AB}} \,=\, 2\overline{\partial_0 L_{AB}}
\\
 \Longrightarrow &&\overline{\Box_g R_{AB}} \,=\, 2\overline{ \Box_g L_{AB}} \,=\, 2\overline{ \Box^{(H)}_g L_{AB}}
 \overset{\eq{cwe1*}}{=}  8\omega_{AC} \omega_B{}^C - 2\overline g_{AB}|\omega|^2
  \;.
\end{eqnarray*}
A rather lengthy computation, which uses \eq{relation_lambda_omega}, reveals that this is equivalent to (set $\Delta_{\tilde g} := \tilde\nabla^A\tilde\nabla_A$)
\begin{eqnarray*}
&&  (\partial_1-r^{-1})\overline{\partial_0R_{AB}} - 2(\partial_1-r^{-1})(\omega_{C(A}\lambda_{B)}{}^C)
 + 2 \tau\tilde\nabla_{(A}\overline L_{B)0}
\\
&& \phantom{xxxxxxxxx} + (\partial^2_{11} - \tau\partial_1+  \Delta_{\tilde g})\omega_{AB}      +  \overline g_{AB}(\frac{1}{2} \tau^2 \overline L_{00}
 + |\omega|^2 ) \,=\, 0
 \;.
\end{eqnarray*}
We take its traceless part and invoke Lemma~\ref{lemma_tracelessness},
\begin{eqnarray}
  (\partial_1-r^{-1})(\overline {\partial_0R_{AB}})\breve{}
 + 2 \tau(\tilde\nabla_{(A}\overline L_{B)0})\breve{}
+ (\partial^2_{11} - \tau\partial_1+   \Delta_{\tilde g})\omega_{AB} &=& 0
 \;.
 \label{ODE_transRAB}
\end{eqnarray}
Let us compute $\partial_0 R_{AB}$ on $\scri^-$, which is done by using \eq{relation_lambda_omega}, \eq{expr_00g1A} and the above formulae for the $u$-differentiated Christoffel symbols,
\begin{eqnarray}
 \overline{\partial_0 R_{AB}} &=&\overline{ \partial^2_{00}\Gamma^{0}_{AB} }+ \partial_{1}\overline{\partial_0\Gamma^{1}_{AB}}
  + \tilde\nabla_{C}\overline{\partial_0\Gamma^{C}_{AB}}  - \tilde\nabla_A\overline{\partial_0\Gamma^{C}_{BC}}
 - \tilde\nabla_A\overline{\partial_0\Gamma^{1}_{1B}}
 \nonumber
 \\
  &&   - \tilde\nabla_A\overline{\partial_0\Gamma^{0}_{0B}}   -\overline  \Gamma{}^{0}_{BC}\overline{\partial_0\Gamma^{C}_{0A}} - \overline\Gamma{}^{0}_{AC}\overline{\partial_0\Gamma^{C}_{B0}}
  - \overline\Gamma{}^{C}_{B0}\overline{\partial_0\Gamma^{0}_{AC}}    - \overline\Gamma{}^{C}_{0A}\overline{\partial_0\Gamma^{0}_{BC}}
 \nonumber
 \\
   &&
   - \overline \Gamma{}^{1}_{BC}\overline{\partial_0\Gamma^{C}_{1A}}  - \overline \Gamma{}^{1}_{AC}\overline{\partial_0\Gamma^{C}_{B1}} + \overline \Gamma{}^{0}_{AB}\overline{\partial_0\Gamma^{\mu}_{\mu 0}} + \overline \Gamma{}^{1}_{AB}\overline{\partial_0\Gamma^{\mu}_{\mu 1}}
 \nonumber
 \\
  &=& -(\partial_1-r^{-1})\overline{\partial^2_{00}g_{AB}} + (\partial_{1}-r^{-1})\omega_{AB}
   -  \frac{1}{2} (\Delta_{\tilde g}-\frac{1}{2}\tau^2)\lambda_{AB}
 \nonumber
\\
 &&  -\tau \tilde\nabla_{(A}\overline{\partial_{|0}g_{0|B)}}    -  \frac{1}{2}\tau \lambda_A{}^C\lambda_{BC}    - 2\lambda_{(A}{}^C\omega_{B)C}
   + f(r,x^C)\overline g_{AB}
 \label{computation_transRAB}
\;.
\end{eqnarray}
The traceless part of  $\overline{\partial_0 R_{AB}} $ reads
\begin{eqnarray}
 (\overline{\partial_0 R_{AB}})\breve{} &=& -(\partial_1-r^{-1})(\overline{\partial^2_{00}g_{AB}})\breve{}  + ( \partial_{1}-r^{-1})\omega_{AB}
  -  \frac{1}{2} (\Delta_{\tilde g}-\frac{1}{2}\tau^2)\lambda_{AB}
 \nonumber
\\
 &&   -\tau (\tilde\nabla_{(A}\overline{\partial_{|0}g_{0|B)}})\breve{}
 \label{traceless_transRAB}
\;.
\end{eqnarray}
Next, we apply $2(\partial_1 - r^{-1})$ to the expression \eq{weyl_component_0A0B} which we have found for $\breve{\overline W}_{0A0B}$. With \eq {traceless_transRAB} and \eq{expr_00g1A} we end up with
\begin{eqnarray}
2 (\partial_{1}-r^{-1}) \breve{ \overline  W}_{0A0B} &=&  (\partial_{1}-r^{-1})\big[ 2\omega_{AB}
+2  (\tilde\nabla_{(A}\overline{\partial_{|0} g_{0|B)}})\breve{}
 - (\overline{\partial^2_{00}g_{AB}})\breve{}\big]
  \nonumber
\\
 &=&   (\partial_{1}-r^{-1})\omega_{AB} -4 ( \tilde\nabla_{(A}\overline L_{B)0} )\breve{}  + (\overline{\partial_0 R_{AB}})\breve{}
  \nonumber
\\
 &&       +  \frac{1}{2} (\Delta_{\tilde g}- 2r^{-2})\lambda_{AB}
 \label{relation_weylcomponent}
 \;.
\end{eqnarray}
On the other hand, from the Bianchi identity $\nabla_{[\mu}R_{iA]B}{}^{\mu}=0$, $i=0,1$, we infer
\begin{eqnarray*}
  (\overline{\nabla_{\mu} W_{i(AB)}{}^{\mu}})\breve{} + \frac{1}{2}(\overline{\nabla_i R_{AB}})\breve{} - \frac{1}{2}(\overline{\nabla_{(A} R_{B)i}})\breve{}
  \,=\,0
   \;.
\end{eqnarray*}
Employing further the tracelessness of the Weyl tensor,
\begin{eqnarray*}
 g^{\mu\nu}\nabla_0  W_{\mu AB\nu}=0 \quad \Longrightarrow\quad
  2 (\overline{\nabla_0  W_{0(AB)1}})\breve{} = (\overline{\nabla_0  W_{1A1B}})\breve{}
  \;,
\end{eqnarray*}
we obtain with $\overline R_{\mu\nu}=2\overline L_{\mu\nu}$, $\overline R_g=0$,  Lemma~\ref{lemma_tracelessness} and since the other components of the Weyl tensor are
already known to vanish initially,
\begin{eqnarray}
  2 (\partial_{1}-r^{-1})\breve{\overline W}_{0A0B} \,=\, (\partial_1-r^{-1})\omega_{AB}  + (\overline{\partial_0 R_{AB}})\breve{}
   -2(\tilde\nabla_{(A} \overline L_{B)0})\breve{}
 \label{binachi_weylcomponent}
   \;.
\end{eqnarray}
Combining \eq{relation_weylcomponent} and \eq{binachi_weylcomponent} we are led to
\begin{eqnarray}
  (\Delta_{\tilde g}-\frac{1}{2}\tau^2)\lambda_{AB}  -4( \tilde\nabla_{(A}\overline L_{B)0} )\breve{}
 &=& 0
\label{box_rel1}
 \;.
\end{eqnarray}
We apply $(\partial_1+r^{-1})$  and use \eq{relation_lambda_omega} to conclude that
\begin{eqnarray}
(\Delta_{\tilde g}-\frac{1}{2}\tau^2)\omega_{AB}  +2(\partial_1+r^{-1})( \tilde\nabla_{(A}\overline L_{B)0} )\breve{}
  &=& 0
 \;,
 \label{box_rel2}
\end{eqnarray}
which will prove to be a useful relation.
Next we apply  $(\partial_1-r^{-1})$  to \eq{relation_weylcomponent}.
With \eq{relation_lambda_omega}, \eq{ODE_transRAB}, \eq{box_rel1} and \eq{box_rel2} we end up with
\begin{eqnarray*}
2(\partial_1-r^{-1})^2 \breve{ \overline  W}_{0A0B}
 &=& (\partial^2_{11} -2r^{-1}\partial_1+2r^{-2})\omega_{AB} -4 (\partial_1-r^{-1})(\tilde\nabla_{(A}\overline L_{B)0})\breve{}
\\
 &&  
  -(\Delta_{\tilde g}- 2r^{-1})(\omega_{AB}+ r^{-1}\lambda_{AB})
   +(\partial_1-r^{-1})(\overline{\partial_0 R_{AB}})\breve{}
\\
 &=& 0
\\
  \Longrightarrow \quad \breve{ \overline  W}_{0A0B}  &=& c_{AB}(x^C)r^2 + d_{AB}(x^C)r = c_{AB}(x^C)r^2
 \;,
\end{eqnarray*}
for any regular solution satisfies $  \breve{ \overline  W}_{0A0B}  =O(r^2)$
in adapted coordinates.

We have $\omega_{AB}=O(r^4)$ and $\lambda_{AB}=O(r^5)=\overline{\partial_0 g_{AB}}$.
A regular solution
satisfies  $  O(r^2) =( \overline{\nabla_{(A} L_{B)0}})\breve{} = \tilde\nabla _{(A}\overline L_{B)0}$.
Similarly, we have $O(r^2) =\overline{ \nabla_0 R_{AB}} = \overline{\partial_0 R_{AB}} - 2\overline \Gamma{}^C_{0(A}\overline R_{B)C}$, which implies
$(\overline{\partial_0 R_{AB}})\breve{}=O(r^2)$,
so the right-hand side of \eq{binachi_weylcomponent} is $O(r^2)$, consequently $\breve{ \overline  W}_{0A0B} =O(r^3)$, whence $c_{AB}=0$ and
\begin{equation*}
 \breve{ \overline  W}_{0A0B} =0
 \;.
\end{equation*}

\subsection{Validity of equation \eq{conf5} on $C_{i^-}$}
\label{sect_validity}

We need to show that \eq{conf5} holds at at least one point. In fact,
since $\overline \Theta$ vanishes and $\overline{\nabla_{\mu}\Theta}$ is null,  one  immediately observes that
it is satisfied on the whole initial surface $C_{i-}$.
%
%

\subsection{Vanishing of $\overline \Upsilon_{\mu}$}

Using the constraint equations \eq{constraint_g} it is easily checked
that the components $\mu=1,A$ of $\overline \Upsilon_{\mu} \equiv \overline{\nabla_{\mu} s} + \overline L_{\mu}{}^{\nu}\overline{\nabla_{\nu}\Theta}$ vanish.
To show that also the $\mu=0$-component vanishes, we need to compute the value of the transverse derivative of $s$ on $\scri^-$,
which is accomplished via the CWE \eq{cwe2*},
\begin{eqnarray*}
 \overline{\Box^{}_g s}\,=\,0 \quad \Longleftrightarrow \quad  (\partial_1+r^{-1})\overline{\partial_0 s} \,=\,0
 \;.
\end{eqnarray*}
The function $\overline{\partial_0 s}$ is bounded. Thus
\begin{equation}
 \overline{\partial_0 s} \,=\,0
 \;,
 \label{vanish_trans_s}
\end{equation}
and the vanishing of $\overline \Upsilon_{\mu}$ is ensured.

\subsection{Vanishing of $\overline \Xi_{\mu\nu}$}
\label{sect_vanishing_Xi}

We consider
\begin{equation*}
 \overline\Xi_{\mu\nu} \equiv \overline{ \nabla_{\mu}\nabla_{\nu}\Theta +\Theta\,   L_{\mu\nu} -  s\, g_{\mu\nu}}
 = \overline{\partial_{\mu}\partial_{\nu}\Theta} - \overline\Gamma{}^0_{\mu\nu}\overline{\partial_0\Theta}  + 2\overline g_{\mu\nu}
 \;.
\end{equation*}
First of all we need to determine the value of $\overline{\partial_0\Theta}$, which is not part of the initial data. It can be derived from the CWE.
Evaluation of \eq{cwe3*} on $\scri^-$ gives
\begin{eqnarray*}
  \overline{\Box^{}_g \Theta} = 4\overline s \quad \Longleftrightarrow \quad
   (\partial_1+ r^{-1})\overline{\partial_0 \Theta}  = -4
 \;.
\end{eqnarray*}
For any sufficiently regular solution of the CWE  the function $\overline{\partial_0\Theta}$ is bounded near the vertex,
and there is precisely one such solution,
\begin{equation}
 \overline{\partial_0 \Theta} \, = \, -2r\;.
 \end{equation}
One straightforwardly checks that $\overline\Xi_{\mu\nu}=0$ for $(\mu\nu)\ne (00)$.
To determine $\overline\Xi_{00}$ we need to compute the second-order transverse derivative of $\Theta$ first.
This is done via the CWE \eq{cwe3*},
\begin{eqnarray*}
 \overline{\partial_0\Box_{g}\Theta} = 4\overline{\partial_0 s} \overset{\eq{vanish_trans_s}}{=}0 \quad \Longleftrightarrow \quad (\partial_1 + r^{-1})\overline{\partial^2_{00}\Theta} - 2r^{-1}=0
\;,
\end{eqnarray*}
where we took into account that $\overline{\partial_0 g_{1\mu}}=0$, $\overline g^{AB} \overline{\partial_0 g_{AB}}=0$ , $\overline{\partial^2_{00}g_{11}}=0$,
as well as the formulae for the $u$-differentiated Christoffel symbols. The general solution of the ODE is $\overline{\partial^2_{00}\Theta}= 2 +cr^{-1}$.
For any sufficiently regular  solution   $\overline{\partial^2_{00}\Theta}= \overline{\nabla_0\nabla_0\Theta}$ is bounded, and we conclude
\begin{equation*}
 \overline{\partial^2_{00}\Theta}\, =\, 2
 \;,
\end{equation*}
which guarantees the vanishing of $\overline \Xi_{00}$.

\subsection{Vanishing of $\overline \varkappa_{\mu\nu\sigma}$}
\label{sect_vanishing_xi}

Recall the definition of the tensor
\begin{eqnarray*}
 \varkappa_{\mu\nu\sigma} \equiv 2\nabla_{[\sigma} L_{\nu]\mu} - \nabla_{\kappa}\Theta d_{\nu\sigma\mu}{}^{\kappa}
 \;.
\end{eqnarray*}
Due to the symmetries $\varkappa_{\mu(\nu\sigma)}=0$, $\varkappa_{[\mu\nu\sigma]}=0$ and $\varkappa_{\nu\mu}{}^{\nu}=0$ (since $\overline\zeta_{\mu}=0$ and $L=0$)
its independent components on the initial surface are
\begin{eqnarray*}
   \overline \varkappa_{11A}\;, \quad \overline \varkappa_{A1B}\;, \quad \overline \varkappa_{01A}\;, \quad \overline \varkappa_{ABC}\; \quad
 \overline\varkappa_{00A}\;, \quad  \overline\varkappa_{A0B}\;.
\end{eqnarray*}
We find (recall that $\overline L_{1\mu}=0$ and $\overline L_{0A} = \frac{1}{2}\tilde\nabla^B\lambda_{AB}$),
\begin{eqnarray*}
 \overline \varkappa_{11A} &=& 0
 \;,
\\
 \overline \varkappa_{A1B} &=& -(\partial_1-r^{-1}) \omega_{AB} - 2r\overline  d_{1A1B} \overset{\eq{initial_d1A1B_spec_gauge}}{=} 0
\;,
\\
 \overline \varkappa_{01A} &=&  -\partial_{1} \overline  L_{0A} + 2r\overline  d_{011A} \overset{\eq{expression_d011A}}{=}   0
\;,
\\
 \overline \varkappa_{ABC} &=& 2 \tilde \nabla_{[C} \omega_{B]A} -\tau  \overline g_{A[B} \overline L_{C]0}   - 2r \overline  d_{1ABC}
\\
 &=&  2 \tilde \nabla_{[C} \omega_{B]A}    - 2\tilde\nabla_D \omega _{[B}{}^D \overline  g_{C]A}
 \overset{\mathrm{tr}(\omega)=0}{=} 0
 \;,
\end{eqnarray*}
where the first equal sign in the last line follows from \eq{relation_d5}, \eq{initial_d011A_spec_gauge},  \eq{constraint_g} and~\eq{relation_lambda_omega}.

To prove the vanishing of the remaining components,
\begin{eqnarray*}
    \overline \varkappa_{A0B} &=& \tilde  \nabla_{B} \overline L_{0A} - \frac{1}{2}\lambda_B{}^C\omega _{ AC} +\frac{1}{2}\tau \overline g_{AB}\overline L_{00}   -  \overline{\nabla_{0} L_{AB}} + 2r \overline d_{0BA1}
  \;,
\\
  \overline \varkappa_{00A} &=&  \tilde\nabla_{A}\overline  L_{00} - \lambda_A{}^B\overline L_{0B} - \overline{ \nabla_{0} L_{0A}} + 2r \overline d_{010A}
 \;,
\end{eqnarray*}
is somewhat more involved as it requires the knowledge of certain transverse derivatives
of $L_{\mu\nu}$ on $\scri^-$. These can be extracted from  \eq{cwe1*},
\begin{eqnarray*}
 \overline{\Box_g L_{AB}} &=&    \overline{\Box^{(H)}_g L_{AB}}  \,=\, 4\omega_{AC}\omega_{B}{}^C - \overline g_{AB} |\omega|^2
\;,
\\
 \overline{\Box_g L_{0A}} &=&  \overline{\Box^{(H)}_g L_{0A}} \,=\, 4\omega_{A}{}^{B}\overline L_{0B}
 \;.
\end{eqnarray*}
We employ the facts, established above, that the Weyl tensor vanishes on $C_{i^-}$ and that $L_{\mu\nu}$ coincides there with the Schouten tensor,
to compute the action of $\Box_g$ on $L_{AB}$ and $L_{0A}$,
\begin{eqnarray*}
 \overline{\Box_g L_{AB}} &=& 2 (\partial_1- r^{-1})\overline{\nabla_0 L_{AB}} +\partial_1(\partial_1- \tau) \omega_{AB} +(\Delta_{\tilde g}-\frac{1}{2}\tau^2 )\omega _{AB}
\\
&&  + 2\tau  \tilde \nabla_{(A}\overline  L_{B)0}
   -\tau  \lambda_{(A}{}^C  \omega_{B)C} + \frac{1}{2}\tau^2 \overline  g_{AB} \overline  L_{00 }
\\
 &\overset{\eq{box_rel2}}{=}&
  2 (\partial_1- r^{-1})\overline{\nabla_0 L_{AB}} +\partial_1(\partial_1- \tau) \omega_{AB}  - 2(\partial_1-r^{-1})( \tilde\nabla_{(A}\overline L_{B)0} )\breve{}
\\
  &&  + \tau  \overline g_{AB} \tilde \nabla^C \overline  L_{0C}
   -\tau  \lambda_{(A}{}^C  \omega_{B)C} + \frac{1}{2}\tau^2 \overline  g_{AB} \overline  L_{00 }
\;,
\\
  \overline{\Box_g L_{0A}} &=& 2 \partial_1\overline{\nabla_0 L_{0A}}+ (\partial_1+r^{-1}) (\partial_1-r^{-1}) \overline L_{0A} -2\omega_A{}^B\overline  L_{0B}
\\
 && +(\Delta_{\tilde g}- r^{2}) \overline L_{0A} + \tau  \tilde \nabla_A \overline L_{00 }
  -\lambda_B{}^C\tilde \nabla^B \omega_{AC}  - \tau  \lambda_A{}^B\overline  L_{0 B}
\\
 &\overset{\eq{relation_lambda_omega}}{=}& 2 \partial_1\overline{\nabla_0 L_{0A}}
  - (\partial_1-r^{-1}) \tilde\nabla^B \omega_{AB}   -2\omega_A{}^B\overline  L_{0B}
\\
 && +(\Delta_{\tilde g} + r^{-2})\overline L_{0A} + \tau  \tilde \nabla_A \overline L_{00 }
  - \lambda_B{}^C\tilde \nabla^B \omega_{AC}  - \tau \lambda_A{}^B\overline  L_{0 B}
 \;.
\end{eqnarray*}
With these expressions,  Lemma~\ref{lemma_tracelessness},
 \eq{constraint_g}-\eq{relation_lambda_omega} and \eq{relation_d2}-\eq{relation_d6} we find
\begin{eqnarray*}
 2  (\partial_1-r^{-1}) \overline \varkappa_{A0B} &=&  2(\partial_1-r^{-1})\tilde  \nabla_{B} \overline L_{A0}
  - \omega _{ A}{}^C(\partial_1-r^{-1})\lambda_{BC}  +  \tau \overline g_{AB}\partial_1 \overline L_{00}
\\
 && - \lambda_B{}^C(\partial_1-\tau )\omega _{ AC}  + 4 r\partial_1\overline d_{0BA1}
 -  2  (\partial_1-r^{-1})\overline{\nabla_{0} L_{AB}}
\\
 &=&  2(\partial_1-r^{-1})\tilde  \nabla_{[B} \overline L_{A]0}
 +\overline g_{AB} (\partial_1+3r^{-1})\tilde  \nabla^C \overline L_{C0}
\\
 &&  - \lambda_B{}^C\partial_1\omega _{ AC} + \tau  \lambda_{[B}{}^C\omega _{ A]C}
 - 2\omega_{AC}\omega_{B}{}^C + \overline g_{AB} |\omega|^2
\\
 && +  \tau \overline g_{AB}(\partial_1+r^{-1}) \overline L_{00} -2 r \overline g_{AB}(\partial_1+\tau)\overline d_{0101}
  + 2 r\partial_1\overline d_{01AB}
\\
 &=&
   - (\partial_1 + r^{-1})  (\lambda_{(A}{}^C\omega _{ B)C})\breve{}
  - 4(\omega_{CA}\omega_{B}{}^C)\breve{}
\\
 &=& 0
 \;,
\end{eqnarray*}
as well as
\begin{eqnarray*}
 2\partial_1 \overline \varkappa_{00A}
 &=&   2 \partial_1\tilde\nabla_{A}\overline  L_{00}
 +  4( \omega_{A}{}^B+ r^{-1} \lambda_A{}^B )\overline L_{0B} + 2  \lambda_A{}^B\tilde\nabla^C \omega_{BC}
  - 2 \partial_1\overline{ \nabla_{0} L_{0A}}
\\
  && + 4 r(\partial_1 + r^{-1}) \overline d_{010A}
\\
 &=&  \frac{1}{2} \tilde\nabla_{A}( \omega_{BC}\lambda^{BC}) + 2(r^{-1} \lambda_A{}^B  -  \omega_{A}{}^B) \overline L_{0B}
- \tilde\nabla_{A} \tilde \nabla^B \overline L_{0B}
- \lambda_{C}{}^B\tilde\nabla^C \omega_{AB}
\\
 &&+  2\lambda_{A}{}^B\tilde\nabla^C \omega_{BC}
 - (\partial_1-r^{-1}) \tilde\nabla^B \omega_{AB}  +(\Delta_{\tilde g}+r^{-2})\overline L_{0A}
\\
  && + 2r \tilde\nabla^B\overline d_{01AB} - 2 r \tilde\nabla^B\overline d_{1A1B}
  + 4  \overline d_{011A}
   + 4 r \lambda_A{}^B\overline d_{011B}
\\
 &=& -\tilde\nabla^B( \lambda_{C(A} \omega_{B)}{}^C)\breve{}
  -  r^{-2}\overline L_{0A}  -2 \tilde\nabla_{[A} \tilde \nabla_{B]} \overline L_{0}{}^B
\\
 &=& 0
 \;.
\end{eqnarray*}
Due to regularity we have $\overline \varkappa_{A0B}=O(r^2)$ and $\overline \varkappa_{00A}=O(r)$,
so  the only remaining possibilities are
\begin{equation*}
 \overline \varkappa_{A0B}\,=\,0 \quad \text{and} \quad \overline \varkappa_{00A} \,=\, 0
 \;.
\end{equation*}

\subsection{Vanishing of $\overline {\nabla_{\rho}d_{\mu\nu\sigma}{}^{\rho}}$}

The independent components of $\overline {\nabla_{\rho}d_{\mu\nu\sigma}{}^{\rho}}$, which by Lemma~\ref{lemma_properties_d} is antisymmetric in its first two indices, trace-free and
satisfies the first Bianchi identity, are
\begin{eqnarray*}
  \overline {\nabla_{\rho}d_{0A0}{}^{\rho}}\;, \quad
  \overline {\nabla_{\rho}d_{0A1}{}^{\rho}}\;, \quad
  \overline {\nabla_{\rho}d_{0AB}{}^{\rho}}\;, \quad
  \overline {\nabla_{\rho}d_{1A1}{}^{\rho}}\;, \quad
  \overline {\nabla_{\rho}d_{1AB}{}^{\rho}}\;, \quad
  \overline {\nabla_{\rho}d_{ABC}{}^{\rho}}\;.
\end{eqnarray*}
We need to show that they vanish altogether.
Let us start with those components which do not involve transverse derivatives. Then their vanishing follows immediately from the constraint equations
\eq{constraint_g}-\eq{relation_lambda_omega} and \eq{relation_d2}-\eq{relation_d6},
\begin{eqnarray*}
 \overline {\nabla_{\rho}d_{0A1}{}^{\rho}} &=&
  -(\partial_{1}+r^{-1})\overline d_{010A} +\frac{1}{2} \tilde \nabla^{B}\overline d_{01AB}  - \frac{1}{2}\tilde \nabla^{B}  \overline d_{1A1B}
   + \frac{1}{2}\tilde \nabla_A\overline d_{0101}
  \\
   &&  + r^{-1}\overline  d_{011A} +  \lambda_{A}{}^B \overline d_{011B} \,=\, 0
 \;,
\\
 \overline {\nabla_{\rho}d_{1A1}{}^{\rho}} &=&
  -(\partial_{1}+ 3r^{-1}) \overline d_{011A} +  \tilde \nabla^{B}\overline d_{1A1B}
    \,=\,  0
 \;.
\end{eqnarray*}
To determine the remaining components we first of all need to compute the transverse derivatives.
This is done by evaluating the CWE \eq{cwe4*} on $C_{i^-}$,
\begin{eqnarray}
  \overline{\Box_g d_{\mu\nu\sigma\rho} } \,=\, \overline{\Box^{(H)}_g d_{\mu\nu\sigma\rho} } \,=\, 0
 \;.
 \label{wave_d_cone}
\end{eqnarray}
Moreover, we will exploit the Lemmas~\ref{lemma_properties_d} and \ref{lemma_tracelessness},
the fact that the Weyl tensor vanishes on $C_{i^-}$, and that $ L_{\mu\nu}$ coincides there with the Schouten tensor,
i.e.\ that \eq{mconf6} holds initially.

Invoking \eq{constraint_g}-\eq{relation_lambda_omega} and \eq{relation_d2}-\eq{relation_d6}
we compute
\begin{eqnarray}
 (\partial_1 - r^{-1}) \overline{ \nabla_{\rho}d_{1AB}{}^{\rho}} &=& -(\partial_1 - r^{-1}) \overline{\nabla_{0}d_{1A1B}} - \frac{1}{2}(\partial_1 - r^{-1})^2   \overline d_{1A1B}
 \nonumber
 \\
  &&    + 2\tau ( \tilde \nabla_{(A} \overline d_{B)110})\breve{}
   - ( \tilde \nabla_{(A} \tilde\nabla^C \overline  d_{B)1C1})\breve{}
  \;.
   \label{div_d1AB}
\end{eqnarray}
 With \eq{expressions_Riemann1} we further find
\begin{eqnarray*}
 \overline{\Box_g d_{1A1B}}
 &=& 2 (\partial_1 - r^{-1})\overline{\nabla_0  d_{1A1B}} + (\partial^2_{11} - 2r^{-1}\partial_1 )  \overline d_{1A1B} + \Delta_{\tilde g}\overline  d_{1A1B}
\\
 &&   +2\tau\tilde \nabla^C \overline d_{1(AB)C}    + \frac{1}{2}\tau^2 \overline  g^{CD} \overline d_{ACBD}  + 2\tau\tilde\nabla_{(A} \overline d_{B)101}
 + \frac{1}{2}\tau^2\overline  g_{AB}\overline d_{0101}
\\
   &=& 2 (\partial_1 - r^{-1})\overline {\nabla_0  d_{1A1B}} + ( \Delta_{\tilde g} + \partial^2_{11} - 2r^{-1}\partial_1 ) \overline d_{1A1B}
   - 4\tau(\tilde \nabla_{(A}\overline  d_{B)110})\breve{}
 \;.
\end{eqnarray*}
The transverse derivative in \eq{div_d1AB} is eliminated via $\overline{\Box_g d_{1A1B}}=0$,
\begin{eqnarray*}
 (\partial_1 - r^{-1}) \overline{ \nabla_{\rho}d_{1AB}{}^{\rho}} &=&  \frac{1}{2}\Delta_{\tilde g} \overline d_{1A1B}
  - r^{-2}  \overline d_{1A1B}   - ( \tilde \nabla_{(A} \tilde\nabla^C \overline  d_{B)1C1})\breve{}
  \;.
\end{eqnarray*}
We need an expression for the $\Delta_{\tilde g}$-term, which can be derived
from \eq{initial_d1A1B_spec_gauge}, \eq{box_rel2}, \eq{constraint_g} and \eq{relation_lambda_omega}
as follows,
\begin{eqnarray}
 \Delta_{\tilde g}  \overline d_{1A1B} &=&
  -\frac{1}{2}r^{-1}(\partial_1+ r^{-1})\Delta_{\tilde g} \omega _{AB}
   \nonumber
\\
 &=& \frac{1}{2}r^{-1}(\partial_1+ r^{-1})[ 2(\partial_1+r^{-1})( \tilde\nabla_{(A}\overline L_{B)0} )\breve{} - \frac{1}{2}\tau^2\omega_{AB}]
 \nonumber
\\
&=& 2r^{-2} \overline d_{1A1B}  + 2( \tilde\nabla_{(A}\tilde\nabla^C \overline d_{B)1C1})\breve{}
\label{Delta_d1A1B}
  \;.
\end{eqnarray}
Plugging that in we are led to the ODE
\begin{equation*}
 (\partial_1 - r^{-1}) \overline{ \nabla_{\rho}d_{1AB}{}^{\rho}} \,=\,0
  \;.
\end{equation*}
For any sufficiently regular  solution of the CWE we have $ \overline{ \nabla_{\rho}d_{1AB}{}^{\rho}} =O(r^2)$ and hence
\begin{equation*}
 \overline{ \nabla_{\rho}d_{1AB}{}^{\rho}} \,=\,0
  \;.
\end{equation*}

To show that the other components of $\nabla_{\rho}d_{\mu\nu\sigma}{}^{\rho}$ vanish initially, we proceed in a similar manner.
In particular, we shall make extensively use of the  constraint equations
\eq{constraint_g}-\eq{relation_lambda_omega}
(and also of their non-integrated counterparts \eq{initialdata_d1}-\eq{initialdata_d5}),
\eq{relation_d2}-\eq{relation_d6}  and of the expressions  \eq{expressions_Riemann1}-\eq{expressions_Riemann4} we computed for the components of the Riemann tensor.

Let us establish the vanishing of $  \overline {\nabla_{\rho}d_{ABC}{}^{\rho}}$. By \eq{wave_d_cone} we have
$  \overline{\Box_g d_{1ABC} } = 0$  on $\scri^-$ with
\begin{eqnarray}
\overline{  \Box_g d_{1ABC}}
   &=& 2 (\partial_1 - \tau)\overline{  \nabla_0 d_{1ABC}}   + (\partial^2_{11} - 4r^{-1}\partial_1 + r^{-2})\overline  d_{1ABC} +  \Delta_{\tilde g} \overline d_{1ABC}
 \nonumber
  \\
   &&  -\tau \tilde\nabla^D \overline d_{D ABC} + \tau \tilde\nabla_A\overline d_{10 BC} + 2\tau \tilde\nabla_{[B} \overline d_{C]0A1} + 2\tau  \tilde\nabla_{[B}\overline d_{C]1A1}
 \nonumber
   \\
   &&
    +2\tilde\nabla_D(\lambda_{[B}{}^D \overline d_{C]1A1})
 -\tau \lambda_{[B}{}^D \overline d_{C]1AD}  - \tau \lambda_{[B}{}^D \overline d_{C]D A1}
 \nonumber
    \\
    && - \frac {1}{2}\tau^2 \overline d_{0A BC}   - \tau^2 \overline g_{A[B} \overline d_{C]010}  - \tau^2 \overline g_{A[B} \overline d_{C]110}
   -\tau \lambda_{A[B} \overline d_{C]110}
 \nonumber
\\
 &=& 2 (\partial_1 - \tau)\overline{\nabla_0 d_{1ABC} }  + 2\overline  g_{A[B} (\partial^2_{11} -5r^{-2})  \overline d_{C]110}
 + 2 \overline  g_{A[B}   \Delta_{\tilde g}  \overline d_{C]110}
 \nonumber
  \\
   && -3 \tau   \overline g_{A[B} \tilde\nabla_{C]}\overline d_{0101}
 - \tau \tilde\nabla_A \overline d_{01 BC}    +\tau \tilde\nabla_{[B} \overline d_{C]A01} + \tau  \tilde\nabla_{[B} \overline d_{C]1A1}
 \nonumber
   \\
   &&    +2\tilde\nabla_D(\lambda_{[B}{}^D \overline  d_{C]1A1})
   -2\tau \overline  g_{A[B}\lambda_{C]}{}^D \overline d_{011D}
  - 2\tau \lambda_{A[B}\overline d_{C]110}
 \label{box_d1ABC_cone}
 \;.
\end{eqnarray}
We determine
\begin{eqnarray*}
  \overline {\nabla_{\rho}d_{ABC}{}^{\rho}} &=& \overline{\nabla_{0}d_{ABC1}} +2  \overline g_{C[A} (\partial_1+r^{-1}) \overline d_{B]010}
   \\
   &&   - 2\overline g_{C[A} \tilde \nabla_{B]} \overline d_{0101}
  -\overline  g_{C[A}\lambda_{B]}{}^D \overline d_{011D}
    +\lambda_{C[A} \overline d_{B]110}
    \\
   &=&\overline{ \nabla_{0}d_{ABC1}} +  \overline g_{C[A} \tilde\nabla^D \overline d_{B]D01} + \overline g_{C[A} \tilde\nabla^D\overline d_{B]11D}
 +\tau \overline g_{C[A}  \overline d_{B]110}
   \\
   &&   - \overline g_{C[A} \tilde \nabla_{B]}\overline  d_{0101}
  +\overline  g_{C[A}\lambda_{B]}{}^D \overline d_{011D}
    +\lambda_{C[A}\overline d_{B]110}
\;.
\end{eqnarray*}
Due to the constraint equations  that yields
\begin{eqnarray*}
&& \hspace{-3em} 2(\partial_1- \tau)  \overline {\nabla_{\rho}d_{CBA}{}^{\rho}}
\\
 &=&  2(\partial_1-\tau) \overline{ \nabla_{0}d_{1ABC}}
 + 2 \overline g_{A[B} \tilde\nabla^D(\partial_{|1|}-\tau) \overline d_{C]1D1}
\\
 && -   2  \overline g_{A[B} \tilde\nabla^D(\partial_{|1|} + r^{-1}) \overline d_{C]D01}
 - 2\tau \overline g_{A[B}(\partial_{|1|}+3r^{-1} )   \overline d_{C]110}
   \\
   &&   + 2 \overline g_{A[B} \tilde \nabla_{C]} (\partial_1+3r^{-1} )\overline d_{0101}
   -2 \overline  g_{A[B}\lambda_{C]}{}^D (\partial_1+ 3r^{-1}) \overline d_{011D}
\\
   &&  - 2\lambda_{A[B}(\partial_{|1|}+ 3r^{-1} ) \overline d_{C]110}
  + 3\tau  \overline g_{A[B} \tilde\nabla^D \overline d_{C]D01} + 4\tau^2 \overline g_{A[B} \overline d_{C]110}
\\
 &&   - 3\tau  \overline g_{A[B} \tilde \nabla_{C]} \overline d_{0101} +4\tau  \overline  g_{A[B}\lambda_{C]}{}^D  \overline d_{011D}
  + 4\tau \lambda_{A[B} \overline d_{C]110}
\\
 && +\underbrace{4\omega_{A[B} \overline d_{C]110} + 4\overline  g_{A[B}\omega_{C]}{}^D \overline d_{011D}}_{=0}
\\
 &=&  2(\partial_1-\tau) \overline { \nabla_{0}d_{1ABC}}
 + 2 \overline g_{A[B} \tilde\nabla^D(\partial_{|1|}-2\tau) \overline d_{C]1D1}
 +4\tau  \overline  g_{A[B}\lambda_{C]}{}^D  \overline d_{011D}
\\
 &&    + 4\tau\lambda_{A[B} \overline d_{C]110}
 - 3\tau  \overline g_{A[B} \tilde \nabla_{C]} \overline d_{0101}
  + 3\tau  \overline g_{A[B}\tilde\nabla^D \overline d_{C]D01}
    + \frac{7}{2}\tau^2 \overline g_{A[B} \overline d_{C]110}
\\
 &&
 +    \overline g_{A[B} \tilde \nabla_{C]}(\lambda^{DE}\overline d_{1D1E})
  +   2  \overline g_{A[B} \Delta_{\tilde g} \overline d_{C]110}
  +   \overline g_{A[B} \tilde\nabla^D(\overline d_{C]1F1} \lambda_D{}^F)
\\
  &&  -     \overline g_{A[B} \tilde\nabla^D  (\lambda_{C]}{}^E\overline d_{1D1E})
\underbrace{ -2 \overline  g_{A[B}\lambda_{C]}{}^E  \tilde\nabla^D\overline d_{1D1E}
  - 2\lambda_{A[B} \tilde\nabla^D\overline d_{C]1D1}}_{=0}
 \;.
\end{eqnarray*}
With $\overline{\Box_g d_{1ABC}}=0$ we eliminate the transverse derivative. Employing further  \eq{initial_d1A1B_spec_gauge}, \eq{initial_d011A_spec_gauge} and \eq{relation_lambda_omega} we end up with
\begin{eqnarray}
 && \hspace{-3em}  2(\partial_1- 2r^{-1})  \overline {\nabla_{\rho}d_{CBA}{}^{\rho}}
\nonumber
\\
 &=&   -\frac{1}{4}\tau^2( \partial_1- r^{-1}) (\overline  g_{A[B}  \tilde\nabla^D\omega_{C]D}  - \tilde\nabla_{[B} \omega_{C]A})
\nonumber
  \\
   &&
 +  \tau \tilde\nabla_A\overline  d_{01 BC}     - \tau \tilde\nabla_{[B}\overline  d_{C]A01}  + 3\tau  \overline g_{A[B}\tilde\nabla^D \overline d_{C]D01}
\nonumber
\\
 &&  +6\tau  \overline  g_{A[B}\lambda_{C]}{}^D  \overline d_{011D}  + 6\tau \lambda_{A[B} \overline d_{C]110}
\nonumber
\\
 &&  -    \overline g_{A[B} \tilde\nabla^D  (\lambda_{C]}{}^E\overline d_{1D1E})
 +    \overline g_{A[B} \tilde \nabla_{C]}(\lambda^{DE}\overline d_{1D1E})
 - \tilde\nabla_D(\lambda_{[B}{}^D  \overline d_{C]1A1})
\nonumber
\\
 &&
  +    \overline g_{A[B} \tilde\nabla^D(\overline d_{C]1E1}\lambda_D{}^E)   - \tilde\nabla_D(\lambda_{[B}{}^D \overline  d_{C]1A1})
\nonumber
\\
 &=&
\label{ODE_dCBA}0
 \;,
\end{eqnarray}
since the terms in each line add up to zero, as one checks e.g.\ by introducing an orthonormal frame for $\tilde g$.
By regularity we have $ \overline {\nabla_{\rho}d_{ABC}{}^{\rho}} =O(r^3)$, so \eq{ODE_dCBA} enforces
\begin{equation*}
  \overline {\nabla_{\rho}d_{ABC}{}^{\rho}}  \,=\,0
 \;.
\end{equation*}

To check the vanishing of $\overline{\nabla_{\rho}d_{0A0}{}^{\rho}}$ we start with the relation $\overline{\Box_g d_{010A} } =0$,
and compute
\begin{eqnarray*}
 \overline { \Box_g d_{010A}}
&=& 2  \overline  L_{0}{}^{B} \overline  d_{0A1B}  +2 \overline L_{0}{}^{B}\overline d_{01AB}+ 2 \overline  L_{0A} \overline d_{0101} + 2 \partial_1\overline {\nabla_0   d_{010A}}    - \frac{1}{2}\tau\lambda_A{}^B \overline d_{ 010B }
\\
 &&  +(\Delta_{\tilde g} +  \partial^2_{11}  -  \frac{5}{4}\tau^2  ) \overline d_{010A}
  - \lambda_B{}^C \tilde \nabla^B\overline d_{C10A}   + \lambda_B{}^C\tilde \nabla^B \overline d_{01AC}
\\
 &&   +\frac{1}{2}  \lambda_B{}^C\lambda^{BD}\overline d_{1C AD}
  -\frac{1}{4}|\lambda|^2\overline d_{01 1A}
+ ( \tau\lambda_A{}^B + \lambda_A{}^C\lambda_C{}^B)\overline d_{011B }
\\
 &&
 - \tau   \tilde \nabla^B   \overline d_{0A 0B}   - \frac{1}{2}\tau\lambda^{BC}\overline  d_{0BAC}
   + \tau  \tilde \nabla_A   \overline d_{0101 }
 + \lambda_A{}^B\tilde \nabla_B  \overline  d_{0101 }
\\
 &=&2 \partial_1\overline {\nabla_0   d_{010A}}
   +  \tilde\nabla_A \tilde\nabla^B\overline d_{011B} -  \frac{1}{2} \Delta_{\tilde g} \overline d_{011A}
 + \frac{1}{4} \tilde\nabla^B(\lambda_A{}^C\overline d_{1B1C})
\\
  &&  + \frac{1}{4}  \tilde\nabla_A(\lambda^{BC}\overline d_{1B1C})
 -  \frac{1}{2} \tilde\nabla^B  (\partial_{1}-5r^{-1})\overline d_{1A1B}     -\tau \tilde\nabla^B \overline d_{01AB}
\\
 && -2\omega_{A}{}^B\overline d_{011B}
  +\lambda_A{}^C \tilde\nabla^B\overline d_{1B1C}  - \tau   \tilde \nabla^B \overline   d_{0A 0B}
  - \frac{3}{4} \tilde \nabla^B(\lambda_B{}^C \overline d_{1A1C})
\\
 &&      + \frac{3}{2}\tilde \nabla^B(\lambda_B{}^C\overline  d_{01AC})
  + \frac{3}{2}\tilde \nabla_B(\lambda_A{}^B \overline   d_{0101 }) +\Delta_{\tilde g} \overline  d_{010A}   - \frac{9}{8}\tau^2 \overline d_{011A}
\\
 &&  -  \frac{3}{4}\tau^2 \overline d_{ 010A }   -\tau\lambda_A{}^B\overline d_{011B }
  + \underbrace{ \frac{3}{2}\lambda_A{}^B\lambda_B{}^C \overline d_{011C }  - \frac{3}{4}|\lambda|^2\overline d_{01 1A} }_{=0}
 \;,
\end{eqnarray*}
as follows from the constraint equations.
We have
\begin{eqnarray*}
 \overline{\nabla_{\rho}d_{0A0}{}^{\rho}} &=& \overline{ \nabla_{0}d_{010A}} + (\partial_1 +\tau ) \overline d_{010A}
  +\tilde \nabla^B \overline d_{0A0B}   - \frac{1}{2}\lambda_A{}^B  \overline d_{011B}
 \;,
\end{eqnarray*}
which implies, again via the constraint equations,
\begin{eqnarray*}
2\partial_1  \overline{\nabla_{\rho}d_{0A0}{}^{\rho}} &=& 2\partial_1 \overline{ \nabla_{0}d_{010A}}
  +2\tilde \nabla^B(\partial_1 - r^{-1} ) \breve{\overline d}_{0A0B}-  \tilde \nabla_A(\partial_1 + 3r^{-1} ) \overline d_{0101}
\\
 && +\tau   \tilde \nabla_A\overline d_{0101}
    -\tau \tilde \nabla^B \overline d_{0A0B}
  -\lambda_A{}^B (\partial_1 + 3r^{-1} )\overline d_{011B}
 -\tau^2 \overline d_{010A}
\\
 && +2 \tau \lambda_A{}^B  \overline d_{011B}
 +2 \omega_{A}{}^B\overline d_{011B}
 +2(\partial_1 + r^{-1})^2 \overline d_{010A}
\\
 &=&2\partial_1 \overline{ \nabla_{0}d_{010A}}
   - \frac{9}{8}\tau^2 \overline d_{011A}
  - \frac{3}{4} \tau^2\overline  d_{010A} + \Delta_{\tilde g} \overline d_{010A}   -\frac{1}{2}\Delta_{\tilde g} \overline d_{011A}
\\
 && +\frac{3}{2}\tilde\nabla^B(\lambda_{(A}{}^C\overline  d_{B)C01})    + \frac{3}{2}\tilde \nabla_B(\lambda_A{}^B\overline d_{0101})
   +  \tilde\nabla^B(\lambda_{[A}{}^C\overline d_{B]1C1})
\\
 && + \tilde\nabla_A\tilde\nabla^B \overline d_{011B}    -\tau \tilde \nabla^B \overline d_{0A0B} -  \tau \tilde\nabla^B\overline d_{01AB}
  -2  \omega_{A}{}^B  \overline d_{011B}
\\
 && - \frac{1}{2}\tilde\nabla^B (\partial_1 -5 r^{-1})\overline d_{1A1B}
   -\tau \lambda_A{}^B\overline d_{011B}
  +  \lambda_A{}^C \tilde\nabla^B\overline d_{1B1C}
 \;.
\end{eqnarray*}
Combining these results we end up with
\begin{eqnarray*}
2\partial_1  \overline{\nabla_{\rho}d_{0A0}{}^{\rho}} &=&
 \underbrace{\frac{1}{2}\tilde \nabla^B(\lambda_{(A}{}^C \overline d_{B)1C1}) - \frac{1}{4}  \tilde\nabla_A(\lambda^{BC}\overline d_{1B1C})}_{=0}
        - \frac{3}{2}\underbrace{\tilde \nabla^B(\lambda_{[B}{}^C \overline d_{A]C01})}_{=0}
 \;,
\end{eqnarray*}
and, as regularity requires $\overline{\nabla_{\rho}d_{0A0}{}^{\rho}} =O(r)$, that gives
\begin{equation*}
 \overline{\nabla_{\rho}d_{0A0}{}^{\rho}} =0
 \;.
\end{equation*}

To continue, we analyse the vanishing of  $\overline {\nabla_{\rho}d_{0AB}{}^{\rho}}$.
We have
\begin{eqnarray*}
 \overline{\Box_g d_{0AB1}}
 &=& 2 (\partial_1 - r^{-1}) \ol{\nabla_0 d_{0AB1}}
 -2\ol L_0{}^C\ol d_{1BAC} - 2\ol L_{0A}\ol d_{01 1B}
\\
&& +(\partial_1 - \tau)\partial_1\ol d_{0AB1}
+ \Delta_{\tilde g} \ol d_{0AB1}
  -\lambda_{AC} \tilde \nabla^C\ol d_{01 1B}
 -\tau   \tilde \nabla_A \ol d_{01 1B}
\\
&& +\tau  \tilde \nabla_B \ol d_{0 10A} - \tau  \tilde \nabla^C \ol d_{0ABC}
  - \lambda^{CD} \tilde \nabla_D \ol d_{1BAC} - \frac{1}{2}\tau  \lambda^{CD}\ol d_{ACBD}
\\
&&
 - \frac{1}{2}\tau  \lambda_A{}^{C} \ol d_{1B1 C}
    + \frac{1}{2}\tau^2 \ol d_{01AB}
+ \frac{1}{2}\tau(\tau \ol g_{AB} +  \lambda_{AB}) \ol d_{01 0 1}
\\
&&     - \tau  \underbrace{\lambda_{[A}{}^C \ol d_{B]C01 }}_{=0}
  + \underbrace{\frac{1}{4}|\lambda|^2\ol d_{1A1B}  - \frac{1}{2} \lambda_A{}^C\lambda_C{}^D \ol d_{1B1D} }_{=0}
  \\
&=& 2 (\partial_1 - r^{-1}) \ol {\nabla_0 d_{0AB1}}
 -2\overline  g_{AB} \ol L_0{}^C  \overline d_{011C}
  - 4\ol L_{0[A}\ol d_{B]1 10} + \frac{1}{2}\Delta_{\tilde g} \overline d_{1A1B}
  \\
&& +\frac{1}{2}(\partial_1 - \tau)\partial_1\overline d_{1A1B}
-  \frac{1}{2}\overline g_{AB} \tilde\nabla^C\tilde\nabla^D\ol d_{1C1D}
+  \frac{1}{2}\overline g_{AB} \omega^{CD}\ol d_{1C1D}
\\
&&
-  \frac{1}{4}\overline g_{AB}\lambda^{CD} (\partial_1-2\tau )\ol d_{1C1D}
 -  \frac{1}{2}\overline g_{AB}\Delta_{\tilde g} \overline d_{0101}
 -\tilde\nabla_{[A}\tilde\nabla^C\ol d_{B]1C1}
 \\
 &&
  -\tau\tilde\nabla_{B}\ol d_{A110}
  +\omega_{[A}{}^C \ol d_{B]1C1}
  -\frac{1}{2}\lambda_{[A}{}^C(\partial_1 - 2\tau)\ol d_{B]1C1}
-\frac{1}{2} \Delta_{\tilde g} \ol d_{01AB}
\\
&&
  -2\lambda_{C[A} \tilde \nabla^C\ol d_{B]1 10}
 -\tau   \tilde \nabla_{(A} \ol d_{B)1 10}
  - 2\tau  \tilde \nabla_{[A} \overline d_{B]010}
   - \frac{1}{2}\tau  \lambda_A{}^{C} \ol d_{1B1 C}
\\
&&   + \tau \overline g_{AB}  \tilde \nabla^C  \overline d_{010C}
 + \frac{5}{2}\tau \overline g_{AB} \tilde \nabla^C  \overline d_{011C}
  -\overline  g_{AB}\lambda^{CD} \tilde \nabla_D \overline d_{011C}
     \;,
\end{eqnarray*}
which vanishes owing to \eq{wave_d_cone}.
Moreover,
\begin{eqnarray*}
 \overline{\nabla_{\rho}d_{0AB}{}^{\rho}} &=&  \ol{\nabla_{0}d_{0AB1}} + (\partial_1 - r^{-1})\ol d_{0AB0}+ (\partial_1- r^{-1}) \ol d_{0AB1}
+\tilde \nabla^C\ol d_{0ABC}
\\
 && + \frac{1}{2}\lambda^{CD} \ol d_{ACBD}
 + \frac{1}{2}\lambda_{A}{}^C \ol d_{01 BC}  + \frac{1}{2}\lambda_{B}{}^C \ol d_{0A1 C}  - \frac{1}{2}\tau \ol d_{01 AB}
\\
 &=&  \ol{\nabla_{0}d_{0AB1}}+  \frac{1}{4}(\partial_1 - r^{-1})\ol d_{1A1B}
 - \frac{1}{4}\lambda_{A}{}^C\ol d_{1B1C}
 + \tilde \nabla_{[A}\overline d_{B]010}
\\
&& +  \frac{1}{2}\tilde \nabla_{(A}\overline d_{B)110}
 - \frac{1}{2}\overline g_{AB}\tilde \nabla^C\overline d_{010C}
   -  \frac{3}{4}\overline g_{AB} \tilde \nabla^C\overline d_{011C}
 + \frac{1}{4}\lambda_{[A}{}^C\ol d_{B]C01 }
   \;,
\end{eqnarray*}
and thus
\begin{eqnarray*}
&&\hspace{-3em} 2(\partial_1 - r^{-1})\overline{\nabla_{\rho}d_{0AB}{}^{\rho}}
\\
 &=&  2(\partial_1 - r^{-1})\ol{\nabla_{0}d_{0AB1}}+  \frac{1}{2}(\partial_1 - r^{-1})^2\ol d_{1A1B}
 - \frac{1}{2} \lambda_{A}{}^C(\partial_1 - \tau)\ol d_{1B1C}
 \\
 && +\omega_{A}{}^C\ol d_{1B1C}
 + 2\tilde \nabla_{[A}(\partial_1 +r^{-1})\overline d_{B]010}
 +  \tilde \nabla_{(A}(\partial_1 +3r^{-1}  )\overline d_{B)110}
\\
&&
 - \overline g_{AB}\tilde \nabla^C(\partial_1+r^{-1} )\overline d_{010C}
   -  \frac{3}{2}\overline g_{AB} \tilde \nabla^C(\partial_1 +3r^{-1} )\overline d_{011C}
   - 2\tau\tilde \nabla_{[A}\overline d_{B]010}
 \\
 &&
 + \frac{1}{2}\lambda_{[A}{}^C(\partial_1+r^{-1})\ol d_{B]C01 }
 -\frac{3}{2} r^{-1} \lambda_{[A}{}^C\ol d_{B]C01 }
   +\tau  \overline g_{AB}\tilde \nabla^C\overline d_{010C}
 \\
 &&    - 2\tau  \tilde \nabla_{(A}\overline d_{B)110}
   +  3\tau \overline g_{AB} \tilde \nabla^C\overline d_{011C}
    -\omega_{[A}{}^C\ol d_{B]C01 }
\\
 &=&  2(\partial_1 - r^{-1})\ol{\nabla_{0}d_{0AB1}}+  \frac{1}{2}(\partial_1 - r^{-1})^2\ol d_{1A1B}
 - \frac{1}{2} \lambda_{A}{}^C(\partial_1 - \tau)\ol d_{1B1C}
 \\
 && +\omega_{A}{}^C\ol d_{1B1C}
 +  \tilde \nabla_{B}\tilde\nabla^C\ol d_{1A1C}
   -\overline g_{AB}\tilde \nabla^C \tilde\nabla^D \ol d_{1C1D}
    - \tau   \tilde \nabla_{B} \ol d_{A110}
 \\
 &&
  - \tau  \tilde \nabla_{(A}\overline d_{B)110}
   +  \frac{5}{2}\tau \overline g_{AB} \tilde \nabla^C\overline d_{011C}
 +  2( \tilde \nabla_{[A}\lambda_{B]}{}^C - \overline g_{AB}\ol L_0{}^C  )\ol d_{011C}
\\
&& +  \tilde \nabla_{[A}\tilde\nabla^C  \ol d_{B]C01}
 - \frac{3}{2}\lambda_{[A}{}^C \tilde\nabla_{B]}\ol d_{C110}
  - 2\tau\tilde \nabla_{[A}\overline d_{B]010}
  +\tau  \overline g_{AB}\tilde \nabla^C\overline d_{010C}
 \\
 &&  -\frac{1}{2}\lambda_{C[A}\tilde\nabla^C\ol d_{B]110}
   - \overline g_{AB}\lambda_{C}{}^D \tilde \nabla^C \ol d_{011D}
    - \frac{1}{2}\overline g_{AB}\Delta_{\tilde g} \ol d_{0101}
    \\
 &&   - \frac{1}{4}\underbrace{\lambda_C{}^D\lambda_{[A}{}^C\ol d_{B]1D1}}_{=0}
 -\underbrace{\omega_{[A}{}^C\ol d_{B]C01 }}_{=0}
 -\frac{3}{2} r^{-1} \underbrace{\lambda_{[A}{}^C\ol d_{B]C01 }}_{=0}
 \;.
\end{eqnarray*}
Using the formula \eq{Delta_d1A1B} we derived for $\Delta_{\tilde g}  \overline d_{1A1B} $, we conclude that
\begin{eqnarray*}
&&\hspace{-3em} 2(\partial_1 - r^{-1})\overline{\nabla_{\rho}d_{0AB}{}^{\rho}}
\\
 &=& 2\{ (\tilde\nabla_C\lambda_{[A}{}^C)\ol d_{B]1 10}  +  (\tilde \nabla_{[A}\lambda_{B]}{}^C) \ol d_{011C}\}
 - \frac{1}{2}(\partial_1-3r^{-1})\{( \lambda_{(A}{}^C \ol d_{B)1C1})\breve{}\}
  \\
 && + \frac{3}{2}\{\lambda_{C[A} \tilde \nabla^C\ol d_{B]1 10}  - \lambda_{[A}{}^C \tilde\nabla_{B]}\ol d_{C110}\}
 + \{ \tilde \nabla_{[A}\tilde\nabla^C  \ol d_{B]C01}
 +\frac{1}{2} \Delta_{\tilde g} \ol d_{01AB}\}
 \\
 &=& 0
 \;,
\end{eqnarray*}
since the terms in each of the braces add up to zero (recall Lemma~\ref{lemma_tracelessness}).
Taking further into account that regularity yields $\overline{\nabla_{\rho}d_{0AB}{}^{\rho}} =O(r^2)$,
we deduce that
\begin{eqnarray*}
\overline{\nabla_{\rho}d_{0AB}{}^{\rho}} &=& 0
\;.
\end{eqnarray*}

\subsection{Main result}

By way of summary we end up with the following result:
\begin{theorem}
\label{main_result}
   Let us suppose we have been given a smooth one-parameter family of $s$-traceless tensors $\omega_{AB}(r,x^A)=O(r^4)$ on the 2-sphere, where $s$ denotes the standard metric. Let $\lambda_{AB}$ be the unique solution of the equation
   \bel{7X12.1}  (\partial_1- r^{-1})\lambda_{AB} = - 2 \omega_{AB}\;,
   \ee
   with $\lambda_{AB}=O(r^5)$.
A smooth solution
$(g_{\mu\nu},  L_{\mu\nu}, d_{\mu\nu\sigma}{}^{\rho}, \Theta,  s)$
of the CWE \eq{cwe1*}-\eq{cwe5*} to the future of $C_{i^-}$, smoothly extendable through $C_{i^-}$,
with initial data  $(\mathring g_{\mu\nu}, \mathring L_{\mu\nu}, \mathring d_{\mu\nu\sigma}{}^{\rho}, \mathring \Theta=0, \mathring s=-2 )$ and with
$ \breve{\mathring L}_{AB}=\omega_{AB}$,
is a solution of the MCFE \eq{mconf1}-\eq{mconf6} with $\lambda=0$
in the
$$(R=0,\overline  s=-2,\kappa=0,\hat g_{\mu\nu}=\eta_{\mu\nu}) \text{-wave-map gauge}\;,   $$
 if and only if the initial data have their usual algebraic properties and  solve the constraint equations \eq{constraint_g}-\eq{initial_L00_spec_gauge}
with boundary conditions \eq{constraints_boundary_data}.

The function $\Theta$ is positive in the interior of $C_{i^-}$ and sufficiently close to $i^-$, and $\mathrm{d}\Theta\ne 0 $ on $C_{i^-}\setminus \{i^-\}$.
\end{theorem}

\begin{remark}
 Note that regularity for the rescaled Weyl tensor implies that the initial data necessarily need to satisfy $\omega_{AB}(r,x^A)=O(r^4)$, cf.\ equation \eq{initial_d1A1B_spec_gauge}.
\end{remark}

\begin{proof}
The previous computations show that Theorem~\ref{inter-thm} is applicable.
The positivity of $\Theta$ inside the cone simply follows from \eq{mconf3} and the negativity of $s$ near the vertex as one might check using e.g.\ normal coordinates.

Concerning  the ``only if''-part: That the constraint equations \eq{constraint_g}-\eq{initial_L00_spec_gauge} are satisfied
by any solution of the MCFE in the $(R=0,\overline s=-2,\kappa=0,\hat g_{\mu\nu}=\eta_{\mu\nu})$-wave-map gauge and with $\overline \Theta =0$
follows directly from their derivation. 
\qed
\end{proof}

\section{Alternative system of conformal wave equations (CWE2)}
\label{alternative_system}

Instead of a wave equation for the rescaled Weyl tensor $d_{\mu\nu\sigma}{}^{\rho}$, it might be advantageous in certain situations to work with the
Weyl tensor itself, which we denote here by $C_{\mu\nu\sigma}{}^{\rho}$, as unknown. The Weyl tensor is a more physical quantity
(it is conformally invariant and thus coincides with the physical Weyl tensor)
and can be expressed
in terms of the metric even on null and timelike infinity.
We shall see that proceeding this way it becomes necessary to regard the Cotton tensor as another unknown, so that
the system of wave equations we are about to derive might be somewhat more complicated.
An advantage is that we just need to require the metric to be regular at $i^-$ rather than the metric and the rescaled Weyl tensor, so the alternative
system might be useful to find a more general class of solutions
(cf.\ the discussion in Section~\ref{comparison}).

Since many of the computations which need to be done to derive the alternative system of wave equations \eq{Wweylwave1}-\eq{Wweylwave6} and  prove Theorem~\ref{main_result2} are very similar to the ones we did for the CWE involving the rescaled Weyl tensor, the computations are partially even more compressed than in the previous part.

\subsection{Derivation}

The \textit{Cotton tensor} in $4$-spacetime dimensions is defined as
\begin{eqnarray*}
 \xi_{\mu\nu\sigma} \,:=\, 2 \nabla_{[\sigma}R_{\nu]\mu} +\frac{1}{3}g_{\mu[\sigma} \nabla_{\nu]} R
\,=\, 4\nabla_{[\sigma}L_{\nu]\mu}
 \;.
\end{eqnarray*}
It is manifestly antisymmetric in its last two indices. Moreover, the Bianchi identities imply the
following properties,
\begin{eqnarray}
 \xi_{[\mu\nu\sigma]} &=& 0 \;,
 \label{Wxi_Bianchi}
\\
   \xi_{\rho\nu}{}^{\rho} &=& 0
 \;,
\\
 \nabla^{\rho}\xi_{\rho\nu\sigma}&=&0
\label{Wrelation_div_xi}
 \;,
\\
 \xi_{\mu\nu\sigma} &=& -  2\nabla_{\alpha}C^{\alpha }{}_{\mu\nu\sigma}
  \label{Wrelation_xi_C}
 \;.
\end{eqnarray}
%
Using the wave equation \eq{waveeqn_L}
  for the Schouten tensor (written in terms of $C_{\mu\nu\sigma}{}^{\rho}$ rather than $\Theta d_{\mu\nu\sigma}{}^{\rho}$) one further verifies the relation
\begin{equation}
  2 L_{\alpha\sigma}C_{\mu}{}^{\alpha}{}_{\nu}{}^{\sigma}+  \nabla^{\sigma}\xi_{\mu\nu\sigma}=0 \;,
 \label{Wrelation_cotton_tensor}
\end{equation}
which expresses the vanishing of the \textit{Bach tensor}.

The second Bianchi identity implies,
\begin{eqnarray}
  2\nabla_{[\alpha}C_{\mu\nu]\sigma}{}^{\rho}
  &=&  g_{\sigma[\mu}\xi^{\rho}{}_{\alpha\nu]} + \delta_{[\mu}{}^{\rho}\xi_{|\sigma|\nu\alpha]}
 \;.
 \label{Wbianchi_weyl}
\end{eqnarray}
(In particular one recovers \eq{Wrelation_xi_C} for $\rho=\alpha$.)
Contracting \eq{Wbianchi_weyl} with $\nabla^{\alpha}$ we find a wave equation
 for the Weyl tensor%
\footnote{Recall that $\Box_g$, acting on higher valence tensors, is not a wave-operator if the metric field belongs to the unknowns.}
\begin{eqnarray}
 \Box_g C_{\mu\nu\sigma\rho}
 &\overset{\eq{Wrelation_cotton_tensor}}{=}&  2\nabla^{\alpha} \nabla_{[\nu}C_{\mu]\alpha\sigma\rho}
 +2 g_{\sigma[\mu}C_{\nu]\alpha\rho\beta}L^{\alpha\beta} -2 g_{\rho[\mu} C_{\nu]\alpha\sigma\beta}L^{\alpha\beta}
-\nabla_{[\sigma}\xi_{\rho]\mu\nu}
 \nonumber
\\
 &\overset{\eq{Wrelation_xi_C}}{=}&  C_{\mu\nu\alpha}{}^{ \kappa} C_{\sigma\rho\kappa}{}^{\alpha}  -4C_{\sigma\kappa[\mu}{}^{\alpha}  C_{\nu]\alpha \rho}{}^{ \kappa}     - 2C_{\sigma\rho\kappa[\mu} L_{\nu]}{}^{ \kappa}    - 2 C_{\mu\nu\kappa[\sigma} L_{\rho]}{}^{\kappa}
 \nonumber
\\
 &&
   -\nabla_{[\sigma}\xi_{\rho]\mu\nu}  -\nabla_{[\mu}\xi_{\nu]\sigma\rho} +  \frac{1}{3} R C_{\mu\nu\sigma\rho}
 \;.
 \label{Wvconf4}
\end{eqnarray}
We observe that the Cotton tensor is needed to eliminate the disturbing second-order derivatives of $C_{\mu\nu\sigma\rho}$.

Finally, we derive a wave equation for the Cotton tensor $\xi_{\mu\nu\sigma}$ by employing the wave equation
\eq{waveeqn_L} for the Schouten tensor, the Bianchi identity and \eq{Wrelation_xi_C}:
\begin{eqnarray}
 \Box_{ g}\xi_{\mu\nu\sigma}
 &\equiv&   4 \nabla_{[\sigma} \Box_{ g} L_{\nu]\mu}  + 8g_{\mu[\nu} L_{|\alpha|}{}^{\kappa}\nabla^{\alpha} L_{\sigma]\kappa}
    - 16L_{[\nu}{}^{\kappa}\nabla_{\sigma]} L_{\mu\kappa}
 \nonumber
\\
 && +2\xi_{\kappa\sigma\nu}  L_{\mu}{}^{\kappa}   +  4 \xi_{\mu\kappa[\sigma}L_{\nu]}{}^{\kappa}
+ C_{\nu\sigma\alpha}{}^{\kappa}\xi_{\mu\kappa}{}^{\alpha} + 8 C_{\alpha[\sigma|\mu|}{}^{\kappa}\nabla^{\alpha} L_{\nu]\kappa}
 \nonumber
\\
 && -  \frac{2}{3} R \nabla_{[\nu} L_{\sigma]\mu}+\frac{2}{3} L_{\mu[\nu}\nabla_{\sigma]}R
  +\frac{2}{3}g_{\mu[\nu}L_{\sigma]\kappa}  \nabla^{\kappa}R
 \nonumber
\\
 &=&  4 \xi_{\kappa\alpha[\nu} C_{\sigma]}{}^{\alpha}{}_{\mu}{}^{\kappa}
 + C_{\nu\sigma\alpha}{}^{\kappa}\xi_{\mu\kappa}{}^{\alpha}
 -  4 \xi_{\mu\kappa[\nu}L_{\sigma]}{}^{\kappa}  + 6g_{\mu[\nu} \xi^{\kappa}{}_{\sigma\alpha]}L_{\kappa}{}^{\alpha}
 \nonumber
\\
 && + 8L_{\alpha\kappa} \nabla_{[\nu}C_{\sigma]}{}^{\alpha}{}_{\mu}{}^{\kappa}  +  \frac{1}{6} R \xi_{\mu\nu\sigma}
  - \frac{1}{3}  C_{\nu\sigma\mu}{}^{\kappa}\nabla_{\kappa} R
    \label{Wvconf5}
 \;.
\end{eqnarray}

Combining these results with the  equations we found for $\Theta$, $s$, $g_{\mu\nu}$ and $L_{\mu\nu}$, we end up with an alternative system of conformal wave equations (of course we need to replace $\Box_g$ by $\Box^{(H)}_g$, cf.\ Section~\ref{ss31VIII11}),
\begin{eqnarray}
  \Box^{(H)}_{ g} L_{\mu\nu}&=&  4 L_{\mu\kappa} L_{\nu}{}^{\kappa} -  g_{\mu\nu}| L|^2
  - 2C_{\mu\sigma\nu}{}^{\rho}  L_{\rho}{}^{\sigma}
 + \frac{1}{6}\nabla_{\mu}\nabla_{\nu} R
   \;,
 \label{Wweylwave1}
\\
 \Box_{ g}  s  &=& \Theta| L|^2 -\frac{1}{6}\nabla_{\kappa} R\,\nabla^{\kappa}\Theta  - \frac{1}{6} s  R
  \;,
 \label{Wweylwave2}
\\
  \Box_{ g}\Theta &=& 4 s-\frac{1}{6} \Theta  R
  \;,
 \label{Wweylwave3}
\\
 \Box^{(H)}_g C_{\mu\nu\sigma\rho} &=&    C_{\mu\nu\alpha}{}^{ \kappa} C_{\sigma\rho\kappa}{}^{\alpha}  -4C_{\sigma\kappa[\mu}{}^{\alpha}  C_{\nu]\alpha \rho}{}^{ \kappa}     - 2C_{\sigma\rho\kappa[\mu} L_{\nu]}{}^{ \kappa}    - 2 C_{\mu\nu\kappa[\sigma} L_{\rho]}{}^{\kappa}
 \nonumber
\\
 &&
   -\nabla_{[\sigma}\xi_{\rho]\mu\nu}  -\nabla_{[\mu}\xi_{\nu]\sigma\rho} +  \frac{1}{3} R C_{\mu\nu\sigma\rho}
  \;,
 \label{Wweylwave4}
\\
 \Box^{(H)}_{ g}\xi_{\mu\nu\sigma} &=& 4 \xi_{\kappa\alpha[\nu} C_{\sigma]}{}^{\alpha}{}_{\mu}{}^{\kappa}
 + C_{\nu\sigma\alpha}{}^{\kappa}\xi_{\mu\kappa}{}^{\alpha}
 -  4 \xi_{\mu\kappa[\nu}L_{\sigma]}{}^{\kappa}  + 6g_{\mu[\nu} \xi^{\kappa}{}_{\sigma\alpha]}L_{\kappa}{}^{\alpha}
 \nonumber
\\
 && + 8L_{\alpha\kappa} \nabla_{[\nu}C_{\sigma]}{}^{\alpha}{}_{\mu}{}^{\kappa}  +  \frac{1}{6} R \xi_{\mu\nu\sigma}
  - \frac{1}{3}  C_{\nu\sigma\mu}{}^{\kappa}\nabla_{\kappa} R
  \;,
 \label{Wweylwave5}
\\
  R^{(H)}_{\mu\nu}[g] &=& 2L_{\mu\nu} + \frac{1}{6} R g_{\mu\nu}
 \label{Wweylwave6}
 \;.
\end{eqnarray}

\begin{remark}
  Note that  \eq{Wweylwave1} and \eq{Wweylwave4}-\eq{Wweylwave6} do not involve the functions $s$ and $\Theta$, so they form a closed system of wave equations for $g_{\mu\nu}$, $L_{\mu\nu}$, $\xi_{\mu\nu\sigma}$ and $C_{\mu\nu\sigma\rho}$.
Once a solution has been constructed,  it remains to solve the linear wave equations  \eq{Wweylwave2} and  \eq{Wweylwave3} for $s$ and $\Theta$.
\end{remark}

We want to investigate under which conditions
 a solution of  the system \eq{Wweylwave1}-\eq{Wweylwave6},
which we denote henceforth by CWE2, provides a solution of the~MCFE.

\subsection{Some properties of the CWE2 and gauge consistency}
\label{Wreduced_cwe}

First of all we want to establish consistency with the gauge conditions $H^{\sigma}=0$ and $R=R_g$.
To do that we assume that there are smooth fields $g_{\mu\nu}$, $s$, $\Theta$, $C_{\mu\nu\sigma}{}^{\rho}$, $L_{\mu\nu}$ and $\xi_{\mu\nu\sigma}$
which solve the CWE2.
We aim to derive necessary and sufficient conditions on the initial surface which guarantee the vanishing of   $H^{\sigma}$ and $R-R_g$.
For definiteness we, again, think of the case where the initial surface consists of either two transversally intersecting null hypersurfaces or a light-cone.
The strategy will be the same as for the CWE, which is to derive a homogeneous system of wave equations for $H^{\sigma}$ as well as some subsidiary fields, and infer the desired result from standard uniqueness results for wave equations by making the assumption, which will be analysed afterwards, that all the fields involved vanish initially.

However, let us first derive some properties of solutions of the CWE2.

\begin{lemma}
\label{Wmetric_symm}
 Assume that the initial data on a characteristic initial surface $S$  of some smooth solution of the CWE2 are such that
$g_{\mu\nu}|_S$ is the restriction to $S$ of a Lorentzian metric, that $ L_{[\mu\nu]}|_S=0$ and $ C_{\mu\nu\sigma\rho}|_S=C_{\sigma\rho\mu\nu}$.
 Then $g_{\mu\nu}$ and $L_{\mu\nu}$ are symmetric and $C_{\mu\nu\sigma\rho} =C_{\sigma\rho\mu\nu}$.
\end{lemma}
\begin{proof}
 Equation \eq{Wweylwave1} yields (cf.\ footnote~\ref{non-symmetric_g})
%
\begin{eqnarray}
  &&\hspace{-3em}\Box^{(H)}_g (C_{\mu\nu\sigma\rho}- C_{\sigma\rho\mu\nu}) \,=\,  \frac{1}{3} R (C_{\mu\nu\sigma\rho} -  C_{\sigma\rho\mu\nu})
 \nonumber
\\
 &&+ 4g^{\alpha\beta}g^{\kappa\gamma}[
  (C_{[\mu|\beta\sigma\kappa|}-C_{\sigma\kappa[\mu|\beta|})C_{\nu]\alpha\rho\gamma}
   + (C_{\rho\alpha[\nu|\gamma|}-C_{[\nu|\gamma\rho\alpha|}) C_{\mu]\kappa\sigma\beta}]
  \nonumber
\\
 && -  4(g^{\alpha\beta}g^{[\kappa\gamma]}
+ g^{\gamma\kappa}g^{[\alpha\beta]} )
 C_{\nu\alpha\rho\gamma}C_{\mu\beta\sigma\kappa}
 \label{Wsymmetry_C}
 \;.
\end{eqnarray}
From \eq{Wweylwave1} and \eq{Wweylwave6} we further find
\begin{eqnarray}
 \Box^{(H)}_{ g} L_{[\mu\nu]}&=&  4g_{[\alpha\beta]}L_{\mu}{}^{\alpha}L_{\nu}{}^{\beta} - g_{[\mu\nu]}| L|^2
  + g^{\rho\gamma}L_{\rho }{}^{\sigma}(C_{\nu\sigma\mu\gamma}-C_{\mu\gamma\nu\sigma})
 \nonumber
\\
 && + 2g^{\sigma\kappa}C_{\mu}{}^{\rho}{}_{\nu\sigma}L_{[\rho\kappa]}
 - 2g^{[\sigma\kappa]}C_{\mu\sigma\nu}{}^{\rho}L_{\rho\kappa}
 \label{Wantisym_L}
 \;,
\\
  R^{(H)}_{[\mu\nu]}[g_{(\sigma\rho)},g_{[\sigma\rho]}] &=&  2L_{[\mu\nu]} + \frac{1}{6} R g_{[\mu\nu]}
 \label{Wantisym_R}
 \;.
\end{eqnarray}
The equations \eq{Wsymmetry_C}-\eq{Wantisym_R} are to be read as a linear, homogeneous system of wave equations satisfied by $g_{[\mu\nu]}$, $L_{[\mu\nu]}$
and $C_{\mu\nu\sigma\rho}- C_{\sigma\rho\mu\nu}$ with all the other fields being given.
Hence if we assume these fields to vanish initially they will vanish everywhere.
\qed
\end{proof}

The lemma shows that the tensor $g_{\mu\nu}$ determines indeed a metric as long as it does not degenerate.
We will only care about  initial data for which the assumptions of this lemma hold.

In analogy to Lemma~\ref{lemma_properties_d} one could show that $C_{\mu\nu\sigma\rho}$ is anti-symmetric in its first two and last two indices and satisfies
$C_{[\mu\nu\sigma]\rho}=0$, and that $\xi_{\mu\nu\sigma}$ is anti-symmetric in its last two indices and fulfills $\xi_{[\mu\nu\sigma]}=0$, supposing that
this is initially the case.
However, these properties will follow a posteriori anyway, so it is not necessary to prove them here.
Due to the appearance of first-order derivatives on the right-hand side of the wave equations for   $C_{\mu\nu\sigma\rho}$ and $\xi_{\mu\nu\sigma}$,
it is not possible to establish tracelessness of  $C_{\mu\nu\sigma\rho}$ and $\xi_{\mu\nu\sigma}$ 
at this stage in a manner  it was possible
for the CWE (where it simplified the subsequent computations), since this would require to have something like the second Bianchi identity;
also these properties can be, again, inferred a posteriori, once we know that $C_{\mu\nu\sigma\rho}$ and $\xi_{\mu\nu\sigma}$
 are Weyl and Cotton tensor of $g_{\mu\nu}$, respectively.

\subsubsection*{Gauge consistency}

Similarly to what we did in Section~\ref{subsect_consist_gauge}, one proceeds to verify the formulae
\begin{eqnarray}
&&\hspace{-3em}  R_{\mu\nu} - \frac{1}{2} R_g g_{\mu\nu} = 2L_{\mu\nu} - (L+ \frac{1}{6}R)g_{\mu\nu}
  + g_{\sigma(\mu}\hat\nabla_{\nu)} H^{ \sigma}
 -\frac{1}{2} g_{\mu\nu} \hat\nabla_{\sigma} H^{ \sigma}
  \;,
     \label{W_Ricci_Schouten}
\\
&&\hspace{-3em} \nabla^{\nu} \hat\nabla_{\nu} H^{ \alpha}+2g^{\mu\alpha} \nabla_{[\sigma} \hat\nabla_{\mu]} H^{ \sigma}
 + 4\nabla^{\nu} L_{\nu}{}^{\alpha}-2\nabla^{\alpha}L - \frac{1}{3}\nabla^{\alpha}R =0
 \label{Wwave_HWeyl}
 \;,
\\
 &&\hspace{-3em} \Box_gH^{\alpha}
 = \zeta^{\alpha}   + f^{\alpha}(x; H,\nabla H)
 \;, \quad  \zeta_{\mu}:= -4\nabla_{\kappa}L_{\mu}{}^{\kappa} + 2\nabla_{\mu}L + \frac{1}{3}\nabla_{\mu}R
 \;,
 \label{Wwave_H*Weyl}
\\
&&\hspace{-3em}  \Box_g K_{\mu\nu}
 = \nabla_{\mu}\zeta_{\nu}
 +f_{\mu\nu}(x;H,\nabla H,\nabla K)
\;, \quad K_{\mu\nu} := \nabla_{\mu}H_{\nu}
 \label{Wwave_DiffH}
 \;,
\\
 &&\hspace{-3em} R_g =   2L + \frac{2}{3}R  + \hat\nabla_{\sigma} H^{ \sigma}
 \label{Weqn_curvscalarWeyl}
 \;.
\end{eqnarray}
From \eq{Wweylwave1} we derive a wave equation for $L-R/6$,
\begin{eqnarray}
 \Box_{ g} (L-\frac{1}{6}R) &=&
 -2 C_{\mu\sigma}{}^{\mu\rho}  L_{\rho}{}^{\sigma}
 \,=\, 2(W_{\mu\sigma}{}^{\mu\rho} -C_{\mu\sigma}{}^{\mu\rho})  L_{\rho}{}^{\sigma}
 \label{Wtr_L}
 \;.
\end{eqnarray}
The tensors  $L_{\mu\nu}$,  $C_{\mu\nu\sigma}{}^{\rho}$ and $\xi_{\mu\nu\sigma}$ are supposed to be part of the given solution of the CWE2;
we stress that it is by no means clear,  whether they, indeed, represent the Schouten,  Weyl  and  Cotton tensor of $g_{\mu\nu}$, respectively.
We denote by $W_{\mu\nu\sigma}{}^{\rho}$  the Weyl tensor associated  to $g_{\mu\nu}$, while we define the tensor $\zeta_{\mu\nu\sigma}$ to be
\begin{eqnarray*}
 \zeta_{\mu\nu\sigma}:= 4\nabla_{[\sigma} L_{\nu]\mu}
 \;.
\end{eqnarray*}
Since we do not know at this stage whether the source term in \eq{Wtr_L} vanishes, we have no analogue of Lemma~\ref{lemma_properties_L}. It is not possible to conclude that
 $L-\frac{1}{6}R$ vanishes as we did for the CWE,
supposing that it vanishes initially. In fact that is the reason for the modified definition of $\zeta_{\mu}$ in \eq{Wwave_H*Weyl}.

Note that once we have established $L=\frac{1}{6} R$ and $H^{\sigma}=0$, \eq{Weqn_curvscalarWeyl}
implies $R=R_g$.
For \eq{Wtr_L} to be part of a homogeneous system of wave equations, we regard  $W_{\mu\nu\sigma\rho}-C_{\mu\nu\sigma\rho}$ as another unknown
and show that it satisfies an appropriate homogeneous wave equation (for later purposes this is more advantageous than to derive a wave equation
for the traces $C_{\mu\sigma}{}^{\mu\rho}$).

From \eq{W_Ricci_Schouten} and \eq{Weqn_curvscalarWeyl} we find for the Weyl tensor, cf.\ \eq{nabla_weyl} and \eq{wave_W1}
(since we do not know yet whether $L-R/6$ vanishes, the formulae differ slightly),
\begin{eqnarray*}
   \nabla_{\alpha}W_{\mu\nu\sigma\rho}
 & =& g_{\mu[\sigma}\zeta_{\rho]\alpha\nu} + g_{\nu[\sigma}\zeta_{\rho]\mu\alpha} - g_{\alpha[\sigma}\zeta_{\rho]\mu\nu}
 -2 \nabla_{[\mu}W_{\nu]\alpha\sigma\rho}
 \\
 && + g_{\mu[\sigma}\nabla_{\rho]}\nabla_{[\nu}H_{\alpha]}
 +g_{\nu[\sigma}\nabla_{\rho]}\nabla_{[\alpha}H_{\mu]}
  + g_{\alpha[\sigma}\nabla_{\rho]}\nabla_{[\mu}H_{\nu]}
  \\
 &&  + \frac{4}{3}g_{\mu[\sigma}g_{\rho][\nu}\nabla_{\alpha]}( L - \frac{1}{6}R  )
 - \frac{2}{3}g_{\alpha[\sigma}g_{\rho]\nu}\nabla_{\mu}( L - \frac{1}{6}R  )
  \\
 &&  + \frac{2}{3}g_{\mu[\sigma}g_{\rho][\nu}\nabla_{\alpha]}  \nabla_{\kappa} H^{ \kappa}
 - \frac{1}{3}g_{\alpha[\sigma}g_{\rho]\nu}\nabla_{\mu}  \nabla_{\kappa} H^{ \kappa}
 + f_{\alpha\mu\nu\sigma\rho}(x; H, \nabla H)
 \;,
\\
 \Box_g W_{\mu\nu\sigma\rho}
 & =&  \nabla_{[\sigma}\zeta_{\rho]\nu\mu}    + 2 \nabla_{[\nu}\nabla^{\alpha}W_{\mu]\alpha\sigma\rho}
 + W_{\mu\nu\alpha}{}^{\kappa}W_{\sigma\rho\kappa} {}^{\alpha}  - 4W_{\sigma\kappa[\mu} {}^{\alpha}W_{\nu]\alpha\rho}{}^{\kappa}
\\
 &&+2 ( g_{\rho[\mu} W_{\nu]\alpha\sigma}{}^{\kappa}- g_{\sigma[\mu}W_{\nu]\alpha\rho}{}^{\kappa})L_{\kappa}{}^{\alpha}
 -2L_{[\mu}{}^{\kappa} W_{\nu]\kappa\sigma\rho}  - 2L_{[\sigma}{}^{\kappa}  W_{\rho]\kappa\mu\nu}
\\
 &&
 + g_{\sigma[\mu}\nabla^{\alpha}\zeta_{|\rho\alpha|\nu]}  - g_{\rho[\mu}\nabla^{\alpha}\zeta_{|\sigma\alpha|\nu]}
  + \frac{1}{2}g_{\nu[\sigma}\nabla_{\rho]}\Box_gH_{\mu}    - \frac{1}{2} g_{\mu[\sigma}\nabla_{\rho]}\Box_gH_{\nu}
\\
 &&
 + \frac{1}{2}g_{\sigma[\mu}\nabla_{\nu]}\nabla_{\rho}\nabla_{\alpha}H^{\alpha}
  - \frac{1}{2}g_{\rho[\mu}\nabla_{\nu]}\nabla_{\sigma}\nabla_{\alpha}H^{\alpha}
 -\frac{1}{3}g_{\mu[\sigma}\nabla_{\rho]}\nabla_{\nu}\nabla_{\kappa}H^{\kappa}
\\
 && + \frac{1}{3}g_{\nu[\sigma} \nabla_{\rho]}\nabla_{\mu}  \nabla_{\kappa}H^{\kappa}
 + \frac{1}{3}g_{\mu[\sigma}g_{\rho]\nu}  \Box_g\nabla_{\kappa}H^{\kappa}
   + \frac{2}{3}g_{\mu[\sigma}g_{\rho]\nu} \Box_g( L - \frac{1}{6}R  )
  \\
 &&  + \frac{4}{3}g_{\alpha[\sigma}g_{\rho][\mu} \nabla_{\nu]}\nabla^{\alpha}( L - \frac{1}{6}R  )
 + \frac{1}{3}R W_{\mu\nu\sigma\rho}
+ f_{\mu\nu\sigma\rho}(x; H, \nabla H,\nabla K)
 \;.
\end{eqnarray*}
We further have (cf.\ \eq{contracted_zeta_B} and \eq{contr_Bian_W})
\begin{eqnarray}
 \nabla^{\alpha}\zeta_{\mu\nu\alpha}
  &=&   2 ( W_{\mu\alpha\nu}{}^{\kappa} - 2C_{\mu\alpha\nu}{}^{\kappa} ) L_{\kappa}{}^{\alpha} + \frac{1}{2} \nabla_{\nu} \Box_g H_{\mu} -\nabla_{\mu}\nabla_{\nu} ( L -\frac{1}{6}R)
   \nonumber
 \\
 &&  - (\frac{5}{3} R_g -6L -  \frac{2}{3}R)L_{\mu\nu}
   + (\frac{2}{3} R_g- 2L -  \frac{1}{3}R)L g_{\mu\nu}
    \nonumber
  \\
  && +f_{\mu\nu}(x;H,\nabla H, \nabla K)
\;,
 \label{W_div_zeta}
\\
 \nabla_{\alpha}W^{\alpha}{}_{\mu\nu\sigma}
 &=& -\frac{1}{2}\zeta_{\mu\nu\sigma} + \frac{1}{2} \nabla_{\mu}\nabla_{[\nu} H_{ \sigma]}+ \frac{1}{6} g_{\mu[\nu} \nabla_{\sigma]}(R_g-R)
 \nonumber
 \\
 &&   + f_{\mu\nu\sigma}(x; H,\nabla H)
 \;,
  \label{W_div_W}
\end{eqnarray}
which yields with \eq{Wwave_H*Weyl} and \eq{Weqn_curvscalarWeyl}
\begin{eqnarray*}
 \Box_g W_{\mu\nu\sigma\rho} &=&
 \nabla_{[\sigma}\zeta_{\rho]\nu\mu}   - \nabla_{[\mu}\zeta_{\nu]\sigma\rho}
+W_{\mu\nu\alpha}{}^{\kappa}W_{\sigma\rho\kappa} {}^{\alpha}  - 4W_{\sigma\kappa[\mu} {}^{\alpha}W_{\nu]\alpha\rho}{}^{\kappa}
\\
 &&\hspace{-4em} -2L_{[\mu}{}^{\kappa} W_{\nu]\kappa\sigma\rho}  - 2L_{[\sigma}{}^{\kappa}  W_{\rho]\kappa\mu\nu}
 +4L_{\kappa}{}^{\alpha}(  W_{\rho\alpha[\mu}{}^{\kappa} - C_{\rho\alpha[\mu}{}^{\kappa} )  g_{\nu]\sigma}
\\
 &&\hspace{-4em}  -4L_{\kappa}{}^{\alpha} ( W_{\sigma\alpha[\mu}{}^{\kappa} - C_{\sigma\alpha[\mu}{}^{\kappa} ) g_{\nu]\rho}
  +\frac{1}{2}  g_{\rho[\mu}  \nabla_{\nu]}\zeta_{ \sigma} - \frac{1}{2} g_{\sigma[\mu}   \nabla_{\nu]} \zeta_{ \rho}
 + \frac{1}{2}g_{\nu[\sigma}\nabla_{\rho]}\zeta_{\mu}
\\
 &&\hspace{-4em}   - \frac{1}{2} g_{\mu[\sigma}\nabla_{\rho]}\zeta_{\nu}
   -\frac{8}{3}g_{\sigma[\mu}L_{\nu]\rho}(L -\frac{1}{6}R )
+\frac{8}{3}g_{\rho[\mu}L_{\nu]\sigma}(  L - \frac{1}{6} R  )
+ \frac{1}{3}R W_{\mu\nu\sigma\rho}
\\
 && \hspace{-4em}
 +\frac{4}{3} L g_{\sigma[\mu} g_{\nu]\rho}(L - \frac{1}{6}R )
 + \frac{1}{3}g_{\mu[\sigma}g_{\rho]\nu}  \Box_g (R_g - R )
   + f_{\mu\nu\sigma\rho}(x; H, \nabla H,\nabla K)
 \;.
\end{eqnarray*}
The first term in the last line is disturbing.
However, invoking \eq{Weqn_curvscalarWeyl} and \eq{Wtr_L} we find the relation
\begin{eqnarray*}
 \Box_g (R_g-R) &=&   2 \Box_g (L - \frac{1}{6} R) +  \Box_g \hat\nabla_{\sigma} H^{ \sigma}
\\
 &=&    4(W_{\mu\sigma}{}^{\mu\rho}- C_{\mu\sigma}{}^{\mu\rho})  L_{\rho}{}^{\sigma}
 + f(x; H,\nabla H,\nabla K)
 \;.
\end{eqnarray*}
Combining with \eq{Wweylwave4} we end up with the wave equation
\begin{eqnarray}
 &&\hspace{-3em}\Box_g( W_{\mu\nu\sigma\rho} - C_{\mu\nu\sigma\rho})\,=\,
 -\nabla_{[\sigma}(\zeta_{\rho]\mu\nu}-\xi_{\rho]\mu\nu})  - \nabla_{[\mu}(\zeta_{\nu]\sigma\rho} - \xi_{\nu]\sigma\rho})
  \nonumber
\\
&&  + (W_{\mu\nu\alpha}{}^{\kappa}-C_{\mu\nu\alpha}{}^{\kappa})W_{\sigma\rho\kappa} {}^{\alpha}
 + C_{\mu\nu\alpha}{}^{\kappa}(W_{\sigma\rho\kappa} {}^{\alpha} -C_{\sigma\rho\kappa} {}^{\alpha} )
  \nonumber
\\
 && - 4(W_{\sigma\kappa[\mu} {}^{\alpha}-C_{\sigma\kappa[\mu} {}^{\alpha})W_{\nu]\alpha\rho}{}^{\kappa}
 -  4C_{\sigma\kappa[\mu} {}^{\alpha}(W_{\nu]\alpha\rho}{}^{\kappa}-C_{\nu]\alpha\rho}{}^{\kappa})
  \nonumber
\\
 && -2(W_{\sigma\rho\kappa[\mu}-C_{\sigma\rho\kappa[\mu})L_{\nu]}{}^{\kappa}-2(W_{\mu\nu\kappa[\sigma} - C_{\mu\nu\kappa[\sigma})L_{\rho]}{}^{\kappa}
  \nonumber
\\
 && +4L_{\kappa}{}^{\alpha}(  W_{\rho\alpha[\mu}{}^{\kappa} - C_{\rho\alpha[\mu}{}^{\kappa} )  g_{\nu]\sigma}
 -4L_{\kappa}{}^{\alpha} ( W_{\sigma\alpha[\mu}{}^{\kappa} - C_{\sigma\alpha[\mu}{}^{\kappa} ) g_{\nu]\rho}
  \nonumber
\\
 &&
 +\frac{1}{2}  g_{\rho[\mu}  \nabla_{\nu]}\zeta_{\sigma}
 - \frac{1}{2} g_{\sigma[\mu}   \nabla_{\nu]}\zeta_{\rho}
+ \frac{1}{2}g_{\nu[\sigma}\nabla_{\rho]}\zeta_{\mu}
 - \frac{1}{2} g_{\mu[\sigma}\nabla_{\rho]}\zeta_{\nu}
  \nonumber
\\
 &&
  +\frac{4}{3} (L - \frac{1}{6}R )(L g_{\sigma[\mu} g_{\nu]\rho} -2 g_{\sigma[\mu}L_{\nu]\rho}+2g_{\rho[\mu}L_{\nu]\sigma})
 + \frac{R}{3}( W_{\mu\nu\sigma\rho} -  C_{\mu\nu\sigma\rho} )
  \nonumber
  \\
 &&  + \frac{4}{3}g_{\mu[\sigma}g_{\rho]\nu}(W_{\kappa\alpha}{}^{\kappa\beta}- C_{\kappa\alpha}{}^{\kappa\beta})  L_{\beta}{}^{\alpha}
+ f_{\mu\nu\sigma\rho}(x; H, \nabla H,\nabla K)
 \label{Wwave_W-C}
 \;,
\end{eqnarray}
which is fulfilled by any solution of the CWE2.

\enlargethispage{\baselineskip}

To end up with a homogeneous system we need to derive wave equations for $\zeta_{\mu}$ and $\zeta_{\mu\nu\sigma} - \xi_{\mu\nu\sigma}$.
Let us start with $\zeta_{\mu}$.
In close analogy to  \eq{equation_Box_zeta} and \eq{wave_zetavec_prov} we find with \eq{Wweylwave1}, \eq{Wtr_L}, \eq{W_div_W}, \eq{W_Ricci_Schouten} and \eq{Weqn_curvscalarWeyl},
\begin{eqnarray}
  \hspace{-2em} \Box_g\zeta_{\mu}&\equiv & -4\nabla_{\kappa}\Box_gL_{\mu}{}^{\kappa}
    -8W_{\alpha\kappa \mu}{}^{\rho}\nabla^{\alpha}L_{\rho}{}^{\kappa}
     -4R_{\kappa\rho}\nabla_{\mu}L^{\rho\kappa}
       +8R_{\alpha\rho}\nabla^{\alpha}L_{\mu}{}^{\rho}
     \nonumber
    \\
    &&  -4L^{\rho\kappa}\nabla_{\mu}R_{\rho\kappa} + 4L^{\rho\kappa}\nabla_{\rho}R_{\mu\kappa}
    + \frac{2}{3}R_{\mu}{}^{\kappa}\nabla_{\kappa}R
      -R_{\mu}{}^{\nu}\zeta_{\nu} + \frac{1}{3}R_g\zeta_{\mu}
         \nonumber
  \\
  &&    +2L_{\mu}{}^{\rho}\nabla_{\rho}R_g
   + \frac{2}{3}R_g \nabla_{\mu}(L-\frac{1}{6}R)
  + 2\nabla_{\mu}\Box_g (L+\frac{1}{6}R)
       \nonumber
\\
 &=&   -8 \nabla^{\nu}[(W_{\mu\sigma\nu}{}^{\rho}-C_{\mu\sigma\nu}{}^{\rho})L_{\rho}{}^{\sigma}]
 +   4 \nabla_{\mu} [ (W_{\alpha\sigma}{}^{\alpha\rho}- C_{\alpha\sigma}{}^{\alpha\rho})  L_{\rho}{}^{\sigma}]
      \nonumber
 \\
  &&   -\frac{8}{3}  L_{\mu}{}^{\nu} \nabla_{\nu}(L - \frac{1}{6}R )
   + \frac{4}{9}R \nabla_{\mu}(L-\frac{1}{6}R)   -\frac{2}{3}(L-\frac{1}{6}R)\nabla_{\mu}  R
       \nonumber
    \\
    &&  +(4  L_{\mu}{}^{\nu}  - R_{\mu}{}^{\nu})\zeta_{\nu} + \frac{1}{3}(R_g-R)\zeta_{\mu}
 + f_{\mu}(x;H,\nabla H,\nabla K)
  \label{Wwave_zetavec}
 \;.
\end{eqnarray}

Finally, let us establish a wave equation which is satisfied by $\zeta_{\mu\nu\sigma}-\xi_{\mu\nu\sigma}$.
The definition of the Weyl tensor  together with the Bianchi identities yield
\begin{eqnarray*}
 \Box_g \zeta_{\mu\nu\sigma}
  &\equiv& 4\nabla_{[\sigma}\Box_g L_{\nu]\mu}
 - 4W_{\nu\sigma\kappa\rho}\nabla^{\rho}L_{\mu}{}^{\kappa}+8W_{\mu\kappa\rho[\sigma} \nabla^{\rho}L_{\nu]}{}^{\kappa}
  -  4R_{\kappa[\nu} \nabla_{\sigma]}L_{\mu}{}^{\kappa}
\\
 && + 4R_{\kappa[\sigma} \nabla_{|\mu|}L_{\nu]}{}^{\kappa} - 4R_{\mu[\sigma}  \nabla_{|\kappa|}L_{\nu]}{}^{\kappa}    -4R_{\rho\kappa} g_{\mu[\sigma} \nabla^{\rho}L_{\nu]}{}^{\kappa}   +  \frac{1}{3}R_g\zeta_{\mu\nu\sigma}
\\
 &&  + \frac{4}{3}R_g g_{\mu[\sigma}\nabla^{\kappa}L_{\nu]\kappa}  + 4L_{\mu}{}^{\kappa}\nabla_{[\nu}R_{\sigma]\kappa}
 +4 L_{\nu}{}^{\kappa}\nabla_{[\mu}R_{\kappa]\sigma} +  4 L_{\sigma}{}^{\kappa}\nabla_{[\kappa}R_{\mu]\nu}
\\
 &\overset{\eq{Wweylwave1}}{=}& 4 L_{\mu}{}^{\kappa}\zeta_{\kappa\nu\sigma}  - 4L_{\mu}{}^{\kappa}\nabla_{[\sigma}R_{\nu]\kappa}
 +  16 L_{\kappa[\nu} \nabla_{\sigma]} L_{\mu}{}^{\kappa}  -  4R_{\kappa[\nu} \nabla_{\sigma]}L_{\mu}{}^{\kappa}
\\
 && -8L_{\rho}{}^{\kappa}g_{\mu[\nu} \nabla_{\sigma]}L_{\kappa}{}^{\rho}   -4R_{\rho\kappa} g_{\mu[\sigma} \nabla^{\rho}L_{\nu]}{}^{\kappa}
 + 4R_{\kappa[\sigma} \nabla_{|\mu|}L_{\nu]}{}^{\kappa}
\\
 &&    - 4R_{\mu[\sigma}  \nabla_{|\kappa|}L_{\nu]}{}^{\kappa}
  +4  L_{\nu}{}^{\kappa}\nabla_{[\mu}R_{\kappa]\sigma} +  4  L_{\sigma}{}^{\kappa}\nabla_{[\kappa}R_{\mu]\nu}
  + \frac{4}{3}R_g g_{\mu[\sigma}\nabla^{\kappa}L_{\nu]\kappa}
\\
 && + 8 L_{\rho}{}^{\kappa}\nabla_{[\nu}C_{\sigma]\kappa\mu}{}^{\rho}
 + 4\zeta_{\alpha\kappa[\nu}C_{\sigma]}{}^{\kappa}{}_{\mu}{}^{\alpha}
 - 2\zeta_{\mu\alpha\kappa}W_{\nu}{}^{\alpha}{}_{\sigma}{}^{\kappa}
  + \frac{1}{3}R_{\sigma\nu\mu}{}^{\kappa}\nabla_{\kappa} R
\\
 &&  +8(W_{\mu}{}^{\rho}{}_{[\nu}{}^{\kappa}- C_{\mu}{}^{\rho}{}_{[\nu}{}^{\kappa}) \nabla_{|\kappa|} L_{\sigma]\rho}
  +  \frac{1}{3}R_g\zeta_{\mu\nu\sigma}
  +f_{\mu\nu\sigma}(x; H, \nabla H,\nabla K)
 \;.
\end{eqnarray*}
Using the relations \eq{W_Ricci_Schouten}, \eq{Weqn_curvscalarWeyl} and $\zeta_{[\mu\nu\sigma]}=0$ one then shows
\begin{eqnarray*}
 \Box_g \zeta_{\mu\nu\sigma}
&=&-  4\zeta_{\mu\kappa[\nu}  L_{\sigma]}{}^{\kappa} + 6g_{\mu[\nu} \zeta^{\kappa}{}_{\sigma\alpha]}L_{\kappa}{}^{\alpha}
  + 8  L_{\rho}{}^{\kappa}\nabla_{[\nu}C_{\sigma]\kappa\mu}{}^{\rho} + 4\zeta_{\alpha\kappa[\nu}C_{\sigma]}{}^{\kappa}{}_{\mu}{}^{\alpha}
\\
 &&
 - 2\zeta_{\mu\alpha\kappa}W_{\nu}{}^{\alpha}{}_{\sigma}{}^{\kappa}
 +8(W_{\mu}{}^{\rho}{}_{[\nu}{}^{\kappa}- C_{\mu}{}^{\rho}{}_{[\nu}{}^{\kappa}) \nabla_{|\kappa|} L_{\sigma]\rho}
   -2 L_{\mu[\nu}\zeta_{\sigma]}
\\
 && + \frac{1}{3}W_{\sigma\nu\mu}{}^{\kappa}\nabla_{\kappa} R
   - \frac{1}{6}(R-2R_g)   \zeta_{\mu\nu\sigma}
   +f_{\mu\nu\sigma}(x; H, \nabla H,\nabla K)
 \;.
\end{eqnarray*}
Combining with \eq{Wweylwave5} we infer that $\zeta_{\mu\nu\sigma}-\xi_{\mu\nu\sigma}$ fulfills the wave equation,
\begin{eqnarray}
 &&\hspace{-3em} \Box_g (\zeta_{\mu\nu\sigma}-\xi_{\mu\nu\sigma})
  \nonumber
  \\&=& 6g_{\mu[\nu}( \zeta^{\kappa}{}_{\sigma\alpha]}- \xi^{\kappa}{}_{\sigma\alpha]})L_{\kappa}{}^{\alpha}
  -  4(\zeta_{\mu\kappa[\nu}-\xi_{\mu\kappa[\nu})  L_{\sigma]}{}^{\kappa}
   + 4(\zeta_{\alpha\kappa[\nu}-\xi_{\alpha\kappa[\nu})C_{\sigma]}{}^{\kappa}{}_{\mu}{}^{\alpha}
 \nonumber
\\
 &&
 + \xi_{\mu\kappa}{}^{\alpha}(W_{\nu\sigma\alpha}{}^{\kappa}-C_{\nu\sigma\alpha}{}^{\kappa})
 +8(W_{\mu}{}^{\rho}{}_{[\nu}{}^{\kappa}- C_{\mu}{}^{\rho}{}_{[\nu}{}^{\kappa}{}) \nabla_{|\kappa|} L_{\sigma]\rho}
  -2 L_{\mu[\nu}\zeta_{\sigma]}
  \nonumber
\\
 &&- \frac{1}{3}(W_{\nu\sigma\mu}{}^{\kappa}-C_{\nu\sigma\mu}{}^{\kappa})\nabla_{\kappa} R
  + (\zeta_{\mu\kappa}{}^{\alpha} - \xi_{\mu\kappa}{}^{\alpha})W_{\nu\sigma\alpha}{}^{\kappa}
   +  \frac{1}{6} R(\zeta_{\mu\nu\sigma} - \xi_{\mu\nu\sigma})
 \nonumber
 \\
 && +  4L_{\mu[\nu}\nabla_{\sigma]}(L-\frac{1}{6}R)
   + \frac{8}{3}( L - \frac{1}{6}R ) g_{\mu[\sigma}\nabla^{\kappa}L_{\nu]\kappa}
     + \frac{2}{9}( L - \frac{1}{6}R ) g_{\mu[\nu}\nabla_{\sigma]} R
     \nonumber
\\
 && + \frac{2}{3}(L - \frac{1}{6}R)   \zeta_{\mu\nu\sigma}
 +f_{\mu\nu\sigma}(x; H, \nabla H,\nabla K)
 \label{Wwave_zeta-xi}
 \;.
\end{eqnarray}

The equations \eq{Wwave_H*Weyl}, \eq{Wwave_DiffH}, \eq{Wtr_L}, \eq{Wwave_W-C}, \eq{Wwave_zetavec} and \eq{Wwave_zeta-xi}
form a closed, linear,  homogeneous system of wave equations satisfied by $H^{\sigma}$, $K_{\mu\nu}$, $L-R/6$, $W_{\mu\nu\sigma\rho} -C_{\mu\nu\sigma\rho} $, $\zeta_{\mu}$ and $\zeta_{\mu\nu\sigma}-\xi_{\mu\nu\sigma}$, with $g_{\mu\nu}$, $L_{\mu\nu}$, etc.\  regarded as being given.
An application of standard uniqueness results for wave equations, cf.\ e.g.\  \cite{friedlander}, establishes that
all the fields vanish identically, supposing that this is initially the case.
In particular this guarantees
the vanishing of $H^{\sigma}$ and, via \eq{Weqn_curvscalarWeyl},  of $R_g-R$, and therefore consistency of the CWE2
with the gauge condition.

Moreover, the computations above reveal that the solution satisfies certain relations expected from the derivation of  the CWE2; e.g.\ it follows from \eq{Wweylwave6} that $L_{\mu\nu}$ is the Schouten tensor of $g_{\mu\nu}$ if $H^{\sigma}=0$ and $R_g=R$.

\begin{proposition}
\label{Wgauge_consistency}
Let us assume we have been given data ($\mathring g_{\mu\nu}$, $\mathring s$, $\mathring \Theta$, $\mathring L_{\mu\nu}$, $\mathring C_{\mu\nu\sigma}{}^{\rho}$, $\mathring \xi_{\mu\nu\sigma}$) on an initial surface $S$ (for definiteness we think either of two transversally intersecting null hypersurfaces or a light-cone) and a gauge source function $R$, such that $\mathring g_{\mu\nu}$ is the restriction to $S$ of a Lorentzian metric, $\mathring L_{\mu\nu}$ is symmetric and $\mathring L = \overline R/6$.
Suppose further that there exists a smooth solution ($g_{\mu\nu}$, $s$, $\Theta$, $L_{\mu\nu}$, $C_{\mu\nu\sigma}{}^{\rho}$, $\xi_{\mu\nu\sigma}$)  of the CWE2 \eq{Wweylwave1}-\eq{Wweylwave6} with gauge source function $R$ which induces the above data on $S$ and fulfills the following conditions:%
\begin{enumerate}
 \item $\overline H{}^{\sigma}[g]=0$,
 \item $\overline K_{\mu}{}^{\sigma}[g]=0$, where $K_{\mu}{}^{\sigma} \equiv \nabla_{\mu}H^{\sigma}$,
\item $\overline W_{\mu\nu\sigma}{}^{\rho}[g] = \overline C_{\mu\nu\sigma}{}^{\rho}$,
 \item $\overline \zeta_{\mu\nu\sigma}[g,L] = \overline \xi_{\mu\nu\sigma}$, where $\zeta_{\mu\nu\sigma}\equiv 4\nabla_{[\sigma}L_{\nu]\mu}$,
\item $\overline \zeta_{\mu}=0$, where $\zeta_{\mu}\equiv  -4\nabla_{\kappa}L_{\mu}{}^{\kappa} + 2\nabla_{\mu}L + \frac{1}{3}\nabla_{\mu}R$.
\end{enumerate}
Then
\begin{enumerate}
 \item[a)] $H^{\sigma}=0$ and $R_g=R$,
\item[b)] $C_{\mu\nu\sigma}{}^{\rho} $ is the Weyl tensor of $g_{\mu\nu}$,
\item[c)] $L_{\mu\nu}$ is the Schouten tensor of $g_{\mu\nu}$,
\item[d)]$\xi_{\mu\nu\sigma}$ is the Cotton tensor of $g_{\mu\nu}$.
\end{enumerate}
\end{proposition}
The validity of the assumptions 1-5  will be the subject of Section~\ref{Wsec_applicability}.

\subsection{Equivalence issue between the CWE2 and the MCFE}

We devote ourselves now to the issue to what extent and under which conditions a solution of the CWE2 is also a solution of the MCFE. It turns out that this issue is somewhat more intricate than for the CWE due to the change of variables.
Note that at this stage the cosmological constant $\lambda$ does not need to vanish.

\subsubsection*{A subsidiary system}

Recall the MCFE,  
\begin{eqnarray}
 && \nabla_{\rho} d_{\mu\nu\sigma}{}^{\rho} =0
 \;,
 \label{Wconf1*}
\\
 && \nabla_{\mu}L_{\nu\sigma} - \nabla_{\nu}L_{\mu\sigma} = \nabla_{\rho}\Theta \, d_{\nu\mu\sigma}{}^{\rho}
 \;,
 \label{Wconf2*}
\\
 && \nabla_{\mu}\nabla_{\nu}\Theta = -\Theta L_{\mu\nu} +  s  g_{\mu\nu}
 \;,
 \label{Wconf3*}
\\
 && \nabla_{\mu}  s = - L_{\mu\nu}\nabla^{\nu}\Theta
 \;,
 \label{Wconf4*}
\\
 && 2\Theta  s -  \nabla_{\mu}\Theta\nabla^{\mu}\Theta = \lambda/3
 \;,
 \label{Wconf5*}
\\
 &&  R_{\mu\nu\sigma}{}^{\kappa}[ g] = \Theta  d_{\mu\nu\sigma}{}^{\kappa} + 2\left( g_{\sigma[\mu} L_{\nu]}{}^{\kappa}  - \delta_{[\mu}{}^{\kappa} L_{\nu]\sigma} \right)
 \label{Wconf6*}
\;.
\end{eqnarray}
The MCFE 
are equivalent to the following system, supposing that $\Theta>0$,
\begin{eqnarray}
 &&
 \nabla_{\rho} C_{\nu\mu\sigma}{}^{\rho} = \nabla_{\mu}L_{\nu\sigma} - \nabla_{\nu}L_{\mu\sigma}  \;,
 \label{Wconf1**}
\\
 && \Theta(\nabla_{\mu}L_{\nu\sigma} - \nabla_{\nu}L_{\mu\sigma}) = \nabla_{\rho}\Theta \, C_{\nu\mu\sigma}{}^{\rho}  \;,
 \label{Wconf2**}
\\
 && \nabla_{\mu}\nabla_{\nu}\Theta = -\Theta L_{\mu\nu} +  s  g_{\mu\nu}  \;,
 \label{Wconf3**}
\\
 && \nabla_{\mu}  s = - L_{\mu\nu}\nabla^{\nu}\Theta  \;,
 \label{Wconf4**}
\\
 && 2\Theta  s -  \nabla_{\mu}\Theta\nabla^{\mu}\Theta = \lambda/3  \;,
 \label{Wconf5**}
\\
 &&  R_{\mu\nu\sigma}{}^{\kappa}[ g] = C_{\mu\nu\sigma}{}^{\kappa} + 2\left( g_{\sigma[\mu} L_{\nu]}{}^{\kappa}  - \delta_{[\mu}{}^{\kappa} L_{\nu]\sigma} \right)
 \label{Wconf6**}
\;.
\end{eqnarray}
This can be seen as follows:
Suppose we have a solution of \eq{Wconf1**}-\eq{Wconf6**}, then we obtain a solution of  \eq{Wconf1*}-\eq{Wconf6*} by identifying
$d_{\mu\nu\sigma}{}^{\rho}$ with $\Theta^{-1} C_{\mu\nu\sigma}{}^{\rho}$ and vice versa
(hence the system \eq{Wconf1**}-\eq{Wconf6**} is also equivalent to the vacuum Einstein equations for $\Theta>0$).
In fact, a solution of \eq{Wconf1*}-\eq{Wconf6*} provides a solution of \eq{Wconf1**}-\eq{Wconf6**} for any $\Theta$
since the identification of $C_{\mu\nu\sigma}{}^{\rho}$ with $\Theta d_{\mu\nu\sigma}{}^{\rho}$ is  possible, even where $\Theta=0$.

We elaborate in somewhat more detail on the characteristic initial value problem for an initial surface $S$ for which the set $\{\overline \Theta =0 \}$ is non-empty.
Since we are mainly interested in a light-cone with $\overline \Theta =0$ everywhere we specialise to the case $S=C_{i^-}$ (we then need to assume
$\lambda=0$).
Let us assume we have been given free initial data $\omega_{AB}\equiv \breve{\overline L}_{AB}$  on $C_{i^-}$,
and that the fields $\mathring g_{\mu\nu}$, $\mathring\Theta=0$, $\mathring s$, $\mathring L_{\mu\nu}$, $\mathring C_{\mu\nu\sigma}{}^{\rho}=0$
and $\mathring\xi_{\mu\nu\sigma}$
have been constructed by solving the constraint equations to be derived below  (cf.\ Section~\ref{Wsec_constraint_equations}).
Let us further assume that there exists a smooth solution  of the system \eq{Wconf1**}-\eq{Wconf6**} to the future of $S$ which induces these data on $S$
and which satisfies $s|_{i^-}\ne 0$. Then $\Theta$ has no zeroes inside the cone and sufficiently close to the vertex.
Moreover, cf.\ the proof of Lemma~\ref{append_lemma_s} in Appendix~\ref{cone_smoothness},
 $\mathrm{d}\Theta\ne 0 $ on $\Scri^-$ and  $\mathrm{d}\Theta|_{i^-} =0$.
Since the tensor $C_{\mu\nu\sigma}{}^{\rho}$ vanishes on $C_{i^-}$ the field
\( C_{\mu\nu\sigma}{}^{\rho}/ \Theta \)
can be smoothly
 continued across $\scri^-$ (though not necessarily across $i^-$).
The solution at hand thus solves \eq{Wconf1*}-\eq{Wconf6*} (except possibly at $i^-$) when identifying $C_{\mu\nu\sigma}{}^{\rho}/\Theta$ with
$d_{\mu\nu\sigma}{}^{\rho}$, smoothly continued across $\scri^-$.

The system  \eq{Wconf1**}-\eq{Wconf6**} is not regular for $\Theta=0$, and thus does not provide a good evolution system. However,
it turns out that it is equivalent to the CWE2, when the latter system is supplemented by the constraint equations,
and thus provides a useful tool to solve the equivalence issue between the MCFE and the CWE2.
The only grievance (or possibly advantage, we will come back to this issue later)  is that  we do not know how $d_{\mu\nu\sigma}{}^{\rho}$ behaves near the vertex, in particular it is by no means clear whether it can be continuously continued across past timelike infinity at all.
Nevertheless, the solution provides a solution of the MCFE up to and excluding the vertex, which induces the free initial
data $\omega_{AB}$ on $C_{i^-}$, and it provides a solution of the vacuum Einstein equations inside the cone, at least near $i^-$.

\subsubsection*{Equivalence of the CWE2 and the subsidiary system}

In this section we  address the equivalence issue  between the CWE2  \eq{Wweylwave1}-\eq{Wweylwave6}
and the subsidiary system \eq{Wconf1**}-\eq{Wconf6**} we just introduced and which, once we have constructed a solution thereof,  provides a solution of the MCFE \eq{Wconf1*}-\eq{Wconf6*}, with the possible exception of the vertex of the cone $C_{i^-}$.
For that we shall demonstrate that a solution of the CWE2 is a solution of the subsidiary system
 supposing that certain relations are satisfied on the initial surface, namely the constraint equations, cf.\  the next section.
The other direction follows from the derivation of the CWE2.
As initial surface we have, as before, two transversally intersecting null hypersurfaces or a light-cone in mind.

Recall the CWE2 \eq{Wweylwave1}-\eq{Wweylwave6}.
%
We assume we have been given a smooth solution $(g_{\mu\nu},L_{\mu\nu}, C_{\mu\nu\sigma}{}^{\rho},\xi_{\mu\nu\sigma},\Theta,s)$
with all the hypotheses of Proposition~\ref{Wgauge_consistency} being satisfied.
Then $L_{\mu\nu}$, $C_{\mu\nu\sigma}{}^{\rho}$ and $\xi_{\mu\nu\sigma}$ are the Schouten, Weyl and Cotton tensor of $g_{\mu\nu}$, respectively.
The equations \eq{Wconf1**} and \eq{Wconf6**} are thus identities and automatically satisfied.
Recall that it suffices for \eq{Wconf5**} to be satisfied at just one point.
Let us derive a homogeneous system of wave equations which establishes the validity of the remaining equations,
\eq{Wconf2**}-\eq{Wconf4**}.

It is convenient to make the following definitions:
\begin{eqnarray*}
 \Lambda_{\sigma\nu\mu} &:=& \frac{1}{2}\Theta\xi_{\sigma\nu\mu} + \nabla_{\rho}\Theta C_{\mu\nu\sigma}{}^{\rho}
\;,
\\
 \Xi_{\mu\nu} &:=& \nabla_{\mu}\nabla_{\nu}\Theta + \Theta L_{\mu\nu} - s g_{\mu\nu}
 \;,
\\
 \Upsilon_{\mu} &:=& \nabla_{\mu}s + L_{\mu\nu}\nabla^{\nu}\Theta
\;.
\end{eqnarray*}

Computations similar to the ones which led us to   \eq{conf3**_wave} and \eq{conf4**_wave}  (now with $H^{\sigma}$ and $K_{\mu\nu}$ vanishing)
reveal that, because of  \eq{Wweylwave1}-\eq{Wweylwave3}, we have
\begin{eqnarray}
\hspace{-2em}  \Box_g \Xi_{\mu\nu}
 &=& 2 \Xi_{\sigma\kappa}( 2 L_{(\mu}{}^{\kappa}\delta_{\nu)}{}^{\sigma}-g_{\mu\nu}L^{\sigma\kappa}- C_{\mu}{}^{\sigma}{}_{\nu}{}^{\kappa})
 + 4\nabla_{(\mu}\Upsilon_{\nu)} + \frac{1}{6}R\Xi_{\mu\nu}
 \label{Wconf3**_wave}
 \;,
\\
\hspace{-2em}  \Box_g\Upsilon_{\mu}
 &=& 6L_{\mu}{}^{\kappa}\Upsilon_{\kappa} +2L^{\rho\kappa}\Lambda_{\rho\kappa\mu}
 +2\Xi_{\nu}{}^{\sigma}\nabla_{\sigma}L_{\mu}{}^{\nu} - \frac{1}{6}\Xi_{\mu}{}^{\nu}\nabla_{\nu}R
 \label{Wconf4**_wave}
 \;.
\end{eqnarray}
Furthermore, in virtue of \eq{Wweylwave3}-\eq{Wweylwave5} and \eq{Wbianchi_weyl} we find that
\begin{eqnarray*}
\Box_g \Lambda_{\sigma\nu\mu}
 &=& s \xi_{\sigma\nu\mu} - 2L_{\rho\kappa} \nabla^{\kappa}\Theta \, C_{\mu\nu\sigma}{}^{\rho}
  -2  \nabla^{\rho}\Theta  C_{\sigma\rho\kappa[\mu}L_{\nu]}{}^{\kappa} -2 \nabla^{\rho}\Theta C_{\mu\nu\kappa[\sigma}L_{\rho]}{}^{\kappa}
\\
 && + \nabla^{\rho}\Theta(\nabla_{[\rho}\xi_{\sigma]\nu\mu}
  +\nabla_{[\mu}\xi_{\nu]\rho\sigma}   )   + \nabla^{\rho}\Theta\nabla_{\sigma}\xi_{\rho\nu\mu}
  + 4\Upsilon_{\rho} C_{\mu\nu\sigma}{}^{\rho}
\\
 &&
  + 2\Xi_{\kappa\rho} \nabla^{\kappa}C_{\mu\nu\sigma}{}^{\rho}
  + 4 C_{\sigma}{}^{\kappa}{}_{[\mu}{}^{\alpha}\Lambda_{|\kappa\alpha|\nu]}
 - C_{\mu\nu\alpha}{}^{\kappa}\Lambda_{\sigma\kappa}{}^{\alpha}
  +\frac{1}{3}R   \Lambda_{\sigma\nu\mu}
 \;.
\end{eqnarray*}
We observe the relation
\begin{eqnarray*}
 && \hspace{-4em} 2\nabla^{\rho}\Theta(\nabla_{[\rho}\xi_{\sigma]\nu\mu}
  +\nabla_{[\mu}\xi_{\nu]\rho\sigma}   )
\\
 &=&  4 \nabla^{\rho}\Theta(\nabla_{[\rho}\nabla_{\mu]}L_{\nu\sigma} - \nabla_{[\rho}\nabla_{\nu]}L_{\mu\sigma} +\nabla_{[\sigma}\nabla_{\nu]}L_{\mu\rho}
  -\nabla_{[\sigma}\nabla_{\mu]}L_{\nu\rho})
\\
 &=& 4 \nabla^{\rho}\Theta(C_{\mu\nu[\sigma}{}^{\kappa} L_{\rho]\kappa} - C_{\sigma\rho[\mu}{}^{\kappa} L_{\nu]\kappa}    )
 \;,
\end{eqnarray*}
which yields
\begin{eqnarray}
 \Box_g\Lambda_{\sigma\nu\mu}
 &=&  2\Xi_{\kappa\rho} \nabla^{\kappa}C_{\mu\nu\sigma}{}^{\rho}  +\xi^{\rho}{}_{\mu\nu}\Xi_{\sigma\rho}   + 2L_{\sigma}{}^{\rho} \Lambda_{\rho\nu\mu}
 + 4 C_{\sigma}{}^{\kappa}{}_{[\mu}{}^{\alpha}\Lambda_{|\kappa\alpha|\nu]}
  \nonumber
\\
 &&    - C_{\mu\nu\alpha}{}^{\kappa}\Lambda_{\sigma\kappa}{}^{\alpha}
  + 4\Upsilon_{\rho} C_{\mu\nu\sigma}{}^{\rho}
 +  \nabla_{\sigma}(\xi_{\rho\nu\mu}\nabla^{\rho}\Theta )
   +\frac{1}{3}R   \Lambda_{\sigma\nu\mu}
 \label{Wconf2**_wave}
 \;.
\phantom{xx}
\end{eqnarray}
%

It remains to derive a wave equation for $ \xi_{\rho\nu\mu}\nabla^{\rho}\Theta$ which follows from \eq{Wweylwave3}, \eq{Wweylwave5} and \eq{Wbianchi_weyl},

\begin{eqnarray}
 \Box_g (\xi_{\rho\nu\mu}\nabla^{\rho}\Theta ) &=&  \xi^{\rho}{}_{\nu\mu} \nabla_{\rho}\Box_g\Theta
 + 2\Xi^{\kappa\rho}( \nabla_{\kappa}\xi_{\rho\nu\mu}    +  2 L_{\kappa}{}^{\delta}C_{\mu\nu\delta\rho})
 - 4 L^{\kappa\rho} \nabla_{\kappa}\Lambda_{\rho\nu\mu}
 \nonumber
\\
 && +  \nabla^{\kappa} \Theta (4 L^{\delta\rho} \nabla_{\delta} C_{\mu\nu\rho\kappa}
  + 4 L_{\kappa}{}^{\rho} \xi_{\rho\nu\mu}
  + \Box_g \xi_{\kappa\nu\mu})  + \frac{1}{6} R \xi_{\rho\nu\mu}\nabla^{\rho}\Theta
 \nonumber
\\
  &=&  4\xi^{\rho}{}_{\nu\mu} \Upsilon_{\rho}
 + 2\Xi^{\kappa\rho}( \nabla_{\kappa}\xi_{\rho\nu\mu}    +  2 L_{\kappa}{}^{\delta}C_{\mu\nu\delta\rho})
 - 4 L^{\kappa\rho} \nabla_{\kappa}\Lambda_{\rho\nu\mu}
\nonumber
\\
 &&   - ( \xi_{\kappa\beta}{}^{\alpha} \nabla^{\kappa} \Theta )C_{\mu\nu\alpha}{}^{\beta}  + 4 \xi^{\alpha\beta}{}_{[\mu}\Lambda_{|\alpha\beta|\nu]}
  -\frac{1}{3}  \Lambda_{\rho\nu\mu}  \nabla^{\rho}R
\nonumber
\\
&& + \frac{1}{2} R\, \xi_{\rho\nu\mu}\nabla^{\rho}\Theta
 \;.
 \label{wave_thetaxi}
\end{eqnarray}

The equations  \eq{Wconf3**_wave}-\eq{wave_thetaxi} form a closed, linear, homogeneous system of wave equations
for the fields $\Xi_{\mu\nu}$, $\Upsilon_{\mu}$, $\Lambda_{\sigma\nu\mu}$
and $\xi_{\rho\nu\mu}\nabla^{\rho}\Theta$.
If we assume that the equations \eq{Wconf2**}-\eq{Wconf4**} are initially satisfied and that $\overline{\xi_{\rho\nu\mu}\nabla^{\rho}\Theta }=0$, we have vanishing initial data,
and standard uniqueness results for wave equations can be applied (cf.\ e.g.\ \cite{friedlander}) to conclude that \eq{Wconf2**}-\eq{Wconf4**}
are fulfilled.

As an extension of Proposition~\ref{Wgauge_consistency} we have proven the following result (note that the cosmological constant $\lambda$ is allowed to be non-vanishing):
%
\begin{theorem}
\label{Winter-thm}
Let us assume we have been given data ($\mathring g_{\mu\nu}$, $\mathring s$, $\mathring \Theta$, $\mathring L_{\mu\nu}$, $\mathring C_{\mu\nu\sigma}{}^{\rho}$, $\mathring \xi_{\mu\nu\sigma}$) on a characterteristic  initial surface $S$ (for definiteness we think either of two transversally intersecting null hypersurfaces or a light-cone) and a gauge source function $R$, such that $\mathring g_{\mu\nu}$ is the restriction of a Lorentzian metric, $\mathring L_{\mu\nu}$ is symmetric and $\mathring L = \overline R/6$.
Suppose further that there exists a smooth solution ($g_{\mu\nu}$, $s$, $\Theta$, $L_{\mu\nu}$, $C_{\mu\nu\sigma}{}^{\rho}$, $\xi_{\mu\nu\sigma}$)  of the CWE2 \eq{Wweylwave1}-\eq{Wweylwave6} with gauge source function $R$ which induces the above data on $S$ and satisfies the following conditions
(since it is the case of physical relevance we assume $\Theta\ne 0$ away from $S$; later on we shall consider only initial data where this is automatically the case, at least sufficiently close to $S$):
\begin{enumerate}
\item The equations \eq{Wconf2**}-\eq{Wconf5**} are satisfied on $S$ (it suffices if \eq{Wconf5**} holds at just  one point on $S$).
\item  The Weyl tensor of $g_{\mu\nu}$  coincides on $S$ with $C_{\mu\nu\sigma}{}^{\rho}$.
\item The relation $\xi_{\mu\nu\sigma}=4\nabla_{[\sigma}L_{\nu]\mu}$ holds on $S$.
\item The covector field $\zeta_{\mu}\equiv  -4\nabla_{\kappa}L_{\mu}{}^{\kappa} + 2\nabla_{\mu}L + \frac{1}{3}\nabla_{\mu}R$ vanishes on $S$.
\item The tensor field $\xi_{\rho\nu\mu}\nabla^{\rho}\Theta$ vanishes on $S$.
\item  The wave-gauge vector $H^{\sigma}$ and its covariant derivative $K_{\mu}{}^{\sigma}\equiv \nabla_{\mu}H^{\sigma}$ vanish on $S$.
\end{enumerate}
Then:
\begin{enumerate}
\item [a)] $H^{\sigma}=0$ and $R_g=R$.
 \item[b)]
The fields $C_{\mu\nu\sigma}{}^{\rho}$, $L_{\mu\nu}$ and $\xi_{\mu\nu\sigma}$ are the Weyl, Schouten and Cotton tensor of $g_{\mu\nu}$, respectively.
\item[c)] Set
 $d_{\mu\nu\sigma}{}^{\rho}:=\Theta^{-1}C_{\mu\nu\sigma}{}^{\rho}$ where $\Theta\ne 0$.
The tensor  $d_{\mu\nu\sigma}{}^{\rho}$ extends  to
the set $\{\overline  \Theta=0, \overline{\mathrm d\Theta}\ne 0 \}\subset S$.
Moreover,   the tuple $(g_{\mu\nu}\;,\, L_{\mu\nu}\;,\,\Theta\;,\,s\;,\,d_{\mu\nu\sigma}{}^{\rho})$  solves the MCFE \eq{Wconf1*}-\eq{Wconf6*} in the $(H^{\sigma}=0,R_g=R)$-gauge.
\end{enumerate}
The conditions 1-6 are  necessary for c) to be fulfilled.
\end{theorem}

We shall investigate next to what extent the conditions~1-6 are satisfied if the initial data  are constructed as solutions of the constraint equations induced by the MCFE on the initial surface.

\subsection{Constraint equations on $C_{i^-}$ in terms of Weyl and Cotton tensor}
\label{Wsec_constraint_equations}

\subsubsection*{Generalized wave-map gauge}

The aim of this section is to determine the constraint equations induced by the MCFE
on the fields $g_{\mu\nu}$, $L_{\mu\nu}$, $\Theta$, $s$, $C_{\mu\nu\sigma\rho}$ and $\xi_{\mu\nu\sigma}$.
For this purpose we assume we have been given some smooth solution ($g_{\mu\nu}$, $L_{\mu\nu}$, $\Theta$, $s$, $d_{\mu\nu\sigma\rho}$)
of the MCFE.
For simplicity and to avoid an exhaustive  case-by-case analysis
we shall restrict attention, as for the CWE, to the case where the initial surface is $S=C_{i^-}$.
This requires to assume
$$\lambda=0\;.$$
As a matter of course the constraints for  $g_{\mu\nu}$, $L_{\mu\nu}$, $\Theta$ and $s$ are the same as before,
cf.\ Section~\ref{constraints_generalizedw}.
The Weyl tensor vanishes on $\scri$  \cite{p2},
\begin{eqnarray}
 \overline C_{\mu\nu\sigma}{}^{\rho} =0\;.
\end{eqnarray}
It thus remains to determine the constraint equations for $\xi_{\mu\nu\sigma}$.
In adapted null coordinates the independent components of the Cotton tensor are
\begin{eqnarray*}
\overline  \xi_{00A}\;, \quad \overline \xi_{01A}\;, \quad \overline \xi_{11A}\;,
\quad \overline \xi_{A0B}\;, \quad \overline \xi_{A1B}\;, \quad \overline \xi_{ABC}\;.
\end{eqnarray*}
We have
\begin{eqnarray}
 \overline \xi_{01A} &=& 2(\overline{ \nabla_A L_{01}} - \overline{\nabla_1  L_{0A}})
 \label{Wcotton1}
  \;,
 \\
 \overline\xi_{11A} &=&  2(\overline{\nabla_A L_{11}} -  \overline{\nabla_1 L_{1A}})
  \label{Wcotton2}
  \,,
 \\
 \overline\xi_{A1B} &=& 2( \overline{\nabla_B L_{1A}} - \overline{\nabla_1 L_{AB}})
  \label{Wcotton3}
  \;,
 \\
 \overline\xi_{ABC} &=& 2( \overline{\nabla_C L_{AB}} -  \overline{\nabla_B L_{AC}})
  \label{Wcotton4}
  \;,
\end{eqnarray}
and no transverse derivatives of $L_{\mu\nu}$ are involved.
The remaining components follow from \eq{Wconf2*},
\begin{eqnarray}
 \overline\xi_{00A} &=& 2\nu^0\overline{\partial_0\Theta} \,\overline d_{010A}
  \label{Wcotton5}
  \;,
  \\
 \overline\xi_{A0B} &=& 2\nu^0\overline{\partial_0\Theta} \,\overline d_{0BA1}
 \label{Wcotton6}
  \;.
\end{eqnarray}

\subsubsection*{$(R=0,\overline s=-2,\kappa=0,\hat g = \eta)$-wave-map gauge}

To make computations easier we restrict attention to the  $(R=0, \overline s=-2, \kappa=0,\hat g=\eta)$-wave-map gauge.
Henceforth all equalities are meant to hold in this particular gauge.
As free initial data we take the  $s$-trace-free tensor $\omega_{AB}=\breve{\ol L}_{AB}$.
The constraint equations for  $\mathring g_{\mu\nu}$, $\mathring L_{\mu\nu}$ and  $\mathring C_{\mu\nu\sigma\rho}$ read (cf.\  \eq{constraint_g}-\eq{relation_lambda_omega})
\begin{eqnarray}
 &\mathring g_{\mu\nu} = \eta_{\mu\nu}\;, \enspace\,  \mathring C_{\mu\nu\sigma\rho}=0 \;, &
\label{W_constrainteqn_g}
\\
&\mathring L_{1\mu}=0\;,\enspace\,
 \mathring L_{0A} = \frac{1}{2}\tilde\nabla^B\lambda_{AB} \;, \enspace\,
  \mathring g^{AB} \mathring  L_{AB} = 0 \;,
 \label{relation_lambdaAB_omegaAB}&
\\
   &4(\partial_1+ r^{-1})\mathring L_{00}  =  \lambda^{AB} \omega_{AB} -4r \rho - 2\tilde \nabla^A \mathring L_{0A}
\;,&
\label{W_constrainteqn_L00}
\end{eqnarray}
where
\begin{eqnarray}
  (\partial_1- r^{-1})\lambda_{AB} &=& - 2 \omega_{AB}
 \;,
\\
 (\partial_1+\ 3r^{-1}) \rho &=& \frac{1}{2}r^{-1}\tilde\nabla^A \partial_1\mathring L_{0A} - \frac{1}{4}\lambda^{AB}\partial_1(r^{-1}\omega _{AB}) \;.
\label{W_constraint_equation_rho}
\end{eqnarray}
The relevant boundary conditions are
\begin{equation}
\mathring L_{00}= O(1)\;, \quad  \lambda_{AB}=O(r^3)\;, \quad \rho=O(1)
\;.
\label{W_boundary_values1}
\end{equation}
The equations \eq{Wcotton1}-\eq{Wcotton6} yield
\begin{eqnarray}
 \mathring \xi_{01A} &=& -2\partial_1 \mathring L_{0A}
 \,= \,\mathring g^{BC}\mathring \xi_{BAC}
  \label{constr_xi01A}
  \;,
 \\
 \mathring\xi_{11A} &=& 0
   \label{constr_xi11A}
  \;,
 \\
 \mathring\xi_{A1B} &=& -2r\partial_1( r^{-1}\omega_{AB})
   \label{constr_xiA1B}
  \;,
 \\
 \mathring\xi_{ABC} &=& 4\tilde\nabla_{[C}\omega_{B]A} - 4r^{-1}\mathring g_{A[B}\mathring L_{C]0}
  \label{constr_xiABC}
   \;,
\\
 \mathring\xi_{00A} &=&-4r\mathring d_{010A}
  \;, \text{ i.e.}
 \label{constr_xi00A_d}
\\
\partial_1 \mathring\xi_{00A} &=&\tilde\nabla^B(\lambda_{[A}{}^C\omega_{B]C})
- 2\tilde\nabla^B\tilde\nabla_{[A}\mathring L_{B]0} +\frac{1}{2}\tilde\nabla^B\mathring \xi_{A1B}
 \nonumber
\\
&& - 2 r\tilde\nabla_A\rho
 + r^{-1}\mathring \xi_{01A}  + \lambda_A{}^B\mathring \xi_{01B} \;,
  \label{constr_xi00A}
\\
 \mathring\xi_{A0B} &=& 2r\mathring g_{AB}\mathring d_{0101}  -2r \mathring d_{1A1B}  -  2  r \mathring d_{01AB}
 \nonumber
\\
  &=&  \lambda_{[A}{}^C\omega_{B]C} - 2\tilde\nabla_{[A}\mathring L_{B]0} + 2r\rho \,\mathring g_{AB} - \frac{1}{2}\mathring\xi_{A1B}
  \;,
    \label{constr_xiA0B}
\end{eqnarray}
with boundary condition
\begin{equation}
\mathring \xi_{00A}=O(r)
\;.
 \label{W_boundary_values2}
\end{equation}
We employed the $d_{\mu\nu\sigma\rho}$-constraints \eq{initial_d1A1B_spec_gauge}-\eq{initial_d010A_spec_gauge} to derive the expressions
for $\mathring \xi_{00A}$ and $\mathring \xi_{A0B}$ (recall that $\overline {\partial_0\Theta}=-2r$, cf.\ Section~\ref{kappa0_wavemap}).

Using the constraints for $\xi_{\mu\nu\sigma}$, one may rewrite the equation for $\rho$,
\begin{eqnarray}
8 (\partial_1+\ 3r^{-1})\rho &=& r^{-1}\lambda^{AB}\mathring \xi_{A1B} - 2r^{-1}\tilde\nabla^A\mathring\xi_{01A}
\;.
\end{eqnarray}
Note that for the Cotton tensor to be regular at $i^-$ the initial data necessarily need to satisfy $\omega_{AB} = O(r^3)$, cf.\ \eq{constr_xiA1B}.

\subsection{Applicability of Theorem~\ref{Winter-thm} on the $C_{i^-}$-cone}
\label{Wsec_applicability}

Let us assume we have been given initial data $\omega_{AB}=O(r^3)$ on $C_{i^-}$,
such that a smooth solution of the CWE2 exists in some neighbourhood to the future of $i^-$,
smoothly extendable through $C_{i^-}$,
 which induces
the prescribed data $\mathring \Theta=0$, $\mathring s=-2$, $\mathring C_{\mu\nu\sigma}{}^{\rho}=0$, $\mathring g_{\mu\nu}=\eta_{\mu\nu}$,
$\mathring L_{\mu\nu}$ and $\mathring \xi_{\mu\nu\sigma}$ on $C_{i^-}$, the last two fields determined
from the hierarchical system of constraint equations \eq{W_constrainteqn_g}-\eq{constr_xiA0B}.
We want to investigate to what extent the hypotheses of Theorem~\ref{Winter-thm} are satisfied under these assumptions.

For convenience let us recall the CWE2 in an $(R=0)$-gauge,
\begin{eqnarray}
  \Box^{(H)}_{ g} L_{\mu\nu}&=&  4 L_{\mu\kappa} L_{\nu}{}^{\kappa} -  g_{\mu\nu}| L|^2
  - 2C_{\mu\sigma\nu}{}^{\rho}  L_{\rho}{}^{\sigma}
   \;,
 \label{Wweylwave1*}
\\
 \Box_{ g}  s  &=& \Theta| L|^2
  \;,
 \label{Wweylwave2*}
\\
  \Box_{ g}\Theta &=& 4 s
  \;,
 \label{Wweylwave3*}
\\
 \Box^{(H)}_g C_{\mu\nu\sigma\rho} &=&    C_{\mu\nu\alpha}{}^{ \kappa} C_{\sigma\rho\kappa}{}^{\alpha}  -4C_{\sigma\kappa[\mu}{}^{\alpha}  C_{\nu]\alpha \rho}{}^{ \kappa}     - 2C_{\sigma\rho\kappa[\mu} L_{\nu]}{}^{ \kappa}    - 2 C_{\mu\nu\kappa[\sigma} L_{\rho]}{}^{\kappa}
 \nonumber
\\
 &&
   -\nabla_{[\sigma}\xi_{\rho]\mu\nu}  -\nabla_{[\mu}\xi_{\nu]\sigma\rho}
  \;,
 \label{Wweylwave4*}
\\
 \Box^{(H)}_{ g}\xi_{\mu\nu\sigma} &=& 4 \xi_{\kappa\alpha[\nu} C_{\sigma]}{}^{\alpha}{}_{\mu}{}^{\kappa}
 + C_{\nu\sigma\alpha}{}^{\kappa}\xi_{\mu\kappa}{}^{\alpha}
 -  4 \xi_{\mu\kappa[\nu}L_{\sigma]}{}^{\kappa}  + 6g_{\mu[\nu} \xi^{\kappa}{}_{\sigma\alpha]}L_{\kappa}{}^{\alpha}
 \nonumber
\\
 && + 8L_{\alpha\kappa} \nabla_{[\nu}C_{\sigma]}{}^{\alpha}{}_{\mu}{}^{\kappa}
  \;,
 \label{Wweylwave5*}
\\
  R^{(H)}_{\mu\nu}[g] &=& 2L_{\mu\nu}
 \label{Wweylwave6*}
 \;.
\end{eqnarray}

\subsubsection*{Vanishing of $\overline H^{\sigma}$}

This can be shown in exactly the same manner as for the CWE, Section~\ref{subsec_vanish_H}.

\subsubsection*{Vanishing of $\overline {\partial_0 H^{\sigma}}$ and $\overline \zeta_{\mu}$}

We know that the wave-gauge vector fulfills the wave equation \eq{Wwave_HWeyl} with $R=0$,
\begin{equation}
 \nabla^{\nu} \hat\nabla_{\nu} H^{ \alpha}+2g^{\mu\alpha} \nabla_{[\sigma} \hat\nabla_{\mu]} H^{ \sigma}
 + 4\nabla^{\nu} L_{\nu}{}^{\alpha}-2\nabla^{\alpha}L =0
 \label{Wwave_H+}
 \;.
\end{equation}
As for the CWE the vanishing of $\overline{\partial_0 H^0}$ and $\overline{\partial_0 H^A}$ follows from \eq{Wwave_H+} with $\alpha=0,A$
by taking regularity at the vertex into account.
Taking the trace of  the restriction of \eq{Wweylwave6*}
to the initial surface
then shows that the curvature scalar vanishes initially,
$  \overline  R_g =0$.

The $\alpha=1$-component of \eq{Wwave_H+} can be written  as
\begin{eqnarray}
 (\partial_1 + r^{-1})\overline{\partial_0 H^1}   + 2(\partial_1 + \tau)\overline  L_{00}  + 2\tilde\nabla^A\overline L_{0A}
 - \overline g^{AB}\overline{\partial_0 L_{AB}}=0
 \label{WODE_transH1}
 \;,
\end{eqnarray}
where we used that $\overline{\partial_0 L_{11}}=0$, cf.\ \eq{vanishing_transL11}, and that
\begin{eqnarray}
  \overline{\partial_{0}L} \,=\, \overline{\partial_0( g^{\mu\nu}L_{\mu\nu})} \,=\,  2\overline{\partial_0 L_{01}} + \overline g^{AB}\overline{\partial_0 L_{AB}}
  -   \lambda^{AB}\overline L_{AB}
 \label{Weqn_transtrace}
 \;.
\end{eqnarray}
We observe that, although we do not know yet whether $\overline {\partial_0 L}$ vanishes, equation \eq{WODE_transH1} coincides with \eq{ODE_transH1} of Section~\ref{sect_vanishing_nablaH}, and thus
the vanishing of $\overline{\partial_0 H^1}$ can be established by proceeding in exactly the same manner as for the CWE;
 one first shows that the source terms in \eq{WODE_transH1} vanishes and then utilizes regularity to deduce the desired result.
Altogether we have
\begin{equation}
 \overline{\nabla_{\mu}H^{\nu}}=0
 \label{Wvanishing_transH}
 \;.
\end{equation}


Inserting the definition \eq{Wwave_H*Weyl} of $\zeta_{\mu}$ into \eq{Wwave_H+}  yields
\begin{equation}
 \overline \zeta_{\mu} =0
 \label{Wvanishing_zeta}
 \;.
\end{equation}

\subsubsection*{Vanishing of $\overline{\xi_{\rho\nu\mu}\nabla^{\rho}\Theta}$}

Since $\overline \Theta=0$, it suffices to show that $\overline\xi_{1\mu\nu}=0$. Invoking the symmetries of $\overline \xi_{\mu\nu\sigma}$, we deduce from the constraint equations
\eq{constr_xi01A}-\eq{constr_xiA0B} that
\begin{eqnarray*}
 \overline \xi_{101} &\equiv& \overline g^{AB} \overline \xi_{A1B} =0
 \;,
\\
 \overline \xi_{10A} &\equiv& \overline g^{BC} \overline \xi_{BAC} -\overline \xi_{01A} - \overline \xi_{11A} =0
 \;,
\\
 \overline \xi_{11A} &=& 0
 \;,
\\
 \overline \xi_{1AB} &\equiv& -2\overline \xi_{[AB]1} =0
\;.
\end{eqnarray*}

\subsubsection*{Vanishing of $\overline{\zeta_{\mu\nu\sigma} -\xi_{\mu\nu\sigma}}$}

We need to show that
\begin{eqnarray*}
 \overline \xi_{\mu\nu\sigma} = \overline \zeta_{\mu\nu\sigma} \equiv 4\overline{\nabla_{[\sigma} L_{\nu]\mu}}
\;.
\end{eqnarray*}
 For the components $\overline \xi_{01A}$,  $\overline \xi_{11A}$,  $\overline \xi_{A1B}$
 and  $\overline \xi_{ABC}$ this follows straightforwardly  from the constraint equations
\eq{constr_xi01A}-\eq{constr_xiABC}.
The remaining independent components $\overline\xi_{00A}$ and $\overline\xi_{A0B}$ are determined by \eq{constr_xi00A} and \eq{constr_xiA0B},
respectively.
We observe that $\overline\zeta_{00A} - \overline\xi_{00A}$
and $\overline\zeta_{A0B}-\overline \xi_{A0B}$ satisfy the same equations as the components $\overline\varkappa_{00A}/2$ and $\overline\varkappa_{A0B}/2$
in Section~\ref{sect_vanishing_xi}, so one just needs to repeat the computations carried out there to accomplish the proof that $\overline \xi_{\mu\nu\sigma} = \overline \zeta_{\mu\nu\sigma}$.

\subsubsection*{Vanishing of $\overline W_{\mu\nu\sigma}{}^{\rho}$}

In the same manner as for the CWE, Section~\ref{sect_vanishing_W}, one shows that the Weyl tensor $ W_{\mu\nu\sigma}{}^{\rho}$ of $g_{\mu\nu}$ vanishes
initially.

\subsubsection*{Validity of the equations \eq{Wconf2**}-\eq{Wconf5**} on $C_{i^-}$}

The validity of  \eq{Wconf2**} on $C_{i^-}$ follows from the vanishing of $\overline \Theta$ and $\overline C_{\mu\nu\sigma\rho}$.
The computation which shows the vanishing of \eq{Wconf3**}-\eq{Wconf5**} is identical to the one we did for the CWE, cf.\ Sections~\ref{sect_validity}-\ref{sect_vanishing_Xi}.

\subsection{Main result concerning the CWE2}

We end up with the following result, which is in close analogy with Theorem~\ref{main_result}:
\begin{theorem}
\label{main_result2}
   Let us suppose we have been given a smooth one-parameter family of $s$-traceless tensors $\omega_{AB}(r,x^A)=O(r^3)$ on the 2-sphere, where $s$ denotes the standard metric.
A smooth solution
$(g_{\mu\nu},  L_{\mu\nu}, C_{\mu\nu\sigma}{}^{\rho}, \xi_{\mu\nu\sigma}, \Theta,  s)$
of the CWE2 \eq{Wweylwave1*}-\eq{Wweylwave6*} to the future of $C_{i^-}$, smoothly extendable through $C_{i^-}$,
 with initial data
$(\mathring g_{\mu\nu}, \mathring L_{\mu\nu}, \mathring C_{\mu\nu\sigma}{}^{\rho}, \mathring\xi_{\mu\nu\sigma},\mathring \Theta=0, \mathring s= -2 )$,
where
$ \breve{\mathring L}_{AB}=\omega_{AB}$,
provides a solution
$$( g_{\mu\nu},  L_{\mu\nu},  d_{\mu\nu\sigma}{}^{\rho} = \Theta^{-1}C_{\mu\nu\sigma}{}^{\rho}, \Theta, s )$$
of the MCFE \eq{Wconf1*}-\eq{Wconf6*} with $\lambda=0$, smoothly continued across $\scri^-$,
 in a neighbourhood of $i^-$ intersected with $J^+(i^-)$, with the possible exception of $i^-$ itself,
 in the
$$(R=0,\overline  s=-2,\kappa=0,\hat g_{\mu\nu} = \eta_{\mu\nu} ) \text{-wave-map gauge}$$
if and only if the initial data have their usual symmetry properties and satisfy the constraint equations \eq{W_constrainteqn_g}-\eq{W_constraint_equation_rho}
and \eq{constr_xi01A}-\eq{constr_xiA0B} with boundary conditions \eq{W_boundary_values1} and \eq{W_boundary_values2},%
\footnote{Note that if $s|_{i^-}<0$ then $\Theta$ is positive in the interior of $C_{i^-}$ and sufficiently close to $i^-$ and  $\mathrm{d}\Theta\ne 0 $ on $C_{i^-}\setminus \{i^-\}$
near  $i^-$, so a  solution of the CWE2 provides a solution of the  MCFE  in $J^+(i^-)\setminus \{i^-\}$ sufficiently close to $i^-$.}
\end{theorem}

\begin{remark}
Note that regularity for the Cotton tensor  implies that the initial data necessarily need to satisfy $\omega_{AB}(r,x^A)=O(r^3)$, cf.\ equation \eq{constr_xiA1B}.
\end{remark}

\section{Conclusions and outlook}
\label{sec_conclusions}

Let us finish by briefly comparing the two systems of wave equations, CWE and CWE2, which we have studied here, and by summarizing  the results
we have established for them.

\subsection{Comparison of both systems CWE \& CWE2}
\label{comparison}

It might be advantageous in certain situations that the Schouten, Weyl and Cotton tensor, which appear in the CWE2-system,  can be directly expressed in terms of the metric. 
In contrast, the rescaled Weyl tensor, which is an unknown of the CWE, can be defined on $\scri$ in terms of the metric and the conformal factor only via a limiting process from the inside.

Once a smooth solution of the CWE has been constructed
 (we think of a characteristic Cauchy problem with data on $C_{i^-}$),
 it is, as a matter of course, known that the rescaled Weyl tensor is regular at $i^-$.
Since both, $\Theta$ and $\mathrm{d}\Theta$, vanish at $i^-$ the same conclusion cannot be straightforwardly drawn
for a solution of CWE2,
even if one takes initial data $\omega_{AB}=O(r^4)$.
Note for this  that the constraint equations for CWE2 are somewhat ``weaker'' than the constraints for the CWE involving the rescaled
Weyl tensor, which is due to the fact that the Cotton tensor has less independent components than the rescaled Weyl tensor. It is the
$\breve{ d}_{0A0B}$-constraint which has no equivalent in the CWE2-system.
Thus it seems to be plausible that the conclusions are weaker, too. One has no control how  $C_{\mu\nu\sigma\rho}/\Theta$ behaves near the vertex.
It  seems to be hopeless to catch the behaviour of $d_{\mu\nu\sigma}{}^{\rho}$ when approaching $i^-$
 in terms of the initial data on $C_{i^-}$.

However, this can be seen as an advantage as well, for there seems to be no reason why the rescaled Weyl tensor should be regular at $i^-$.
It might be more sensible
to assume just the unphysical metric to be regular there.
Note, however,  that for analytic data the rescaled Weyl tensor will be regular at $i^-$ \cite{F10}, while for smooth data this is an open issue.
The CWE2 might be predestined to construct solutions of the Einstein equations with a rescaled Weyl tensor which cannot be extended across $i^-$, supposing of course that such solutions do exist at all.
In fact, we have seen that any smooth solution of the CWE (supplemented by the constraint equations) necessarily requires initial data $\omega_{AB}=O(r^4)$,
while for the CWE2 we just needed to require $\omega_{AB}=O(r^3)$.
So if one is able to construct a solution of the CWE2 from free initial data $\omega_{AB}$ which are properly $O(r^3)$, the corresponding solution of the MCFE will lead to a rescaled Weyl tensor which could not be regular at $i^-$.

\subsection{Summary and outlook}

Both  CWE and  CWE2 have been extracted from the MCFE by imposing a generalized wave-map gauge condition.
Similar to Friedrich's reduced conformal field equations, they provide, in $3+1$ dimensions,  a well behaved system of evolution equations.
The main object of this paper was to investigate the issue under which conditions a solution of the CWE/CWE2 is also a solution of the MCFE.
Since, roughly speaking, the CWE/CWE2 have been derived from the MCFE by differentiation, one needs to make sure, on characteristic initial surfaces, that the MCFE are initially satisfied, as made rigorous by   Theorems~\ref{inter-thm} and \ref{Winter-thm}.

One would like to construct the initial data for the CWE/CWE2 in such a way, that all the hypothesis in these theorems are fulfilled.
The expectation is that this is the case whenever the data are constructed from suitable  free ``reduced'' data by solving a set of constraint equations induced by the MCFE on the initial surface, which is a hierarchical system of algebraic equations
and ODEs, as typical for characteristic initial value problems for Einstein's vacuum field equations.
In this work, we have restricted attention  to the $C_{i^-}$-cone, which requires $\lambda=0$, and, for computational purposes, to a specific gauge, and showed that this is indeed the case, cf.\   Theorems~\ref{main_result} and \ref{main_result2}.

Analogous results should be expected to hold for e.g.\ a light-cone for $\lambda\geq 0$ whose vertex is located at $\scri^-$, or for two transversally intersecting
null hypersurfaces, one of which belongs to $\scri^-$ for $\lambda=0$, or where the intersection manifold is located at $\scri^-$ for $\lambda >0$,
and also for such surfaces with a
vertex, or intersection manifold, located in the physical space-time.
Furthermore,  any  generalized wave-map gauge with sufficiently well-behaved gauge functions should lead to the same conclusions.
We will not work out the details here.
 It should also be clear that  results similar to  Theorems~\ref{main_result} and \ref{main_result2} can be established with initial data of finite differentiability.

The equivalence issue between CWE/CWE2 and MCFE is also of relevance for spacelike Cauchy problems.
This  has been analysed in \cite{ttp2}. It is shown there that, roughly speaking, a solution of the CWE is a solution
of the MCFE if the MCFE and their transverse derivatives are satisfied on the initial surface.
As in the characteristic case, it should be expected that this can be guaranteed whenever the initial data are constructed as solutions of an appropriate set of constraint equations. 
In \cite{ttp2} this has been proved to be the case if the initial surface is a spacelike $\scri^-$ (here one needs to assume $\lambda>0$).


\vspace{1.2em}

\noindent {\textbf {Acknowledgements}}
I am grateful to  my advisor Piotr T.\ Chru\'sciel for numerous stimulating discussions and suggestions, and for reading a first draft of this article.
Moreover, it is a  pleasure to thank  Helmut Friedrich for his valuable input to rewrite the conformal field equations as a system of wave equations.
This work was supported in part by the  Austrian Science Fund (FWF) P 24170-N16.

\appendix

\section{Cone-smoothness and proof of Lemma~\ref{extension_s}}
\label{cone_smoothness}

In order to prove Lemma~\ref{extension_s} we need some facts about \textit{cone-smooth} functions.
Therefore let us briefly review the notion of cone-smoothness as well as some basic properties of cone-smooth functions.
For the details we refer the reader to~\cite{ChJezierskiCIVP,C1}.

We denote by $\{y^0\equiv t,y^i\}$ coordinates in a  $4$-dimensional spacetime for which
\begin{equation*}
 C_O := \Big\{ y^{\mu} \in \mathbb{R}^{4} : y^0= \sqrt{\sum_i (y^i)^2} \Big\}
  \;.
\end{equation*}
is the light-cone emanating from a point $O$. Such coordinates exist at least sufficiently close to the vertex.
Adapted null coordinates are denoted by $\{x^0\equiv u, x^1\equiv r, x^A\}$.
Both coordinate systems are related via a transformation of the form (cf.\ \cite{ChJezierskiCIVP}),
\begin{eqnarray*}
 y^0 = x^1 - x^0\;, \quad y^i=x^1\Theta^i(x^A) \quad \text{with} \quad \sum_{i=1}^3[\Theta^i(x^A)]^2=1\;.
\end{eqnarray*}

\begin{definition}[\cite{ChJezierskiCIVP}]
 A function $\varphi$ defined on $C_O$ is said to belong to $\mycal C^k(C_O)$, $k\in\mathbb{N}\cup\{\infty\}$, if it can be written as $\hat\varphi+ r\check\varphi$ with
 $\hat\varphi$ and $\check\varphi$ being $\mycal{C}^k$-functions of $y^i$. If $k=\infty$ the function $\varphi$ is called cone-smooth.
\end{definition}

\begin{remark}
 We are particularly interested in the cone-smooth case $k=\infty$.
\end{remark}

\begin{proposition}[\cite{ChJezierskiCIVP}]
 \label{prop_cone_smooth}
 Let  $\varphi: C_O \rightarrow \mathbb{R}$  be a function and let $k\in\mathbb{N}\cup\{\infty\}$. The following statements are equivalent:
 \begin{enumerate}
 \item[(i)] The function $\varphi$ can be extended to a $\mycal{C}^k$ function on $\mathbb{R}^{4}$.
 \item[(ii)] $\varphi\in \mycal{C}^k(C_O)$.
 \item[(iii)] The function $\varphi$ admits an expansion of the form ($ \varphi_{i_1 \dots i_p}, \varphi'_{i_1 \dots i_{p-1}}\in \mathbb{R}$)
  \[ \varphi = \sum_{p=0}^{k} \varphi_p r^p + o_k(r^k) \enspace\text{where} \enspace \varphi_p := \varphi_{i_1 \dots i_p} \Theta^{i_1}\dots \Theta^{i_p} + \varphi'_{i_1 \dots i_{p-1}} \Theta^{i_1}\dots \Theta^{i_{p-1}}\;. \]
 \end{enumerate}
\end{proposition}

\begin{lemma}[\cite{C1}]
\label{lemma_cone_smooth}
  Let $\varphi\in\mycal{C}^k(C_O)$ with $k\in\mathbb{N}\cup\{\infty\}$. Then
 \begin{enumerate}
 \item[(i)]  $ \exp(\varphi)\in \mycal{C}^k(C_O)$,  and
  \item[(ii)] $ r^{-1} \int_0^r  \varphi(\hat r, x^A)\,\mathrm{d}\hat r \in \mycal{C}^k(C_O)$.
 \end{enumerate}
  If, in addition, $\varphi(0)=0$, then
 \begin{enumerate}
 \item[(iii)]  $ \int_0^r \hat r^{-1} \varphi(\hat r, x^A)\,\mathrm{d}\hat r \in \mycal{C}^k(C_O)$.
 \end{enumerate}
\end{lemma}

\begin{lemma}
\label{append_lemma_s}
Consider any smooth solution of the MCFE in $4$ spacetime dimensions
 in some neighbourhood $\mycal U$ to the future  of $i^-$,
smoothly extendable through $C_{i^-}$,
 which satisfies
 \begin{equation}
 s_{i^-} := s|_{i^-} \ne 0\;.
\end{equation}
 Let $\rho$ be any
  function on $\mycal U\cap \partial J^+(i^-)$ with $\rho_{i^-}:=\rho|_{i^-}\ne 0$
and $\lim_{r\rightarrow 0}\partial_r\rho =0$
which can be extended to a smooth spacetime function.
 Then the equation
\begin{eqnarray}
 \overline{\nabla_{\mu}\Theta\nabla^{\mu} \mathring \phi + \mathring\phi\, s - \mathring \phi^{2}\rho}=0
 \label{B_ODE_omega_rho_pre}
\end{eqnarray}
is a Fuchsian
 ODE for $\mathring \phi$ and any solution satisfies (set $ \mathring\phi_{i^-} := \mathring \phi|_{i^-} $)
\begin{equation}
\mathrm{sign}(\mathring\phi_{i^-} ) =  \mathrm{sign}(s_{i^-}) \mathrm{sign}(\rho_{i^-})
\end{equation}
(in particular $\mathring\phi_{i^-}\ne 0$), and is the restriction to $C_{i^-}$ of a smooth spacetime function.
\end{lemma}

\begin{proof}
We assume a sufficiently regular gauge so that the regularity conditions (4.41)-(4.51) in~\cite{CCM2} hold.
Evaluation of the MCFE \eq{conf3} on $\scri^-$ in coordinates adapted to the cone implies the relation
\begin{eqnarray}
  \overline g^{AB}\overline{ \nabla_A\nabla_B \Theta }= 2 \overline s \quad \Longleftrightarrow \quad
 \tau  \overline{\partial_0 \Theta} = 2\nu_0 \overline s
\;.
 \label{rel_tau_transTheta}
\end{eqnarray}
The notation is introduced at the beginning of Section~\ref{sec_constraint_equations}. The expansion $\tau$ of the light-cone satisfies
\begin{equation*}
  \tau = \frac{2}{r} + O(r)\;,
  \quad \partial_1 \tau = -\frac{2}{r^2} + O(1)
  \;.
\end{equation*}
Moreover, regularity requires
\begin{equation*}
\nu_0 = 1+ O(r^2)\;, \enspace \partial_1\nu_0 = O(r)\;, \enspace \overline s=O(1)\;, \enspace \partial_1 \overline s =O(1)
 \;.
\end{equation*}
Hence
\begin{equation}
 \overline{\partial_0 \Theta} = s_{i^-} r  + O(r^2) \quad \text{and} \quad
 \partial_1\overline{\partial_0 \Theta} = s_{i^-}   + O(r)
 \;.
  \label{vertex_behaviour_Theta}
\end{equation}
The $r$-component of the MCFE \eq{conf4} yields
\begin{eqnarray*}
 \partial_1 \overline s + \nu^0 \overline L_{11}\overline{\partial_0 \Theta} =0
\quad \Longrightarrow \quad
  \partial_1 \overline s |_{i^-} =0
\end{eqnarray*}
due to regularity (note that $\overline L_{11}=O(1)$), i.e.\
\begin{eqnarray*}
 \overline s = s_{i^-} + O(r^2)
 \;.
\end{eqnarray*}
%

In adapted null coordinates \eq{B_ODE_omega_rho_pre} reads
\begin{eqnarray}
 \nu^0\overline{\partial_0\Theta}\partial_1 \mathring\phi +  \overline s \mathring \phi- \rho\mathring \phi^{2}=0
 \label{ODE_omega_rho}
 \;,
\end{eqnarray}
i.e., since $\mathrm{d}\Theta|_{\scri^-}\ne 0$, \eq{B_ODE_omega_rho_pre} is a Fuchsian
ODE for $\mathring\phi$ along the null geodesics emanating from $i^-$.
By assumption, the functions $\overline s$
and $\rho$ are cone-smooth.
In~\cite{C1} it is shown that $\nu^0$ and $r\tau$  are cone-smooth.
That implies that the function
\begin{eqnarray*}
 \psi := \frac{\overline{\partial_0\Theta}}{r} \overset{\eq{rel_tau_transTheta}}{=} \frac{2\nu_0 \overline s}{r \tau}
 = s_{i^-} + O(r^2)
\end{eqnarray*}
is cone-smooth, as well (note that $(r\tau)|_{i^-}\ne 0$).
 Since we have assumed $s_{i^-}\ne 0$, the function $\psi$ has no zeros near $i^-$, so $\psi^{-1}$ exists  near $i^-$ and is cone-smooth.
 The ODE \eq{ODE_omega_rho} thus takes the form
  \begin{eqnarray}
 r\partial_1 \mathring\phi +  \hat\omega \mathring \phi - \omega\mathring \phi^{2}=0
  \;,
  \label{form_ODE_omega}
 \end{eqnarray}
 where the functions $\hat\omega :=\nu_0 s\psi^{-1}=1+O(r^2)$ and $\omega:=\nu_0\rho \psi^{-1}=\frac{\rho_{i^-}}{s_{i^-}}+  O(r^2)$
 are  cone-smooth and non-vanishing near the tip of the cone.%
\footnote{
\label{comment_rho}
The absence of first-order terms is due to the vanishing of $\partial_1\overline s|_{i^-}$ and $\partial_1 \rho|_{i^-}$. In fact, a term proportional to $r$ in  $\omega$ would produce
logarithmic terms in the expansion of $\zeta$, \eq{definition_zeta}, and thereby in the expansions of
$r\gamma$, \eq{expan_r_gamma}, and $\mathring\phi$, \eq{expan_mathring_phi}.}

Let $\varepsilon>0$ be sufficiently small.
 We introduce the function
 \begin{equation}
  \gamma := e^{-\int_{\varepsilon}^r \hat r^{-1} \hat\omega\,\mathrm{d}\hat r}\,\mathring \phi^{-1}
   \;,
 \end{equation}
 so that \eq{form_ODE_omega} becomes
   \begin{eqnarray}
 r^2\partial_1 \gamma  + \zeta=0
  \;,
 \end{eqnarray}
where
\begin{eqnarray}
 \zeta:= \varepsilon\omega e^{-\int_{\varepsilon}^r \hat r^{-1}(\hat\omega- 1)\,\mathrm{d}\hat r}
  =\underbrace{\varepsilon\frac{\rho_{i^-}}{s_{i^-}} e^{\int_0^{\varepsilon} \hat r^{-1}(\hat\omega- 1)\,\mathrm{d}\hat r}}_{=:c} + O(r^2)
  \label{definition_zeta}
\end{eqnarray}
 is cone-smooth by Lemma~\ref{lemma_cone_smooth} and has a sign near the tip,
\begin{equation*}
 \mathrm{sign}(\zeta_{i^-}) \,=\, \mathrm{sign}(s_{i^-}) \mathrm{sign}(\rho_{i^-})
\;.
\end{equation*}
Consequently,
 \begin{eqnarray}
  r\gamma \,=\, -r \int   \hat r^{-2}\zeta \,\mathrm{d}\hat r \,=\, c + \hat c r  + O(r^2)
  \label{expan_r_gamma}
 \end{eqnarray}
 is cone-smooth and has a sign as follows immediately from the expansions in Proposition~\ref{prop_cone_smooth} and term-by-term integration.
The constant $\hat c$ can be regarded as representing the, possibly $x^A$-dependent,  integration function. We conclude that the function
 \begin{eqnarray}
    \mathring \phi &=& \varepsilon e^{-\int_{\varepsilon}^r \hat r^{-1} (\hat\omega- 1)\,\mathrm{d}\hat r}\,(r\gamma)^{- 1}  =\frac{s_{i^-}}{\rho_{i^-}} + O(r)
    \label{expan_mathring_phi}
 \end{eqnarray}
 is cone-smooth and has a sign near the vertex of the cone,
\begin{eqnarray*}
  \mathrm{sign}(\mathring\phi_{i^-}) \,=\,  \mathrm{sign}(s_{i^-}) \mathrm{sign}(\rho_{i^-})
\end{eqnarray*}
(Note that
  there remains a gauge freedom to choose $\partial_1\mathring \phi|_{i^-}$.)
 \qed
\end{proof}

\end{document}